\documentclass[a4paper,onecolumn,11pt,accepted=2024-05-03]{quantumarticle}
\pdfoutput=1

\usepackage{amsfonts}
\usepackage{sidecap}
\usepackage{subcaption}
\usepackage{graphics}
\usepackage[braket, qm]{qcircuit}
\usepackage{enumerate}
\usepackage{hyperref}
\usepackage{ulem}
\usepackage{mathtools}
\usepackage{lipsum}
\mathtoolsset{showonlyrefs}
\usepackage{amsmath}
\usepackage{graphicx}
\usepackage{enumerate}
\usepackage{amssymb}
\usepackage{color}
\usepackage{amsthm}
\usepackage[braket]{qcircuit}
\usepackage{comment}
\usepackage{tikz}
\usepackage{float}
\usepackage{authblk}
\usepackage{multirow}
\usepackage{cite}
\usepackage[noend]{algpseudocode}
\usepackage[framemethod=tikz]{mdframed}
\usepackage{array, makecell, cellspace}
\usepackage[left=0.75in,right=0.75in,top=0.75in,bottom=0.75in]{geometry}

\newcommand{\appref}[1]{\hyperref[#1]{{Appendix~\ref*{#1}}}}
\newcommand{\be}{\begin{eqnarray} \begin{aligned}}
\newcommand{\ee}{\end{aligned} \end{eqnarray} }
\newcommand{\benn}{\begin{eqnarray*} \begin{aligned}}
\newcommand{\eenn}{\end{aligned} \end{eqnarray*}}
\newcommand*{\textfrac}[2]{{{#1}/{#2}}}

\newcommand*{\bbR}{\mathbb{R}}
\newcommand*{\bbC}{\mathbb{C}}

\newcommand*{\cB}{\mathcal{B}}

\newcommand*{\cE}{\mathcal{E}}

\newcommand*{\cG}{\mathcal{G}}
\newcommand*{\cH}{\mathcal{H}}

\newcommand*{\cM}{\mathcal{M}}

\newcommand*{\cD}{\mathcal{D}}

\newcommand*{\tr}{\mathop{\mathrm{tr}}\nolimits}
\newcommand*{\sbin}{\{0,1\}}

\newcommand{\bc}{\begin{center}}
\newcommand{\ec}{\end{center}}

\newtheorem{theorem}{Theorem}[section]
\newtheorem{lemma}[theorem]{Lemma}
\newtheorem{claim}[theorem]{Claim}
\newtheorem{definition}[theorem]{Definition}

\newtheorem{corollary}[theorem]{Corollary}

\usepackage{amsfonts}

\def\01{\{0,1\}}

\newcommand{\proj}[1]{|#1\rangle\langle#1|}

\newcommand*{\ExpE}{\mathbb{E}}

\newcommand*{\motimes}{\tilde{\otimes}}
\newcommand*{\Pf}{\mathsf{Pf}}

\newcommand*{\mult}{\mathsf{Mult}}
\newcommand*{\sbineven}{\{0,1\}_+}
\newcommand*{\cHeven}{\cH_+}
\newcommand*{\cGeven}{\cG^+}
\newcommand*{\stab}{\mathsf{STAB}}
\newcommand*{\parity}{\sigma}
\newcommand*{\cov}{\Gamma}
\newcommand*{\Gammaset}{E}

\DeclarePairedDelimiterX\lrangle[1]{\langle}{\rangle}{#1}

\newcommand*{\Adescribe}{\mathsf{describe}}
\newcommand*{\Aevolve}{\mathsf{evolve}}
\newcommand*{\Aoverlap}{\mathsf{overlap}}
\newcommand*{\Aoverlaptriple}{\mathsf{overlaptriple}}
\newcommand*{\Ameasure}{\mathsf{postmeasure}}
\newcommand*{\Apostmeasure}{\mathsf{postmeasure}}
\newcommand*{\desc}{\mathsf{Desc}}
\newcommand*{\Asupport}{\mathsf{findsupport}}
\newcommand*{\Arelate}{\mathsf{relatebasiselements}}
\newcommand*{\Aconvert}{\mathsf{convert}}
\newcommand*{\Ameasprob}{\mathsf{measureprob}}
\newcommand*{\Anorm}{\mathsf{norm}}
\newcommand*{\runtime}{\mathsf{time}}

\newcommand*{\Asampletheta}{\mathsf{samplestate}}
\newcommand*{\Anaivenorm}{\mathsf{naivenorm}}
\newcommand*{\Afastnorm}{\mathsf{fastnorm}}
\newcommand*{\Asamplestate}{\mathsf{samplestate}}
\newcommand*{\approxAmeasprob}{\mathsf{approx}\Ameasprob}
\newcommand*{\approxAevolve}{\mathsf{approx}\Aevolve}
\newcommand*{\approxApostmeasure}{\mathsf{approx}\Apostmeasure}

\begin{document}

\title{Classical simulation of non-Gaussian fermionic circuits}
\author[1,2]{Beatriz Dias}
\author[1,2]{Robert K{\"o}nig}
\affil[1]{Department of Mathematics, School of Computation, Information and Technology, Technical University of Munich, Garching, Germany}
\affil[2]{Munich Center for Quantum Science and Technology, Munich, Germany}
\maketitle

\begin{abstract}
We propose efficient algorithms for classically simulating fermionic linear optics operations applied to non-Gaussian initial states. By gadget constructions, this provides algorithms for fermionic linear optics with non-Gaussian operations. We argue that this  problem is analogous to that of simulating Clifford circuits with   non-stabilizer initial states: Algorithms for the latter problem immediately translate to the fermionic setting. Our construction is based on an extension of the covariance matrix formalism which permits to efficiently track relative phases in superpositions of Gaussian states. It yields simulation algorithms with polynomial complexity in the number of fermions, the desired accuracy, and certain quantities capturing the degree of non-Gaussianity of the initial state. We study one such quantity, the fermionic Gaussian extent, and show that it is multiplicative on tensor products 
when the so-called fermionic Gaussian fidelity is. We establish this property for the tensor product of two arbitrary pure states of four fermions with positive parity. 
\end{abstract}

\tableofcontents

\section{Introduction}

While universal polynomial-time quantum computation is believed to
exceed the capabilities of efficient classical algorithms, restricted  classes of quantum computations are amenable to efficient classical simulation. Identifying such models and corresponding simulation algorithms is a central goal in the study of quantum computing. On the one hand, a good characterization of the boundary between the computational power of classical and quantum computational models provides insight into potential quantum advantages. On the other hand, efficient classical simulation methods can be used to assess the merits and scalability of quantum information-processing proposals. For example, the resilience of certain quantum codes against restricted noise models has successfully been studied by means of classical simulation methods giving threshold estimates for large-scale systems, see e.g.,~\cite{DennisKitaevPreskill02,Katzgraberetal12,bravyiDisorderAssistedErrorCorrection2012,darmawanpoulin17,bravyienglbrechtetal,tuckettdarmawanetal19} for an incomplete list of relevant references.

\subsection{Efficiently simulable quantum computations\label{sec:efficientlysimulable}}

Most known examples of efficiently simulable quantum computations can be summarized by the following ingredients: 
\begin{enumerate}[(i)]
\item
A set~$\cD$ of states with the property that each element~$\Psi\in\cD$ has a succinct classical description~$d_\Psi$.
In the following, we will refer to~$\cD$ as a  dictionary. 
\item
A set~$\cE$ of operations (unitary or non-unitary evolutions), again with the property that each element~$E\in\cE$ has a succinct classical description~$d_E$. Following resource-theoretic conventions, we call~$\cE$ the set of free operations.
\item
A set~$\cM$ of measurements (quantum instruments) with an efficient classical descriptions~$d_M$ for each~$M\in\cM$, and the property that every post-measurement state (associated with different measurement outcomes) obtained by applying~$M\in\cM$ to a state~$\Psi\in\cD$ belongs to~$\cD$.
\end{enumerate}
A triple~$(\cD,\cE,\cM)$ gives rise to a (generally restricted) quantum computational model by composing these ingredients. 
A typical (non-adaptive) computation proceeds by preparing an initial state~$\Psi\in\cD$, applying a sequence~$\{E_t\}_{t=1}^T\subset\cE$ of operations, and performing measurements~$\{M_k\}_{k=1}^L\subset\cM$ in succession. Assuming for simplicity that~$\cE$ consists of a set of unitaries, and that for each~$k\in [L] = \{1, \dots, L\}$, the measurement~$M_k$ realizes a POVM~$M_k=\{M^{(k)}_m\}_{m\in \cM_k}$ with outcomes in a set~$\cM_k$, such a computation produces a sample~$m=(m_1,\ldots,m_L)\in \cM_1\times\cdots\times \cM_L$ from the distribution
\begin{align}
p(m_1,\ldots,m_L)&=\langle \Psi_T, M^{(L)}_{m_L}\cdots M^{(1)}_{m_1}\Psi_T\rangle\qquad\textrm{ where }\qquad \Psi_T= E_T \circ \cdots \circ E_1(\Psi)\ .\label{eq:distributionmeasurementresultscomputation}
\end{align}
More generally, one may consider circuits where operations are chosen adaptively depending on intermediate measurement results, assuming that the dependence is given by an efficiently computable function.

The task of classically simulating the computational model associated with~$(\cD,\cE,\cM)$ comes in two flavors. The input in both cases is the collection~$(d_\Psi,\{d_{E_t}\}_{t=1}^T,\{d_{M_k}\}_{k=1}^L)$ of descriptions of the initial state, the set of operations applied, and the measurements. The problem of weak simulation then consists in producing a sample~$m\in \cM_1\times\cdots\times\cM_L$ drawn according to the (ideal) output distribution~$p(m)$ of the circuit given by Eq.~\eqref{eq:distributionmeasurementresultscomputation}.  In contrast, the problem of strong simulation consists of computing the output probability~$p(m)$ for a given (potential) measurement outcome~$m$. 

Relaxing the requirements of weak and strong simulation, one may allow for an approximation error. For weak simulation, this is typically formalized by demanding that the (probabilistic) classical algorithm outputs a sample~$m$ drawn from a distribution~$\tilde{p}$ which is~$\delta$-close in~$L^1$-distance (for a chosen error parameter~$\delta>0$) to the ideal output distribution~$p$. Similarly, for strong simulation, the output~$\tilde{p}$ is required to be close to the value~$p(m)$ with a controlled (additive or multiplicative) error.

In cases where a computational model specified by~$(\cD,\cE,\cM)$ is amenable to efficient classical simulation, associated classical simulation algorithms are typically constructed by considering evolution and measurement separately. 
The basic problem then consists in constructing efficient classical algorithms with the following functionalities:
\begin{enumerate}[(a)]
\item\label{it:simulationalgorihma}
An algorithm~$\Aevolve$ which, given classical descriptions~$d_\Psi$ of a state~$\Psi\in\cD$ and~$d_E$ of an evolution operation~$E\in\cE$, computes a classical  description~$d_{E(\Psi)}$ of  the evolved state~$E(\Psi)$.
\item\label{it:simulationalgorithmb} Given a classical description~$d_\Psi$ of a state~$\Psi\in\cD$, a classical description~$d_M$ of a measurement~$M\in\cM$ (with associated set of measurement outcomes~$\cM_M$) and a measurement outcome~$m\in\cM_M$, 
    \begin{enumerate}[($\text{b}$1)]
        \item \label{it:simulationalgorithmb1} an algorithm~$\Ameasprob$ which outputs the probability~$p(m)$ (determined by Born's rule) of obtaining measurement outcome~$m$, and
        \item \label{it:simulationalgorithmb2} an algorithm~$\Ameasure$ which outputs a classical description of the post-measurement state
    \end{enumerate}
    when applying the measurement~$M$ to~$\Psi$.
\end{enumerate}
It is clear that a triple~$(\Aevolve,\Ameasprob,\Ameasure)$ of such algorithms immediately gives rise to an efficient algorithm for strong simulation of the model~$(\cD,\cE,\cM)$, with a runtime
\begin{align}
    T\cdot \mathsf{time}(\Aevolve)+L\cdot (\mathsf{time}(\Ameasprob)+\mathsf{time}(\Apostmeasure))
\end{align} which is linear in the number~$T$ of operations applied, and linear in the number~$L$ of measurements. Assuming that for any measurement~$M\in\cM$, the set of measurement outcomes~$\cM_M$ associated with~$M$ is of constant cardinality, the 
 triple~$(\Aevolve,\Ameasprob, \allowbreak\Ameasure)$ also gives rise to a randomized algorithm for weak simulation:  Such an algorithm is obtained by using~$\Ameasprob$ to compute the entire distribution~$\{p(m)\}_{m\in \cM_M}$ of measurement outcomes (when applying a measurement~$M$), and then drawing~$m\in\cM_M$ randomly according to this distribution.  The runtime of this probabilistic algorithm is
 \begin{align}
     T\cdot \mathsf{time}(\Aevolve)+L\cdot (\mathsf{time}(\Ameasprob)\cdot w+\mathsf{time}(\Apostmeasure))
 \end{align} where~$w=\max_{M\in\cM} |\cM_M|$ bounds the maximal cardinality of the set of measurement outcomes. 

\subsubsection{Clifford circuits / Stabilizer computations}
Perhaps the most well-known example of  a computational model~$(\cD,\cE,\cM)$ where efficient algorithms~$(\Aevolve,\Ameasprob,\Ameasure)$  can be provided is the Gottesman-Knill-theorem for stabilizer computations on~$n$~qubits. Here~$\cD$ is the set~$\stab_n$ of~$n$-qubit stabilizer states (whose elements can be specified by their stabilizer generators, i.e., corresponding stabilizer tableaux), $\cE$ is the set of Clifford unitaries (described by symplectic matrices), and~$\cM$ are measurements of single-qubit Pauli~$Z$ operators (described by an index~$j\in[n]$).
In this case,  there are efficient algorithms with runtimes given in Table~\ref{tab:algruntimestabilizer}.

\begin{table}[h]
\centering
\setlength\extrarowheight{-2pt}
\renewcommand\arraystretch{1.7}
\setlength\tabcolsep{15pt}
\vspace{3mm}
\begin{tabular}{c|c}
\hline
algorithm    & $\mathsf{time}$\\
\hline\hline
$\Aevolve$ & $O(n)$\\ \hline 
$\Ameasprob$ & $O(n)$ \\ \hline 
$\Apostmeasure$ & $O(n^2)$ \\
\hline
\end{tabular}
\caption{Runtimes of building blocks~$\Aevolve$, $\Ameasprob$, $\Ameasure$ for classical simulation of~$n$-qubit stabilizer circuits as given in~\cite{aaronsonImprovedSimulationStabilizer2004}.
Evolution corresponds to application of an~$n$-qubit Clifford unitary, and each measurement is that of a Pauli observable~$Z_j$ with~$j\in[n]$. \label{tab:algruntimestabilizer}}
\end{table}

\subsubsection{Fermionic linear optics / Fermionic Gaussian computations\label{sec:fermioniclinearoptics}}
A different class of efficiently simulable computations -- the one we are interested in here -- is that of fermionic linear optics on~$n$ fermions. We focus on pure-state computations: Here the dictionary~$\cD$ consists of the set~$\cG_n$ of pure fermionic Gaussian states. An element~$\Psi\in\cG_n$ in the dictionary can be described by its covariance matrix~$\cov_\Psi$, an antisymmetric~$2n\times 2n$~matrix with real entries. The set~$\cE=\cE_{\textrm{Gauss}}$ can be taken as the set of Gaussian unitary operations. Each such unitary~$U=U_R$ is fully determined by an  element~$R\in O(2n)$ of the orthogonal group on~$\mathbb{R}^{2n}$, where~$R\mapsto U_R$ defines a (projective) unitary representation of~$O(2n)$ on the space~$\cH^n$ of~$n$~fermions. The set~$\cM=\cM_{\textrm{number}}$ consists of all occupation number measurements. As in the case of stabilizer states, there are polynomial-time algorithms~$(\Aevolve,\Ameasprob,\Ameasure)$
 for classical simulation with runtimes summarized in Table~\ref{tab:fermionicruntimesimulation}. In particular, the covariance matrix~$\cov_{U_R\Psi}$ of a Gaussian state~$\Psi$ evolved under a Gaussian unitary~$U_R$ can be computed in time~$O(n^3)$ from~$\cov_\Psi$ and~$R$. 
The outcome probability of observing~$0$ (respectively~$1$) when performing an occupation number measurement can be computed in  time~$O(1)$, and the covariance matrix of the post-measurement state can be computed in time~$O(n^2)$ \cite{10.1145/380752.380785,PhysRevA.65.032325,knill2001fermionic} (see also~\cite{bravyiDisorderAssistedErrorCorrection2012}).

\begin{table}[h]
\centering
\setlength\extrarowheight{-2pt}
\renewcommand\arraystretch{1.7}
\setlength\tabcolsep{15pt}
\vspace{3mm}
\begin{tabular}{c|c}
\hline
algorithm    & $\mathsf{time}$\\
\hline\hline
$\Aevolve$ & $O(n^3)$\\ \hline
$\Ameasprob$ & $O(1)$\\ \hline
$\Apostmeasure$ & $O(n^2)$ \\
\hline
\end{tabular}
\caption{Runtimes of building blocks~$\Aevolve$, $\Ameasprob$, $\Ameasure$ for classical simulation of~$n$-fermion linear optics circuits
as proposed in~\cite{10.1145/380752.380785,PhysRevA.65.032325,knill2001fermionic}, see also~\cite{bravyiDisorderAssistedErrorCorrection2012}.
Evolution amounts to application of a fermionic Gaussian unitary. Measurement corresponds to measuring an observable~$a_j^\dagger a_j$ (occupation number) for~$j\in [n]$.
\label{tab:fermionicruntimesimulation}}
\end{table}

\subsection{Classical simulation algorithms and measures of magic} 

A natural way of extending the power of a quantum computational model specified by~$(\cD,\cE,\cM)$
consists in providing resources/capabilities  that do not belong to the specified sets. ``Magic states'' are a prime example: Here a state~$\Psi\not\in\cD$ not belonging to the dictionary is provided as an (initial) state  in the quantum computation, thereby providing additional capabilities to the computational model. For example, non-Clifford unitaries can be realized by certain stabilizer-computations (sometimes referred to as ``gadgets'') applied to so-called magic states~\cite{PhysRevA.71.022316}. Similarly, non-Gaussian initial states can be combined with fermionic linear optics operations  to realize non-Gaussian operations~\cite{PhysRevA.73.042313,PhysRevLett.123.080503}. While such a magic state  can even promote the computational model to universal quantum computation, this is generally not the case for all states~$\Psi$. It is thus a natural question to quantify the degree of ``magicness'' provided by a state~$\Psi\not\in\cD$. For the set~$\stab_n$ of~$n$-qubit stabilizer states, corresponding magic monotones considered in the literature include the robustness of magic~\cite{PhysRevLett.118.090501,Heinrich2019robustnessofmagic}, the exact and approximate stabilizer rank~
\cite{bravyiTradingClassicalQuantum2016a,bravyiImprovedClassicalSimulation2016,bravyiSimulationQuantumCircuits2019a}, the stabilizer extent~\cite{bravyiSimulationQuantumCircuits2019a,heimendahlStabilizerExtentNot2021}, the stabilizer nullity~\cite{Beverland_2020}, the generalized robustness~\cite{PRXQuantum.2.010345} and the magic entropy~\cite{bu2023stabilizer}. 

The maximum overlap of a given state~$\Psi$ with an element of the dictionary~$\cD$, i.e., the quantity
\begin{align}
    F_{\cD}(\Psi)&=\sup_{\varphi\in\cD} |\langle \varphi,\Psi\rangle|^2\ ,
    \label{eq:defmaxoverlap}
\end{align}
is arguably one of the most direct ways of quantifying how far~$\Psi$ is from a ``free'' state, i.e., a state belonging to~$\cD$. Motivated by the analogously defined notion of stabilizer fidelity in Ref.~\cite{bravyiSimulationQuantumCircuits2019a}, we call~$F_{\cD}(\Psi)$ the~$\cD$-fidelity of~$\Psi$ in the following. This quantity plays an important role in our arguments when considering multiplicativity properties. However, the~$\cD$-fidelity~$F_{\cD}(\Psi)$ is  not a good quantifier of hardness of classical simulation because simply replacing~$\Psi$ by an element of~$\cD$ typically leads to a significant approximation error.

From the point of view of classical simulation, a relevant magicness measure for a state~$\Psi\not\in\cD$ relates to the (added) complexity when trying to simulate a quantum computation with initial state~$\Psi$, built from a triple~$(\cD,\cE,\cM)$ allowing for efficient classical simulation.
One such measure, introduced in Ref.~\cite{bravyiTradingClassicalQuantum2016a} for the case of stabilizer computations, is the~$\cD$-rank~$\chi_\cD(\Psi)$ of~$\Psi$. (For~$\cD=\stab_n$, this is called the stabilizer rank of~$\Psi$.) It is defined as the minimum number of terms when decomposing~$\Psi$ as a linear combination of elements of~$\cD$, i.e.,
\begin{align}
\chi_\cD(\Psi)=\min\left\{\chi\in \mathbb{N}\ |\ \exists \{\varphi_j\}_{j=1}^\chi\subset\cD, \{\gamma_j\}_{j=1}^\chi\subset\mathbb{C}\textrm{ such that }\Psi=\sum_{j=1}^\chi \gamma_j\varphi_j\ \right\}\ .\label{eq:drankdefinition}
\end{align}
In the context of signal processing, the corresponding optimization problem is  referred to as a sparse approximation problem. The~$\cD$-rank~$\chi_\cD(\Psi)$ appears naturally when constructing and analyzing simulation algorithms, but it suffers from a number of shortcomings: On the one hand, the set of states~$\Psi\in\cH$
whose~$\cD$-rank is less than the dimension of the Hilbert space~$\cH$ is a set of zero Lebesgue measure~\cite[Proposition 4.1]{troppone}. On the other hand,
the quantity~$\chi_\cD(\Psi)$ relates to the classical simulation cost of exactly simulating dynamics involving the state~$\Psi$. In practice, some approximation error is typically acceptable, and corresponding simulations can be achieved with lower cost. In other words, the quantity~$\chi_\cD(\Psi)$ does not accurately reflect the cost of approximate simulation.

A more operationally relevant quantity is the~$\delta$-approximate~$\cD$-rank~$\chi_\cD^\delta(\Psi)$ of~$\Psi$ introduced in Ref.~\cite{bravyiImprovedClassicalSimulation2016}, again for stabilizer computations. For a fixed approximation error~$\delta>0$, this is given by the minimum~$\cD$-rank of any state~$\Psi'$ that is~$\delta$-close to~$\Psi$, i.e., 
\begin{align}
    \chi_\cD^\delta(\Psi) = \min \left\{ \chi_{\cD}(\Psi') \ | \ \Psi'\in\cH \text{ such that } \|\Psi - \Psi'\| \leq \delta  \right\} \ .\label{eq:approximaterankdefinition}
\end{align}
An exact classical simulation algorithm whose complexity scales with the exact~$\cD$-rank $\chi_\cD(\Psi)$ provides an approximate simulation at  a cost with an identical scaling  in the approximate (instead of exact)~$\cD$-rank~$\chi_\cD^\delta(\Psi)$ of~$\Psi$. Here approximate weak simulation means that instead of 
sampling from the ideal output distribution~$P$ of a circuit, the simulation samples from a distribution~$P'$ whose~$L^1$-distance from~$P$ is bounded by~$O(\delta)$. Similarly, in approximate (strong) simulation, output probabilities are approximately computed with a controlled approximation error.

A different quantity of interest is obtained by replacing the ill-behaved rank function (i.e., size of the support) in the definition of the~$\cD$-rank~$\chi_\cD(\Psi)$ by the~$L^1$-norm of the coefficients when representing~$\Psi$ as a linear combination. In the context of stabilizer states the corresponding quantity was introduced by Bravyi et al.~\cite{bravyiSimulationQuantumCircuits2019a} under the term stabilizer extent: For a state~$\Psi\in(\mathbb{C}^2)^{\otimes n}$ it is defined as
\begin{align}
    \xi_{\stab_n}(\Psi)&=\inf\left\{
    \|\gamma\|_1^2\ |\ \gamma:\stab_n\rightarrow\mathbb{C}\textrm{ such that }\Psi=\sum_{\varphi\in\stab_n}\gamma(\varphi)\varphi\right\}\ ,
\end{align}
where~$\|\gamma\|_1=\sum_{\varphi\in\stab_n}|\gamma(\varphi)|$ denotes the~$1$-norm of~$\gamma$. The corresponding convex optimization problem is known as the basis pursuit problem~\cite{chenetal} (when~$\stab_n$ is replaced by e.g., a finite dictionary~$\cD$). Sufficient conditions for when the basis pursuit problem yields a solution of the sparse approximation problem where investigated in a series of works culminating in Fuchs' condition~\cite{fuchs} (see also~\cite{troppRecoveryShortComplex2005}). More importantly for (approximate) simulation, feasible solutions of the basis pursuit problem provide upper bounds on the sparse approximation problem.  For the stabilizer rank, a sparsification result  (see~\cite[Theorem~1]{bravyiSimulationQuantumCircuits2019a})  gives an upper bound on the~$\delta$-approximate stabilizer rank~$\chi_{\stab_n}(\Psi)$ in terms of the stabilizer extent~$\xi_{\stab_n}(\Psi)$, for any~$\delta>0$ (see  Section~\ref{sec:sparsification}).

Building on earlier results~\cite{bravyiImprovedClassicalSimulation2016}, it was shown in Ref.~\cite{bravyiSimulationQuantumCircuits2019a}
that a stabilizer circuit on~$n$~qubits with~$L$~Clifford gates initialized in a state~$\Psi$ can be weakly simulated with error~$\delta$ in a time scaling as~$O(\textfrac{\xi_{\stab_n}(\Psi)}{\delta^2}\cdot \mathsf{poly}(n,L))$. The error~$\delta$ expresses the~$L^1$-norm distance of the distribution of simulated measurement outcomes from the output distribution of the actual quantum computation. Here we are not accounting for the time required to perform classical computations when adaptive quantum circuits are considered. 
In addition, Ref.~\cite{PRXQuantum.3.020361} provided a classical algorithm for strong simulation of a circuit~$U$ with~$L$~Clifford gates and~$t$~$T$-gates initialized in a stabilizer state~$\Psi$ with an additive error~$\delta$.
Their algorithm outputs an estimate of the probability~$|\lrangle{x, U \Psi}|^2$ of obtaining measurement outcome~$x\in \{0,1\}^n$ up to an additive error~$\delta$, with success probability greater than~$1-p_f$. It has  runtime~$O(\xi_{\stab_n}(\ket{T}^{\otimes t})  \log(1/p_f) \cdot \mathsf{poly}(n, L, \delta^{-1}))$,
scaling linearly with the stabilizer extent~$\xi_{\stab_n}(\ket{T}^{\otimes t})$ of~$t$ copies of the single-qubit magic state~$\ket{T}$ associated with a~$T$-gate~\cite{PhysRevA.71.022316}.

\subsection{The fermionic Gaussian extent}

In the following, we generalize the notion of the extent beyond stabilizer computations to any dictionary~$\cD$. We  refer to the corresponding quantity as the~$\cD$-extent~$\xi_\cD(\Psi)$ of~$\Psi$.  We assume throughout that we are interested in pure state quantum computations on a Hilbert space~$\cH$, and that the dictionary~$\cD$ is a subset of pure states on~$\cH$. Then the~$\cD$-extent~$\xi_{\cD}(\Psi)$ of~$\Psi\in\cH$ is defined as
\begin{align}
\xi_{\cD}(\Psi)&=\inf_{N\in \mathbb{N}}\inf_{\varphi_1,\ldots,\varphi_N\in \cD} \left\{\|\gamma\|^2_1\ \big|\ \gamma\in\mathbb{C}^N\textrm{ such that }\sum_{j=1}^N \gamma_j\varphi_j=\Psi\right\}\  .\label{eq:Loneminimization}
\end{align}
Here~$\|\gamma\|_1=\sum_{j=1}^N |\gamma_j|$ is the~$L^1$-norm of the vector~$\gamma$. That is, the~$\cD$-extent~$\xi_{\cD}(\Psi)$ is the~$L^1$-norm of the coefficients minimized over all decompositions of~$\Psi$ into a finite linear combination of elements of the dictionary~$\cD$. As  mentioned above, quantities of the form~\eqref{eq:Loneminimization} are well-studied in the context of signal-processing.

When the dictionary~$\cD$ is a finite subset of a Hilbert space~$\cH\cong\mathbb{C}^d$, the~$\cD$-extent~$\xi_\cD(\Psi)$ of a state~$\Psi\in\cH$ can be expressed as a second-order cone program~\cite{alizadehSecondorderConeProgramming2003} (see also e.g.,  \cite{boyd2004convex}), as in Appendix~A of Ref.~\cite{heimendahlStabilizerExtentNot2021}. Second-order cone programs can be solved in time polynomial in~$\max(d, |\cD|)$. We are typically interested in cases where
$\cD$ contains a basis of~$\cH$ (such that every state can indeed be represented as a linear combination of dictionary elements): Here this runtime is at least polynomial in~$d$. For example, in the case~$\cD=\stab_n$ of stabilizer states on~$n$~qubits, this leads to an exponential scaling in~$n^2$. Beyond algorithmic considerations related to the evaluation of the extent, the fact that~$\xi_\cD(\Psi)$ is given by a second-order cone program  provides useful analytical insight by convex programming duality. Indeed, this fact has previously been exploited both for showing multiplicativity of the stabilizer extent for states of small dimension~\cite{bravyiSimulationQuantumCircuits2019a}, as well as to show non-multiplicativity in high dimensions~\cite{heimendahlStabilizerExtentNot2021}. In Section~\ref{sec:multoverlapimpmultextent}, we also exploit this connection to relate the~$\cD$-fidelity~$F_{\cD}(\Psi)$ with the~$\cD$-extent~$\xi_\cD(\Psi)$.

In contrast, the~$\cD$-extent~$\xi_\cD(\Psi)$  for an infinite, i.e., continuously parameterized, dictionary~$\cD$ constitutes additional mathematical challenges as an optimization problem. This is the case of interest here as we are considering the dictionary~$\cD=\cG_n$ consisting of all~$n$-fermion Gaussian states in the following. We call the associated quantity~$\xi_{\cG_n}(\Psi)$ the (fermionic) Gaussian extent of an~$n$-fermion state~$\Psi$. Our focus here is on discussing the role of the quantity~$\xi_{\cG_n}(\Psi)$ in the context of classically simulating fermionic linear optics, and its behavior on tensor products. A detailed discussion of the algorithmic problem of computing~$\xi_{\cG_n}(\Psi)$ for an arbitrary state~$\Psi$, and finding a corresponding optimal decomposition of~$\Psi$ into a linear combination of Gaussian states is beyond the scope of this work. We refer to e.g.,~\cite{Chandrasekaranetal} where semidefinite relaxations are given for the related atomic norm minimization problem in cases where the atomic set (corresponding to the dictionary) has algebraic structure. Similar techniques may be applicable to the fermionic Gaussian extent.

\subsection{On the (sub)multiplicativity of the extent}

Consider a situation where an operation~$E\not\in \cE$
not belonging to the set~$\cE$ of efficiently simulable operations
is implemented by using a ``magic'' resource state~$\Psi\not\in\cD$. For example, if~$\cD=\stab_n$ is the set of stabilizer states, $\cE$ the set of Clifford unitaries and~$\cM$ the set of single-qubit Pauli-$Z$-measurements, then a non-Clifford gate (such as the~$T$-gate) can be realized by an (adaptive) Clifford circuit at the cost of consuming a non-Clifford state (such as the state~$\ket{T}$)~\cite{PhysRevA.71.022316}. Similar ``gadget constructions'' exist for fermionic linear optics, where non-Gaussian unitaries are realized by Gaussian unitaries and non-Gaussian states~\cite{PhysRevA.73.042313,PhysRevLett.123.080503}. A natural question arising  in this situation is to characterize the cost of simulating the application of two  independent magic gates~$E_1,E_2\not\in\cE$, each realized by efficiently simulable operations (belonging to~$\cE$) using magic states~$\Psi_1,\Psi_2$. For any reasonable simulation algorithm, we expect the required simulation effort to increase at most multiplicatively. Indeed, this feature is reflected in the submultiplicativity property
\begin{align}
    \xi_{\cD_{3}}(\Psi_1\otimes\Psi_2) & \leq \xi_{\cD_{1}}(\Psi_1)\xi_{\cD_{2}}(\Psi_2)\qquad\textrm{ for all }\qquad \Psi_1\in \cH_1\textrm{ and }\Psi_2\in\cH_2\ \label{eq:submultiplicativity}
\end{align}
of the~$\cD$-extent. In Eq.~\eqref{eq:submultiplicativity}, we are considering  Hilbert spaces~$\cH_1$, $\cH_2$ and their tensor product~$\cH_3=\cH_1\otimes\cH_2$, as well as dictionaries~$\cD_j\subset\cH_j$ for~$j\in[3]$. The submultiplicativity property~\eqref{eq:submultiplicativity}  follows immediately from the definition of the extent if the three dictionaries satisfy the inclusion property
\begin{align}
    \cD_1\otimes\cD_2\subset \cD_3\ .\label{eq:cdtwoz}
\end{align}
In particular, this is satisfied e.g., when the dictionary~$\cD_j=\stab_{n_j}\subset (\mathbb{C}^2)^{\otimes n_j}$ is the set of~$n_j$-qubit stabilizer states for~$j\in [3]$, with~$n_3=n_1+n_2$, or when considering the set of (even) Gaussian states (see below).

While the submultiplicativity property~\eqref{eq:submultiplicativity} is a trivial consequence of Eq.~\eqref{eq:cdtwoz}, the question of whether or not the stronger multiplicativity property
\begin{align}
\xi_{\cD_{3}}(\Psi_1\otimes\Psi_2) & = \xi_{\cD_{1}}(\Psi_1)\xi_{\cD_{2}}(\Psi_2)\qquad\textrm{ for all }\qquad \Psi_1\in \cH_1\textrm{ and }\Psi_2\in\cH_2\ \label{eq:multiplicativityclaim}
\end{align}
holds for the~$\cD$-extent is a much less trivial problem. If the multiplicativity property~\eqref{eq:multiplicativityclaim} is satisfied,  then computing the extent of a product state can be broken down into several smaller optimization problems: It suffices to compute the extent of each factor in the tensor product. Furthermore,  the classical simulation cost (with typical algorithms) 
when applying several non-free (``magic") gates constructed by gadgets increases at an exponential rate determined by the individual gates. In contrast, if the extent is not multiplicative (i.e., the equality in~\eqref{eq:multiplicativityclaim} is not satisfied for some states~$\Psi_j\in \cH_j$, $j\in[2]$), then such a simplification is not possible. More surprisingly, such a violation of multiplicativity implies that the classical simulation cost of applying certain non-free gates can be reduced by treating these jointly instead of individually. We note that in the slightly different context of so-called circuit knitting, similar savings in complexity have been shown to be significant~\cite{piveteausutter}.

Previous work established that the stabilizer extent is multiplicative even for multiple factors, that is, 
\begin{align}
\xi_{\stab_{n_1+\cdots +n_r}}(\Psi_1\otimes \cdots \otimes \Psi_r)&=\prod_{j=1}^r \xi_{\stab_{n_j}}(\Psi_j)\qquad\textrm{ for all }\qquad \Psi_j\in (\mathbb{C}^2)^{\otimes n_j}, j\in [r]
\end{align}
if the factors are single-qubit, $2$- or~$3$-qubit states, i.e., $n_j\in[3]$, see Ref.~\cite{bravyiSimulationQuantumCircuits2019a}. An example is the stabilizer extent of a tensor product of~$t$ copies of the magic (single-qubit) state~$\ket{T} = (\ket{0} + e^{i \pi/4} \ket{1})/\sqrt{2}$  associated with the~$T$-gate. Multiplicativity for qubit states gives~$\xi_{\stab_1}(\ket{T}^{\otimes t}) = \xi_{\stab_t}(\ket{T})^t$, where~$\xi_{\stab_1}(\ket{T})$ is known to be approximately~$1.17$ \cite{bravyiImprovedClassicalSimulation2016}. This translates to an overhead exponential in~$t$ in the runtime of stabilizer computations supplemented with~$t$~$T$-gates. 
Surprisingly, the stabilizer extent has been shown not to be multiplicative (for all pairs of  states) in high dimensions~\cite{heimendahlStabilizerExtentNot2021}.

For (pure) Gaussian states, the Gaussian extent of a~$1$-, $2$- and~$3$-mode pure fermionic state is trivially one because any~$1$-, $2$- and~$3$-mode pure fermionic state is Gaussian \cite{meloPowerNoisyFermionic2013a} and is thus an element of the dictionary. Hence the Gaussian extent is (trivially) multiplicative if the factors are~$1$-, $2$- or~$3$-mode fermionic  states. The simplest non-trivial case is that of~$n=4$ fermionic modes in each factor.

\subsection{Our contribution\label{sec:ourcontribution}}

Our results concern fermionic linear optics, the computational model introduced in  Section~\ref{sec:fermioniclinearoptics}
described by the triple~$(\cG_n,\cE_{\textrm{Gauss}},\cM_{\textrm{number}})$ of fermionic Gaussian pure states on~$n$ fermions, Gaussian unitary operations and number state measurements. We propose classical simulation algorithms for the case where the initial state~$\Psi\in\cH_n$ is an arbitrary pure state in the~$n$-fermion Hilbert space~$\cH_n$ (instead of belonging to the set~$\cG_n\subset\cH_n$ of Gaussian states). Our results are two-fold:

{\bf New simulation algorithms. } We give algorithms  realizing the functionalities described in Section~\ref{sec:efficientlysimulable} exactly for the triple~$(\cH_n,\cE_{\textrm{Gauss}},\cM_{\textrm{number}})$.
This immediately gives rise to efficient algorithms for  weak and strong simulation of circuits with non-Gaussian initial states. The corresponding runtimes of these building blocks, which we refer to as~$(\chi\Aevolve, \chi\Ameasprob,\allowbreak \chi\Apostmeasure)$,  depend on the Gaussian rank~$\chi=\chi_{\cG_n}(\Psi)$ of the initial state~$\Psi$ and are summarized in Table~\ref{tab:runtimeextendedalgo}.

\begin{table}[h]
\centering
\setlength\extrarowheight{-2pt}
\renewcommand\arraystretch{1.7}
\setlength\tabcolsep{15pt}
\vspace{3mm}
\begin{tabular}{c|c} 
\hline
algorithm    & $\mathsf{time}$\\
\hline\hline
$\chi\Aevolve$ & $O(\chi n^3)$\\ \hline 
$\chi\Ameasprob$ & $O(\chi^2 n^3)$ \\ \hline 
$\chi\Apostmeasure$ & $O(\chi n^3)$ \\
\hline
\end{tabular}
\caption{Runtimes of the building blocks~$\chi\Aevolve$, $\chi\Ameasprob$, $\chi\Ameasure$ for exact  simulation of~$n$-qubit fermionic linear optics circuits with a non-Gaussian initial state~$\Psi$ of Gaussian rank~$\chi=\chi_{\cG_n}(\Psi)$. 
Evolution corresponds to the application of a Gaussian unitary from a set~$\cE_{\textrm{Gauss}}$ of generators (specified below), and the set of measurements is given by occupation number measurements on each of the modes.\label{tab:runtimeextendedalgo}}
\end{table}

Key to the construction of these algorithms is a novel way of keeping track of relative phases  in superpositions of Gaussian states, see Section~\ref{sec:trackingphases}. We argue that our techniques can be applied more generally to adapt simulation procedures developed, e.g., for Clifford circuits, to the setting of fermionic linear optics. 
In order to illustrate this procedure, we apply it to the simulation algorithms  of~\cite{bravyiImprovedClassicalSimulation2016,bravyiSimulationQuantumCircuits2019a} for efficient (approximate) classical simulation algorithms. In this way, we obtain new approximate simulation algorithms with runtimes depending linearly on the fermionic Gaussian extent~$\xi=\xi_{\cG_n}(\Psi)$ of the initial state~$\Psi$, see Table~\ref{tab:approximateruntimes} for a summary of the corresponding runtimes.
\begin{table}[h]
\centering
\setlength\extrarowheight{-2pt}
\renewcommand\arraystretch{1.7}
\setlength\tabcolsep{15pt}
\vspace{3mm}
\begin{tabular}{c|c}
\hline
algorithm    & $\mathsf{time}$\\
\hline\hline
$\approxAevolve$ & $O(\xi \delta^{-2}n^3)$\\ \hline 
$\approxAmeasprob$ & $O(\xi \delta^{-2}\epsilon^{-2}p_f^{-1} n^{7/2})$ \\ \hline 
$\approxApostmeasure$ & $O(\xi \delta^{-2}n^3)$\\
\hline
\end{tabular}
\caption{Runtimes of building blocks~$\approxAevolve$, $\approxAmeasprob$, $\approxApostmeasure$ for approximate  simulation of~$n$-qubit fermionic linear optics circuits with a non-Gaussian initial state~$\Psi$ of Gaussian extent~$\xi=\xi_{\cG_n}(\Psi)$. 
The parameters~$(\epsilon,\delta,p_f)$ determine the quality of the approximation.
\label{tab:approximateruntimes}}
\end{table}
They depend inverse-polynomially on parameters~$(\delta,\epsilon,p_f)$ determining the accuracy of the simulation. 
The error~$\delta$ describes a  certain ``offset'', i.e., a systematic error: Instead of simulating the dynamics of the circuit with the (ideal) initial state~$\Psi$, the simulation algorithm emulates the dynamics when using a different starting state~$\tilde{\Psi}$  which is~$\delta$-close to~$\Psi$, i.e., which satisfies~$\|\Psi-\tilde{\Psi}\|\leq \delta$. The algorithm~$\approxAevolve$ computes evolution exactly on the state used in the simulation (i.e., it preserves the approximation error~$\delta$ relative to the ideal  initial state). In contrast, the procedure~$\approxAmeasprob$ 
can fail with probability~$p_f$, and both~$\approxAmeasprob$ and~$\approxApostmeasure$ introduce an additional error quantified by~$\epsilon$ (if~$\approxAmeasprob$ succeeds): Instead of returning the probability~$p(0)$ of obtaining zero occupation number when measuring the state,  the output  of~$\approxAmeasprob$ is a value~$\tilde{p}$ which satisfies~$|\tilde{p}-p(0)|\leq O(\epsilon)$.  Similarly, the output of~$\approxApostmeasure$
is a description of a  state that is~$O(\epsilon)$-close to the actual post-measurement state.
These parameters and runtimes are analogous to those obtained in~\cite{bravyiSimulationQuantumCircuits2019a} for simulating Clifford circuits with non-stabilizer initial states. In particular, they imply that a circuit with initial state~$\Psi$ involving~$T$~Gaussian unitaries and~$L$~occupation number measurements can be weakly simulated 
in time~$\tilde{O}(\epsilon^{-2}\xi)$, such that the sampled measurement outcomes are~$\epsilon$-close in~$L^1$-distance to the ideal (joint) output distribution of all measurements.  Here the notation~$\tilde{O}(\cdot)$ suppresses a factor polynomial in~$n,T,L$ and~$\log(\epsilon^{-1})$, see~\cite{bravyiImprovedClassicalSimulation2016} for details.  

{\bf On the multiplicativity of the Gaussian extent and the Gaussian fidelity.}
Motivated by the relevance of the Gaussian extent~$\xi_{\cG_n}(\Psi)$ for characterizing the complexity of classical simulation, we study multiplicativity properties of both the~$\cD$-fidelity~$F_{\cD}(\Psi)$ as we well as the~$\cD$-extent~$\xi_{\cD}(\Psi)$ for a general infinite, i.e., continuously parameterized, dictionary~$\cD$.  We show that multiplicativity of the~$\cD$-fidelity is closely related to that of the~$\cD$-extent: For a general family of (discrete or continuous) dictionaries~$\cD_{j}\subset \cH_{j}$ for~$j\in[3]$ with the property 
\begin{align}
    \cD_{1} \otimes \cD_{2} \subset \cD_{3} \ ,
\end{align}
multiplicativity of the~$\cD$-fidelity, i.e.,
\begin{align}
    F_{\mathcal{D}_3}\left(\Psi_1 \otimes \Psi_2\right)=F_{\mathcal{D}_1}\left(\Psi_1\right) F_{\mathcal{D}_2}\left(\Psi_2\right) \quad \text { for all } \quad \Psi_j \in \mathcal{H}_j \text { for } j\in[2]
\end{align}
implies multiplicativity of the~$\cD$-extent, i.e.
\begin{align}
    \xi_{\mathcal{D}_3}\left(\Psi_1 \otimes \Psi_2\right)=\xi_{\mathcal{D}_1}\left(\Psi_1\right) \xi_{\mathcal{D}_2}\left(\Psi_2\right) \quad \text { for all } \quad \Psi_j \in \mathcal{H}_j \text { for } j\in[2] \ .
\end{align}
We note that for stabilizer states~$\cD=\stab_n$, 
a similar 
route was followed  in Ref.~\cite{bravyiSimulationQuantumCircuits2019a} to show multiplicativity of the stabilizer extent~$\xi_{\stab_n}$ with respect to the tensor product of~$1$-, $2$- and~$3$-qubit states. Our main contribution is an extension of this  connection to the case of infinite dictionaries by the use of nets. We expect this connection to be helpful in proving or disproving multiplicativity of the extent more generally.

We subsequently make use of this connection to the Gaussian fidelity to show that the fermionic Gaussian extent is multiplicative for the tensor product of any two 4-mode fermionic states with positive parity, i.e., 
\begin{align}
    \xi_{\cG_8} (\Psi_1 \otimes \Psi_2) = 
    \xi_{\cG_4} (\Psi_1)
    \xi_{\cG_4} (\Psi_2) 
    \qquad \textrm{ for all }\qquad
    \Psi_1, \Psi_2 \in \cH^4_+
    \  .\label{eq:tensorproductfidelityadditivitym}
\end{align}
Here~$\cH^4_+$ denotes the set of~$4$-mode fermionic states with positive parity. The proof of~\eqref{eq:tensorproductfidelityadditivitym} relies on the Schmidt decomposition of fermionic Gaussian states and specific properties of~$4$-mode (positive parity) fermionic states.

The result~\eqref{eq:tensorproductfidelityadditivitym} gives the first non-trivial example of multiplicativity of the Gaussian extent. Multiplicativity for more general cases such as that of multiple~$4$-mode fermionic factors remains an open problem.

\subsection{Prior and related work} 
The starting point of our work is the fact that fermionic Gaussian operations acting on Gaussian states can be efficiently simulated classically as shown in pioneering work by Terhal and DiVincenzo~\cite{PhysRevA.65.032325} and Knill~\cite{knill2001fermionic}. The model and its simulability are closely related to that of matchgate computations  introduced by Valiant~\cite{10.1145/380752.380785}, where so-called matchgates  correspond to a certain  certain subset of Gaussian operations (see also~\cite{jozsamiyake08}).  In analogy to the fermionic context, the efficient simulability of bosonic Gaussian circuits was recognized at around the same time~\cite{bartlettsanders02b,Bartlettetal02}.  In an effort to identify commonalities between simulation algorithms for a variety of quantum computational models, Somma et al.~\cite{sommaetalbarnumknill}
provided a unifying  Lie algebraic treatment which gives a counterpart to the  Gottesman-Knill theorem for the simulability of Clifford circuits~\cite{gottesman1997stabilizer,aaronsonImprovedSimulationStabilizer2004}. 

While matchgate circuits, fermionic and bosonic linear optics, and Clifford circuits provide  rich classes of efficiently simulable models for the study of many-body dynamics associated with quantum circuits, it is desirable to extend the applicability of such simulation methods. There has been significant interest in this problem resulting in a range of approaches. We only briefly discuss these here to give an overview, without attempting to give an  exhaustive treatment.

A first prominent approach  is the use of quasi-probability distributions to describe states and corresponding dynamics. Such a description typically applies to a subset of density operators: For example, it has been shown in~\cite{VeitchFerriGrossEmerson,MariEisert12} in the context of continuous-variable systems that circuits applying bosonic Gaussian operations to initial states with a positive Wigner function (a strict superset of the set of bosonic Gaussian states) can be simulated efficiently.
Negativity of the Wigner function (both in the continuous-variable as well as the qubit context) thus serves as a resource for quantum computation,  also see e.g.,~\cite{PhysRevLett.115.070501,PhysRevA.101.012350}.  It is also closely related to contextuality, see~\cite{RaussendorfBrownOkayBermejo17},  and thus connects contextuality to the complexity of classical simulation~\cite{PhysRevLett.119.120505,Frembs_2018}. Not unlike the notorious sign problem in  quantum Monte-Carlo methods applied in many-body physics, the runtimes of corresponding (randomized)  simulation algorithms scale with certain measures of ``negativity'' of the initial state.

The concept of a convex-Gaussian state was introduced and studied in~\cite{meloPowerNoisyFermionic2013a} 
to extend the range of fermionic linear optics simulation methods. This is related to quasi-probability representations in the sense that initial states of a particular form are shown to lead to efficient simulability.  Here a density operator is called convex-Gaussian if it is a convex combination of fermionic Gaussian states.
The utility of this concept was illustrated in~\cite{meloPowerNoisyFermionic2013a} by showing a converse to the fault-tolerance threshold theorem: Sufficiently noisy quantum circuits can be simulated classically because the corresponding states turn out to be convex-Gaussian. A detailed characteriziation of convex-Gaussianity is necessary to translate this into explicit (numerical) threshold estimates. An infinite hierarchy of semidefinite programs was constructed in~\cite{meloPowerNoisyFermionic2013a} to detect convex-Gaussianity, and this was subsequently shown to be complete~\cite{VershyninaPhysRevA.90.062329}. This hierarchy also provides a way of
determining whether a state is close to being convex-Gaussian~\cite{VershyninaPhysRevA.90.062329}.

A second important class of approaches are rank-based methods. Here the non-free resource (either a state or an operation) is decomposed into a linear combination of free (i.e., efficiently simulable) resources. Our work follows this approach, which is detailed in  Section~\ref{sec:extendingtosuperpos} for pure states. For Clifford computations, this involves writing general states as superpositions of stabilizer states. The development of such simulators was pioneered by Bravyi, Smith, and Smolin~\cite{bravyiTradingClassicalQuantum2016a} with subsequent work dealing with approximate stabilizer decompositions~\cite{bravyiImprovedClassicalSimulation2016}. 

The concept of low-rank (approximate) decompositions of quantum states or operations into more easily treatable basic objects appears in a variety of forms: For example, 
the work~\cite{bravyiSimulationQuantumCircuits2019a} also discusses -- in addition to state vector decompositions -- decompositions of non-Clifford unitaries into sums of Clifford operations. In Ref.~\cite{PhysRevResearch.4.043100}, a similar approach was taken to approximately decompose non-Gaussian fermionic unitary operations into linear combinations of Gaussian channels. In all these cases, the main challenge is to identify optimal (or simply good) decompositions (e.g., in terms of rank or an extent-like quantity).

In more recent work, Mocherla, Lao and Browne~\cite{mocherla2023extending} study the problem of simulating matchgate circuits using universality-enabling gates. They provide a simulation algorithm and associated runtime estimates  for estimating expectation values of single-qubit observables in output states obtained by applying a matchgate circuit to a product state. This problem is closely related to the problem considerd in this work as matchgate circuits efficiently describe evolution under a quadratic fermionic Hamiltonian. The approach taken in~\cite{mocherla2023extending} is quite different from ours, however: The classical simulator keeps track of the density operator by tracking its coefficients in the Pauli (operator) basis, using the structure of corresponding linear maps associated with matchgates. The effect of a specific set of universality-enabling gates is then analyzed in detail. This extends the sparse simulation method for matchgate circuits to circuits augmented with such gates. The runtime estimates of~\cite{mocherla2023extending} apply to certain universality-providing gates. In contrast, our  constructions can in principle also be applied to (gadget-based) constructions of arbitrary gates and provide gate-specific information.  For gates close to the identity, for example, this may provide additional resource savings (in terms of e.g., the rate of growth for several uses of such a gate).

Near the completion of our work, we became aware of concurrent independent work on fermionic circuits with non-Gaussian operations, see the papers~\cite{reardon-smithImprovedClassicalSimulation2022,StrelchukCudby} which were posted simultaneously with our work to the arXiv. Reference~\cite{reardon-smithImprovedClassicalSimulation2022} proposes an alternative classical simulation algorithm
for such circuits whose basic building blocks have identical runtime as our algorithm. In particular, when  applying a generator~$\exp(\vartheta/2 c_j c_k)$ of the Gaussian unitary group,  the runtime of the corresponding state update routine is~$O(n^3)$ as in our work. Specifying a circuit in terms of a sequence of generators is a common assumption when discussing restricted gate sets. If instead,  a general Gaussian unitary is specified by an element~$R\in SO(2n)$, the runtimes are different:  By decomposing~$R$ into~$O(n^2)$~Givens rotations (generators),  our techniques yield a runtime of order~$O(n^5)$, whereas the algorithm of~\cite{reardon-smithImprovedClassicalSimulation2022} runs in time~$O(n^4)$ in this case. A more significant distinction is the underlying approach: While we rely on a minimal extension of the covariance matrix formalism, the approach of~\cite{reardon-smithImprovedClassicalSimulation2022} involves an explicit decomposition of fermionic Gaussian unitaries inspired by a certain canonical form of Clifford circuits~\cite{bravyiSimulationQuantumCircuits2019a}. This decomposition of a Gaussian unitary~$U$ is subsequently used to specify a Gaussian state~$U|0\rangle$ including its global phase.
            
The work~\cite{StrelchukCudby} analyses the set of fermionic Gaussian states to investigate properties of the (exact and approximate) Gaussian rank, the Gaussian fidelity and the Gaussian extent. A major result obtained in~\cite{StrelchukCudby} is
the multiplicativity of the extent for a tensor product state with any number of~$4$-mode factors. This mirrors what is known for two-qubit and three-qubit factors in the stabilizer context~\cite{bravyiSimulationQuantumCircuits2019a}. The proof in Ref.~\cite{StrelchukCudby} is inspired by Ref.~\cite{heimendahlStabilizerExtentNot2021}
and makes use of tools from convex optimization. It involves the use of a net of Gaussian states to establish an upper bound on the Gaussian fidelity. While we only establish multiplicativity for two factors (i.e., a special case of their result), our approach is different and may be of independent interest: We use a net to show that multiplicativity of the Gaussian fidelity implies that of the Gaussian extent, and employ the Schmidt decomposition~\cite{BOTERO200439} of fermionic Gaussian states.

\subsection{Outline}

The paper is structured as follows. In Section~\ref{sec:background}, we give some background on fermionic linear optics, reviewing fermionic Gaussian operations and states, inner product formulas for Gaussian states and tensor products of fermionic systems. In Sections~\ref{sec:trackingphases} and \ref{sec:classicalsimulationalgorithms} we describe classical algorithms for simulation of Gaussian and non-Gaussian fermionic circuits, respectively. Specifically, in Section~\ref{sec:trackingphases} we provide an algorithm
$\Aoverlap$ for computing the overlap of two Gaussian states, an algorithm~$\Aevolve$ to simulate unitary evolution of a Gaussian state,  and algorithms~$\Ameasprob$ and~$\Ameasure$ to simulate
measurements of occupation numbers. All these algorithms keep track of the phase of the state.   In Section~\ref{sec:classicalsimulationalgorithms} we extend the simulation described in Section~\ref{sec:trackingphases} to allow for non-Gaussian input states. The remainder of this work is focused on the multiplicativity of the fermionic Gaussian extent. In Section~\ref{sec:gaussianoverlap}, we prove the multiplicativity of the fermionic Gaussian fidelity for the tensor product of any two 4-mode fermionic states with positive parity. Section~\ref{sec:multoverlapimpmultextent} is devoted to showing that the multiplicativity of the~$\cD$-fidelity implies multiplicativity of the~$\cD$-extent for general (finite and infinite, i.e., continuously parameterized) dictionaries. Finally, the results from Sections~\ref{sec:gaussianoverlap} and \ref{sec:multoverlapimpmultextent} are used to prove the main result in Section~\ref{sec:multextent}, namely the multiplicativity of the fermionic Gaussian extent for the tensor product of any two 4-mode fermionic states with positive parity.

\section{Background \label{sec:background}}

In this section, we give some background on fermionic linear optics to fix notation. 

\subsection{Dirac and Majorana operators}

Throughout, we consider fermionic systems composed of~$n$ modes, with (Dirac) creation- and annihilation operators~$a_j^\dagger, a_j$, $j\in [n]$, satisfying the canonical anticommutation relations
\begin{align}
    \{ a_j, a_k^\dagger \} =  \delta_{j,k} I\qquad \textrm{ and }\qquad \{ a_j, a_k \} =  \{ a_j^\dagger, a_k^\dagger \} =  0\qquad\textrm{ for all }\qquad
    j,k\in [n]\ .
\end{align}
The fermionic vacuum state~$\ket{0_F}$ is the (up to a phase) unique unit vector satisfying $a_j\ket{0_F}=0$ for all~$j\in [n]$. For~$x=(x_1,\ldots,x_n)\in \{0,1\}^n$, we define the number state~$\ket{x}$ by 
\begin{align}
\label{eq:fermionicBasicStatesDef}
	\ket{x} = (a_1^\dagger)^{x_1} \cdots (a_n^\dagger)^{x_n} \ket{0_F} \ . 
\end{align}
The states~$\{\ket{x}\}_{x\in \{0,1\}^n}$ are an orthonormal basis of the underlying Hilbert space~$\cH^n\cong (\mathbb{C}^2)^{\otimes n}$. A state~$\ket{x}$ is a simultaneous eigenstate of the occupation number operators~$a_j^\dagger a_j$, $j\in [n]$, where~$x_j$ is the eigenvalue of~$a_j^\dagger a_j$.  For later reference, we note that 
\begin{align}
	a_j \ket{x} = (-1)^{\eta_j(x)} x_j  \ket{x \oplus e_j} 
    \qquad\text{and}\qquad
	a^\dagger_j \ket{x} = (-1)^{\eta_j(x)} \overline{x_j}  \ket{x \oplus e_j} \label{eq:ajOnState}\ ,
\end{align}
with the definition 
\begin{align}
\label{eq:etajdef}
    \eta_j(x) = \sum_{k=1}^{j-1} x_k\qquad\textrm{ for }\qquad j\in [n]\ ,
\end{align}
where we write~$\overline{0}=1$ and~$\overline{1}=0$, where~$e_j\in \{0,1\}^{n}$ is given by~$(e_j)_k=\delta_{j,k}$ for~$k\in [n]$, and where~$\oplus$ denotes bitwise addition modulo~$2$.

It will be convenient to work with Majorana operators~$\{c_j\}_{j=1}^{2n}$ defined by 
\begin{align}
    c_{2 j-1} = a_j + a_j^{\dagger} 
    \qquad\text{and}\qquad
    c_{2 j} = i\left(a_j - a_j^{\dagger} \right) \ .  \label{eq:MajoranaToDiracFermion}
\end{align}
Majorana operators are self-adjoint
and satisfy the  relations
\begin{align}
    \{ c_j, c_k \}  = 2 \delta_{jk} I\qquad\textrm{ and }\qquad c_j^2 = I\qquad\textrm{ for }\qquad j,k\in [2n]\ .
\end{align}
For~$\alpha\in \sbin^{2n}$, we call the self-adjoint operator
\begin{align}
    c(\alpha)&=i^{|\alpha|\cdot (|\alpha|-1)/2}
    c_1^{\alpha_1}\cdots c_{2n}^{\alpha_{2n}}\ \label{eq:MajoranaMonoialDef}
\end{align}
a Majorana monomial. Here~$|\alpha|=\sum_{j=1}^{2n}\alpha_j$ denotes the Hamming weight of~$\alpha\in\{0,1\}^{2n}$. The set~$\{c(\alpha)\}_{\alpha\in \{0,1\}^n}$
constitutes an orthonormal basis of the real vector space of self-adjoint operators on~$\cH^n$ equipped with the (normalized) Hilbert-Schmidt inner product~$\langle A,B\rangle=2^{-n} \tr(A^\dagger B)$.  The Majorana monomials satisfy 
\begin{align}
	c(y)c(x) = (-1)^{|x|\cdot|y|+x\cdot y} c(x)c(y)
    \qquad\text{with}\qquad x,y \in \{0,1\}^{2n}
 \ ,
\end{align}
where~$x\cdot y=\sum_{j=1}^n x_jy_j$. 
In particular, if either~$x$ or~$y$ have even Hamming weight then~$c(x)c(y) = (-1)^{x \cdot y} c(y)c(x)$.
In the following, we will denote the set of even- and odd-weight~$2n$-bit strings by~$\{0,1\}_+^{2n}$ and~$\{0,1\}_-^{2n}$, respectively.

The parity operator 
\begin{align}
    P= i^{n} c_1 c_2 \cdots c_{2n}\ 
\end{align}
is the Majorana monomial associated with~$\alpha=1^{2n}=(1,\ldots,1)$.  The parity operator commutes with every even-weight Majorana monomial and anti-commutes with every odd-weight Majorana monomial, i.e., we have 
\begin{align}
Pc(\alpha)&=(-1)^{|\alpha|} c(\alpha) P\qquad\textrm{ for every }\qquad \alpha\in \{0,1\}^{2n}\ . \label{eq:evenweightcommuteparity}
\end{align}
The Hilbert space~$\cH^n=\cH^n_+\oplus \cH^n_-$ associated with~$n$ fermions decomposes into a direct sum of positive- and negative-parity vectors
\begin{align}
\cH^n_+ &=\left\{\Psi\in \cH^n\ |\ P\Psi=\Psi\right\} \ \text{ and }\\
\cH^n_- &=\left\{\Psi\in \cH^n\ |\ P\Psi=-\Psi\right\} \ .
\end{align}
We call a state~$\Psi\in \cH^n$  of definite parity if either~$\Psi\in\cH^n_+$ or~$\Psi\in \cH^n_-$.
An element~$X\in \cB(\cH^n)$ belonging to the set~$\cB(\cH^n)$ of linear operators on~$\cH^n$ is called  even (odd) if it is a linear combination of Majorana monomials~$c(\alpha)$ with~$\alpha \in \{0,1\}^{2n}$ of even (odd) weight. An immediate consequence of these definitions is that a state~$\Psi\in\cH^n$ has definite parity if and only if~$\proj{\Psi}$ is even (see e.g.,~\cite[Proposition 1]{amosovfillippov} for a proof).

\subsection{Gaussian unitaries \label{sec:gaussianunitaries}} 

A unitary operator~$U$ on~$\cH^n$ is Gaussian if and only if it maps a Majorana operator~$c_j$ to a linear combination of Majorana operators, i.e.
\begin{align}
    \label{eq:UGaussianRRelation}
	U c_j U^\dagger = \sum_{k=1}^{2n} R_{jk} c_k \ ,
\end{align}
where~$R \in O(2n)$ is a real orthogonal matrix. Ignoring overall phases, the group of Gaussian unitary operators is generated by 
operators of the form
\begin{align}
     U_{j,k}(\vartheta)=\exp(\vartheta/2 c_j c_k)\qquad\textrm{ with }\qquad \vartheta\in [0,2\pi )\textrm{ and } j < k\in [2n]
\end{align} and by operators 
\begin{align}
    U_j=c_j\qquad\textrm{ with }\qquad j\in [2n]\ .
\end{align}
The operator~$U_{j,k}(\vartheta)$ implements the rotation 
\begin{equation}\label{eq:ujkthetaaction}
  \begin{alignedat}{2}
    U_{j,k}(\vartheta) c_j  U_{j,k}(\vartheta)^\dagger &= \cos (\vartheta) c_j-\sin (\vartheta) c_k&& \\
     U_{j,k}(\vartheta) c_k  U_{j,k}(\vartheta)^\dagger &= \sin (\vartheta) c_j+\cos (\vartheta) c_k&& \\
     U_{j,k}(\vartheta) c_\ell  U_{j,k}(\vartheta)^\dagger &= c_{\ell} 
     &&\qquad\textrm{for }\qquad\ell\not\in\{j,k\}\ .
  \end{alignedat}
\end{equation}
The operator~$U_j=c_j$ leaves~$c_j$ invariant and flips the sign of each~$c_k$ with~$k \neq j$, i.e., it implements the reflection 
\begin{align}
        U_j c_j U_j^{\dagger} &= c_j\\
        U_j c_k U_j^{\dagger} &= -c_k \qquad\textrm{ for }\qquad k\neq j\ .\label{eq:ujdef}
\end{align}
We note that by relation~\eqref{eq:evenweightcommuteparity}, every generator~$U_{j,k}(\vartheta)$ is parity-preserving, whereas every generator~$U_j$ reverses the parity, i.e., 
\begin{equation}\label{eq:UjkvarthetaP}
  \begin{alignedat}{2}
    U_{j,k}(\vartheta)PU_{j,k}(\vartheta)^\dagger &=P
    &&\qquad\text{ for all }\qquad k > j \in [n], \vartheta\in [0,2\pi) \ ,
    \\
    U_jPU_j^\dagger &=-P
    &&\qquad\text{ for all }\qquad j \in [n] \ .
  \end{alignedat}
\end{equation}

Every orthogonal matrix~$R$ gives rise to a  Gaussian unitary~$U_R$ satisfying~\eqref{eq:UGaussianRRelation}. The unitary~$U_R$ is unique up to a global phase, and~$R\mapsto U_R$ is called the metaplectic representation. We can fix the global phase of~$U_R$ uniquely, e.g.,  by the following procedure. Every element~$R\in O(2n)$ can be uniquely decomposed into a product
\begin{align}
    R&=S_0S_1\cdots S_L\label{eq:soL}
\end{align}
with~$L\leq \frac{2n(2n-1)}{2}$ and where
\begin{align}
S_0&=\begin{cases}
    I\qquad\textrm{ if } R\in SO(2n)\\
    R_1\qquad\textrm{ otherwise }\ ,   
\end{cases}
\end{align}
where for each~$r\in [L]$, the matrix~$S_r$ is of the form
\begin{align}
    S_r&=R_{j_r,k_r}(\vartheta_r)\qquad\textrm{ for some }\qquad j_r < k_r\in [2n], \vartheta_r\in [0,2\pi)\ .
\end{align}
Here~$R_1\in O(2n)$ is associated with the unitary~$U_1$ by Eq.~\eqref{eq:ujdef}, whereas~$R_{j,k}(\vartheta)\in SO(2n)$ is associated with~$U_{j,k}(\vartheta)$ according to Eq.~\eqref{eq:ujkthetaaction}. We note that~$R_{j,k}(\vartheta)\in SO(2n)$ is a so-called Givens rotation, introduced in Ref.~\cite{doi:10.1137/0106004}, and a decomposition
of the form~\eqref{eq:soL} can be found by a deterministic algorithm with runtime~$O(n^3)$ (see e.g., Section 5.2.3 in Ref.~\cite{10.5555/248979}).  In particular, application of this algorithm defines a function taking an arbitrary element~$R\in O(2n)$ to a unique product of the form~\eqref{eq:soL}.
Given the (unique) decomposition~\eqref{eq:soL} of~$R\in O(2n)$, we can then define~$U_R$ as the product
\begin{align}
    U_R&=U_1 U_{j_1,k_1}(\vartheta_1)\cdots U_{j_L,k_L}(\vartheta_L)\ .
\end{align}
Overall, this defines a function~$R\mapsto U_R$ from~$O(2n)$ to the set of Gaussian unitaries, fixing the phase ambiguity.  Throughout the remainder of this work, $U_R$ will denote the Gaussian unitary uniquely fixed by~$R$.

\subsection{Fermionic Gaussian (pure) states}
 
The set of pure fermionic Gaussian  states is the orbit of the vacuum state~$\ket{0_F}$ under the action of~$O(2n)$ defined by the metaplectic representation, i.e., fermionic Gaussian states are of the form~$U_R\ket{0_F}$ with~$U_R$ a fermionic Gaussian unitary. In more detail, every fermionic Gaussian state~$e^{i\theta}U_R\ket{0_F}$ is uniquely specified by a pair~$(\theta,R)$ with~$\theta\in [0,2\pi)$ and~$R\in O(2n)$. We will denote the set of all fermionic Gaussian states by
\begin{align}
    \cG_n=\left\{ e^{i\theta}U_R\ket{0_F}\ |\ \theta\in [0,2\pi), R\in O(2n)\right\}\ .
\end{align}
By Eq.~\eqref{eq:UjkvarthetaP} and because~$P\ket{0_F}=\ket{0_F}$,  every pure fermionic  Gaussian state~$\Psi$ has a fixed parity, i.e., it is an eigenvector of the parity operator~$P$. This defines a disjoint partition~$\cG_n=\cG_n^+\cup \cG_n^-$ of the set of fermionic Gaussian states into positive- and negative-parity states.

\subsection{Gaussianity condition}

In Ref.~\cite{bravyiLagrangianRepresentationFermionic2004} Bravyi established a necessary and sufficient condition to determine if a (possibly mixed) state~$\rho\in\cB(\cH^n)$ is Gaussian (see Theorem~1 therein). Here Gaussianity of a density operator~$\rho$ is defined by the condition that~$\rho$ has the form
\begin{align}
\rho &= K\exp\left(i\sum_{j,k=1}^{2n} A_{j,k}c_jc_k\right)\label{eq:rhogaussiandensityoperator}
\end{align}
for an antisymmetric matrix~$A=-A^T\in \mathsf{Mat}_{2n\times 2n}(\mathbb{R})$ and  a constant~$K>0$. We note that a pure state~$\Psi\in\cH^n$ is Gaussian if and only if the associated density operator~$\rho=\proj{\Psi}$ is Gaussian. (This follows from the fact that~$\proj{0_F}=\frac{1}{2^{n/2}}\exp\left(i\frac{\pi}{4}\sum_{j=1}^n c_{2j-1}c_{2j}\right)$. Indeed, if~$\rho=\proj{\Psi}$ is a rank-one projection of the form~\eqref{eq:rhogaussiandensityoperator}, then it follows from Williamson's normal form for antisymmetric matrices  that there is~$R\in O(2n)$ such that~$RAR^T=\bigoplus_{j=1}^n \left(\begin{smallmatrix}
    0 & 1\\
    -1 & 0
    \end{smallmatrix}\right)$. This implies that~$U_R\proj{\Psi}U_R^\dagger=\proj{0_F}$ and thus~$\ket{\Psi}=U^\dagger_R\ket{0_F}$. 
Conversely, for a Gaussian state~$\ket{\Psi}=e^{i\theta}U_R\ket{0_F}\in \cG_n$, we can use the expression for~$\proj{0_F}$ to argue that~$\proj{\Psi}$ is of the form~\eqref{eq:rhogaussiandensityoperator}, i.e., Gaussian.)
The characterization of Gaussian density operators established in~\cite{bravyiLagrangianRepresentationFermionic2004} is the following. 
\begin{theorem}[Theorem~1 in \cite{bravyiLagrangianRepresentationFermionic2004}]
Define~$\Lambda = \sum_{j=1}^{2n} c_j \otimes c_j$. An even state~$\rho\in\cB(\cH^n)$ is Gaussian if and only if 
$[\Lambda,\rho\otimes \rho]=0$.
\end{theorem}
Based on this characterization~\cite{bravyiLagrangianRepresentationFermionic2004}, the following was shown in~\cite{meloPowerNoisyFermionic2013a}.
\begin{lemma}[Corollary~1 in~\cite{meloPowerNoisyFermionic2013a}]\label{lem:fermionicgaussianityconditionm}
    Let~$\rho\in\cB(\cH^n)$ be an even state. Then~$\rho$ is a Gaussian  pure state if and only if~$\Lambda \left(\rho\otimes\rho\right)=0$.
    \end{lemma}
In the following, we only use the statement of Lemma~\ref{lem:fermionicgaussianityconditionm} applied to pure states in order to 
 distinguish between Gaussian and non-Gaussian pure states. We formulate this as follows:

\begin{lemma}
\label{lem:gaussianitylambdapurestates}
    Let~$\Psi\in\cH^n$ be a pure state with fixed parity. Then~$\Psi$ is Gaussian  if and only if 
    \begin{align}
        \Lambda (\ket{\Psi}\otimes\ket{\Psi})&=0\ .
    \end{align}
\end{lemma}
\begin{proof}
This follows immediately from the equivalence of the concepts of Gaussianity of pure states (vectors) and density operators because the density operator~$\proj{\Psi}$ is even for any fixed-parity state~$\Psi$.
\end{proof}
We note that there is an elegant representation-theoretic interpretation of this characterization of Gaussianity~\cite{lichtenstein}. It is derived from the fact that Gaussian states are the orbit of the vacuum state~$\ket{0_F}$ (a highest weight state) under the action of the metaplectic group, cf.~\cite[Section IV]{kusClassicalQuantumStates2009} and~\cite{oszmanieckus}. We use a version of this reformulation for~$4$~fermions, see Lemma~\ref{lem:thetamapGaussianity} below, that has first been obtained in~\cite{oszmaniecClassicalSimulationFermionic2014a}.

\subsection{Covariance matrices, Gaussian states and Wick's theorem
\label{sec:covmatrixwick}} 

The covariance matrix~$\cov=\cov(\Psi)\in\mathsf{Mat}_{2n\times 2n}(\mathbb{R})$ of a state~$\Psi\in\cH^n$ is the antisymmetric matrix with entries
\begin{align}
\cov_{j,k} (\Psi) &=
\begin{cases}
\lrangle{ \Psi,  i c_j c_k \Psi}  	&\quad\text{for } j \neq k \\
0 				&\quad\text{for } j = k 
\end{cases} 
\label{eq:covarianceMatrixDef}
\end{align}
with~$j,k \in [2n]$. It satisfies~$\cov \cov^T = I$ for any state~$\Psi\in\cH^n$.

The expectation value of a Hermitian operator with respect to a Gaussian state~$\Psi$ is fully determined by its covariance matrix~$\cov=\cov(\Psi)$. This is because  the expectation value of a Majorana monomial~$c(\alpha)$, $\alpha\in \{0,1\}^{2n}$, is given by Wick's theorem
\begin{align}
    \lrangle{\Psi, c(\alpha) \Psi} =
    \begin{cases}
         \Pf(\cov[\alpha]) \qquad & \textrm{ if } |\alpha| \textrm{ is even }\\
         0 & \textrm{ otherwise }
    \end{cases}\label{eq:wickstheorem} \ .
\end{align}
Here~$\cov[\alpha]\in \mathsf{Mat}_{|\alpha|\times|\alpha|}(\mathbb{R})$ is the submatrix of~$\cov$ which includes all rows and columns with index~$j\in [2n]$ such that~$\alpha_j=1$. Evaluating such expectation values, i.e., computing Pfaffians of~$|\alpha|\times |\alpha|$-matrices (with~$|\alpha|$ even), takes time~$O(|\alpha|^3)$.
(Here and below we use the number of elementary arithmetic operations to quantify the time complexity of algorithms.)

\subsection{Inner product formulas for Gaussian states \label{sec:innerPrpductFormulas}}

The modulus of the inner product of two Gaussian states~$\Phi_1,\Phi_2$ with identical parity~$\sigma\in \{\pm 1\}$ and covariance matrices~$\cov_1,\cov_2$ is given by the expression~\cite{Loewdin55}
\begin{align}
    |\langle \Phi_1,\Phi_2\rangle |^2 &=\sigma 2^{-n}\Pf(\cov_1+\cov_2)\ .\label{eq:innerproductoftwomatrices}
\end{align}

For three Gaussian states~$\Phi_0,\Phi_1,\Phi_2$,  the expression
$\left\langle\Phi_0 , \Phi_1\right\rangle \cdot\left\langle\Phi_1 , \Phi_2\right\rangle\cdot \langle \Phi_2,\Phi_0\rangle$ is invariant under a change of the global phase of any of the states, and can therefore be computed by the covariance matrix formalism. An explicit expression was derived by L\"owdin in~\cite{Loewdin55}. In Ref.~\cite{bravyiComplexityQuantumImpurity2017a} Bravyi and Gosset gave  the formula
\begin{align}
\label{eq:threestateexpr}
\left\langle\Phi_0 , \Phi_1\right\rangle \cdot\left\langle\Phi_1 , \Phi_2\right\rangle \cdot\left\langle\Phi_2 , \Phi_0\right\rangle=\sigma 4^{-n} i^n \operatorname{Pf}\left(\begin{array}{ccc}
i \cov_0 & -I & I \\
I & i \cov_1 & -I \\
-I & I & i \cov_2
\end{array}\right)
\end{align}
for three Gaussian states~$\{\Phi_j\}_{j=0}^2$ of identical parity~$\sigma\in \{\pm 1\}$, where~$\cov_j=\cov(\Phi_j)$ is the covariance matrix of~$\Phi_j$ for~$j=0,1,2$.  More generally, they 
obtained the formula
\begin{align}
\label{eq:threestatemajoranaexpr}
\left\langle\Phi_0 , \Phi_1\right\rangle \cdot\left\langle\Phi_1, c(\alpha) \Phi_2\right\rangle \cdot\left\langle\Phi_2 , \Phi_0\right\rangle=\sigma 4^{-n} 
i^{n+|\alpha|\cdot (|\alpha|-1)/2}
\operatorname{Pf}\left(R_\alpha\right) \ ,
\end{align}
for any even-weight Majorana monomial~$c(\alpha)$, $\alpha\in \sbin^{2n}_+$, where
\begin{align}
R_\alpha=\left(\begin{array}{cccc}
i \cov_0 & -I & I & 0\\
 I & i \cov_1 & -I & 0\\
-I & I & i D_\alpha \cov_2 D_\alpha & J_\alpha^T+i D_\alpha \cov_2 J_\alpha^T \\
0 & 0& -J_\alpha+i J_\alpha \cov_2 D_\alpha & i J_\alpha \cov_2 J_\alpha^T
\end{array}\right)\in \mathsf{Mat}_{(6n+|\alpha|)\times(6n+|\alpha|)}(\mathbb{R})\ .\label{eq:Ralphadefinitionm}
\end{align}
Here
$D_\alpha=\mathsf{diag}(\{1-\alpha_j\}_{j=1}^{2n})$ is a diagonal matrix, whereas~$J_\alpha\in\mathsf{Mat}_{|\alpha|\times 2n}(\mathbb{R})$
has entries
defined in terms of the indices~$\{i\in [2n]\ |\ \alpha_i\neq 0\}=\{i_1< \cdots <i_r\}$ associated with non-zero entries of~$\alpha$, that is,
\begin{align}
    (J_\alpha)_{j,k}=\begin{cases}
     \delta_{i_j,k}\qquad &\textrm{ if }j\leq r\\
     0 &\textrm{ otherwise }\ .
    \end{cases}
\end{align}
In other words, $\left(J_\alpha\right)_{j, k}=1$ if  and only if~$k$ is the position of the~$j$-th nonzero element of~$\alpha$.

As argued in~\cite{bravyiComplexityQuantumImpurity2017a}, 
expressions~\eqref{eq:threestateexpr} and~\eqref{eq:threestatemajoranaexpr} determine the
inner product~$\langle \Phi_1,\Phi_2\rangle$ and an expression of the form~$\langle \Phi_1,c(\alpha)\Phi_2\rangle$ entirely in terms of covariance matrices, assuming that the remaining two overlaps~$\langle \Phi_0,\Phi_1\rangle$, $\langle \Phi_2,\Phi_0\rangle$ with a Gaussian reference state~$\Phi_0$ are given and non-zero. In this situation, these quantities can be computed in time~$O(n^3)$.

\subsection{Gaussian evolution and occupation number measurement
\label{sec:backgroundFermionicGaussianClassicalSimulation}
}

Underlying the known classical simulation algorithms for fermionic linear optics is the fact that Gaussian unitaries and occupation number measurements preserve Gaussianity. Explicitly, this can be described as follows: Given a Gaussian state~$\Psi$ with covariance matrix~$\cov(\Psi)$
\begin{enumerate}[(i)]
\item the covariance matrix~$\cov(U_R\Psi)$ of~$\Psi$ evolved under the Gaussian unitary~$U_R$, $R\in O(2n)$,  is given by~$\cov(U_R \Psi) = R \cov(\Psi) R^T$. 
\item measurement of the observable~$a_j^\dagger a_j$ for~$j\in[n]$ gives the outcome~$s\in\{0,1\}$ with probability
\begin{align}
\label{eq:probNumberMeas}
    P_j(s) &= \|\Pi_j(s)\Psi\|^2 \\ 
    &= \langle \Psi, \Pi_j(s) \Psi \rangle \\
    &= \frac{1}{2} \left( 1  + (-1)^{s} \langle \Psi,  i c_{2j-1} c_{2j} \Psi \rangle \right) \\
    &= \frac{1}{2} (1+(-1)^{s}\Gamma_{2j-1,2j}) \ .
\end{align}
where the third identity follows from the definition
\begin{align}
    \Pi_j(s) = \frac{1}{2}(1+(-1)^{s}ic_{2j-1} c_{2j})
\end{align}
and the last identity follows from the definition of the covariance matrix in Eq.~\eqref{eq:covarianceMatrixDef}. Given that the  measurement outcome is~$s\in \{0,1\}$, the post-measurement state 
\begin{align}
\Psi(s)&=\left( \Pi_j\left(s\right) \ket{\Psi} \right)/ \sqrt{P_j\left(s\right)}
\end{align}
is Gaussian with covariance matrix~$\cov(\Psi(s))$ defined by the lower-diagonal entries (see e.g.,~\cite[Proof of Lemma~4]{bravyiDisorderAssistedErrorCorrection2012})
\begin{align}
\label{eq:MpostMeasurement}
\cov(\Psi(s))_{k, \ell} = 
\begin{cases}
(-1)^{s}      &\text{if } (k,\ell)=(2j,2j-1)\\
\cov_{k,\ell}+ \frac{(-1)^{s}}{2 P_j\left(s\right)} \left(\cov_{2 j-1, \ell} \cov_{2 j, k}  -\cov_{2 j-1, k} \cov_{2 j, \ell}\right) &\text{otherwise}
\end{cases}
\end{align}
for~$k>\ell$. 
\end{enumerate}
In particular, the corresponding resulting covariance matrices can be computed in time $O(n^3)$ \cite{10.1145/380752.380785,PhysRevA.65.032325,knill2001fermionic,bravyiLagrangianRepresentationFermionic2004} and~$O(n^2)$ \cite{bravyiDisorderAssistedErrorCorrection2012} for unitary evolution and measurement, respectively.

\subsection{The tensor product of two fermionic states\label{sec:tensorproductevenf}}

Two density operators~$\rho_j\in\cB(\cH^{n_j})$, $j\in[2]$, have a joint extension if and only if there is an element~$\rho\in\cB(\cH^{n_1+n_2})$ such that
\begin{align}
\label{eq:jointextension}
\tr(c(\alpha_1\|\alpha_2)\rho)&=\tr(c(\alpha_1)\rho_1)\tr(c(\alpha_2)\rho_2)\qquad\textrm{ for all }\qquad \alpha_j\in \{0,1\}^{2n_j}, j\in[2]\ .
\end{align} 
Here~$\alpha_1\|\alpha_2\in \{0,1\}^{2(n_1+n_2)}$ denotes the concatenation of~$\alpha_1$ and~$\alpha_2$. 
Theorem~1 in~\cite{arakimoriya} implies that if either~$\rho_1$ or~$\rho_2$ is even, then a unique joint extension~$\rho$ of~$(\rho_1,\rho_2)$ exists. Furthermore, this extension is even if and only if  both~$\rho_1$ and~$\rho_2$ are even.  Theorem~2 in~\cite{arakimoriya} shows that if~$\rho$ is even and~$\rho_1$ and~$\rho_2$ are pure, then~$\rho$ is also pure.

In particular, this means that for states~$\Psi_1,\Psi_2$ of definite parity, there is a unique joint pure extension~$\rho=\proj{\Psi}$ of~$\left(\proj{\Psi_1},\proj{\Psi_2}\right)$. Since~$\rho$ is pure, this also means that~$\Psi$ is of definite parity.
We will write~$\Psi=\Psi_1\motimes \Psi_2$ for this state in the following, and we call $\tilde{\otimes}$ the fermionic tensor product. Note that~$\Psi$ is only defined up to a global phase. 
It follows immediately from these definitions that 
\begin{align}
\left|\langle x,y| \Psi_1\motimes\Psi_2\rangle\right| &=  \left|\langle x,\Psi_1\rangle \cdot \langle y,\Psi_2\rangle\right|\qquad\textrm{ for all }\qquad x\in \{0,1\}^{n_1}\textrm{ and }y\in \{0,1\}^{n_2}\ .\label{eq:tensorproductproperty}
\end{align}
\begin{proof}
Let~$x\in \{0,1\}^{n_1}$ and~$y\in \{0,1\}^{n_2}$ be arbitrary.
By definition, we have
\begin{align}
\proj{x,y} &=\left(
\prod_{j=1}^{n_1} \frac{1}{2}(I+(-1)^{x_j} i c_{2j-1}c_{2j})
\right)\left(
\prod_{k=1}^{n_2} \frac{1}{2}(I+(-1)^{y_k} i c_{2n_1+2k-1}c_{2n_1+2k})
\right)\\
&=\left(\sum_{\alpha \in \{0,1\}^{2n_1}_+} \gamma_{x}(\alpha) c(\alpha || 0^{2n_2})\right)
\left(\sum_{\beta \in \{0,1\}^{2n_2}_+} \gamma_{y}(\beta) c( 0^{2n_1} ||\beta)\right)\ .
\end{align}
for certain coefficients~$\{\gamma_x(\alpha)\}_{\alpha\in \{0,1\}^{2n_1}_+}$ and~$\{\gamma_y(\beta)\}_{\beta\in \{0,1\}^{2n_2}_+}$.
Since~$\rho=\proj{\Psi_1\motimes\Psi_2}$ is an extension of~$(\proj{\Psi_1},\proj{\Psi_2})$ and~$c(\alpha || 0^{2n_2}) c( 0^{2n_1} ||\beta ) =i^{-|\alpha|\cdot |\beta|}c(\alpha||\beta)= c(\alpha || \beta)$
for (even-weight)~$\alpha\in \{0,1\}^{2n_1}_+$ and~$\beta\in \{0,1\}^{2n_2}_+$, it follows that
\begin{align}
\left|\langle x,y| \Psi_1\motimes\Psi_2\rangle\right|^2 &=
\sum_{\alpha\in \{0,1\}^{2n_1}_+}\gamma_x(\alpha)\sum_{\beta\in \{0,1\}^{2n_2}_+}\gamma_y(\beta)
\tr(c(\alpha\|\beta)\rho)\\
&=\langle \Psi_1,\left(\sum_{\alpha\in \{0,1\}^{2n_1}_+}\gamma_x(\alpha) c(\alpha)\right)\Psi_1\rangle\cdot \langle\Psi_2,\left(\sum_{\beta\in \{0,1\}^{2n_2}_+}\gamma_y(\beta) c(\beta)\right)\Psi_2\rangle\\
&=|\langle x,\Psi_1\rangle|^2\cdot |\langle y,\Psi_2\rangle|^2\ .
\end{align}
\end{proof}
Refining Eq.~\eqref{eq:tensorproductproperty}, (relative) phase information between these matrix elements can be obtained from 
the explicit construction of~$\Psi_1\motimes \Psi_2$ given in~\cite[Section~3.1]{arakimoriya} (see also~\cite[Proof of Theorem~1]{amosovfillippov}): Consider 
the isometry
\begin{align}
    \begin{matrix}
     U: & \cH^{n_1+n_2} & \rightarrow &\cH^{n_1}\otimes \cH^{n_2}\\
      & \ket{x_1,\ldots,x_{n_1+n_2}} & \mapsto &\ket{x_1,\ldots,x_m}\otimes \ket{x_{n_1+1},\ldots,x_{n_1+n_2}}
    \end{matrix}
\end{align}
whose action is given by
\begin{align}
Ua_j U^\dagger &=\begin{cases}
    a_j\otimes I \qquad & \textrm{ if }j \in [n_1]\\
    P_1 \otimes a_{j-n_1} \qquad &\textrm{ if } j\in \{n_1+1,\ldots,n_1+n_2\}\ .
\end{cases}
\end{align}
where~$P_1$, the parity operator acting on~$\cH_{n_1}$, introduces phases.
Then 
\begin{align}
    \Psi_1\motimes\Psi_2&=U^\dagger (\Psi_1\otimes \Psi_2)\ .
\end{align}
It is straightforward from this definition to check that~$\Psi_1\motimes\Psi_2$ is the extension of~$(\Psi_1,\Psi_2)$ and 
\begin{align}
\langle x,y|\Psi_1\motimes\Psi_2\rangle &=(-1)^{|x|}
\cdot \langle x,\Psi_1\rangle\cdot \langle y,\Psi_2\rangle\qquad\textrm{ for all }\qquad x\in \{0,1\}^{n_1}\textrm{ and }y\in \{0,1\}^{n_2}\ .\label{eq:phaseintensorproduct}
\end{align}

We note that the fermionic tensor product preserves Gaussianity in the following sense.
\begin{lemma}
Let~$\Psi_j\in\cG_{n_j}^+$ be positive-parity fermionic Gaussian states for~$j\in[2]$. 
Then $\Psi_1 \tilde{\otimes} \Psi_2\in\cG_{n_1+n_2}^+$, i.e., it is an even fermionic Gaussian state.
\end{lemma}
\begin{proof}
By definition  of an extension (see Eq.~\eqref{eq:jointextension})
and Wick's theorem (Eq.~\eqref{eq:wickstheorem}), the tensor product~$\Psi=\Psi_1\motimes \Psi_2$ 
satisfies
\begin{align}
\langle \Psi,c(\alpha_1\|\alpha_2)\Psi\rangle&=
\begin{cases}
 \Pf(\cov_1[\alpha_1])
 \Pf(\cov_1[\alpha_1]) \qquad & \textrm{ if both } |\alpha_1|\textrm{ and }|\alpha_2| \textrm{ are even }\\
         0 & \textrm{ otherwise }\ 
   \end{cases}\label{eq:blockexpectation}
\end{align}
for all~$\alpha_j\in \{0,1\}^{2n_j}$, where~$\cov_j$ is the covariance matrix of~$\Psi_j$ for~$j\in[2]$. Because the Pfaffian satisfies
\begin{align}
    \Pf(A_1\oplus A_2)&=\Pf\begin{pmatrix}
A_1 & 0\\
0 & A_2
    \end{pmatrix}=\Pf(A_1)\Pf(A_2)
\end{align}
for block-matrices, it follows from~\eqref{eq:blockexpectation} that the tensor product~$\Psi$ satisfies Wick's theorem~\eqref{eq:wickstheorem} with covariance matrix~$\cov_1\oplus \cov_2$. In particular, it is Gaussian.
\end{proof}

\section{Tracking relative phases in fermionic linear optics\label{sec:trackingphases}}

The covariance matrix~$\cov(\Psi)$ of a fermionic Gaussian state~$\ket{\Psi}=e^{i\theta}U_R\ket{0_F}\in\cG_n$ fully determines expectation values by Wick's theorem, which is why Gaussian states and dynamics are amenable to efficient classical simulation (see Section~\ref{sec:backgroundFermionicGaussianClassicalSimulation}). However, the description of~$\Psi$ in terms of the covariance matrix~$\cov(\Psi)$ does not capture information on the global phase~$e^{i\theta}$ of the state.
For processes involving  non-Gaussian states expressed as superpositions of Gaussian states, such phase information needs to be available for computing norms, expectation values and overlaps. 
 
Here we provide an extended (classical) description of fermionic Gaussian states  that incorporates phase information. A central feature of our construction is the fact that this description permits to compute overlaps (including relative phases, i.e., not only the absolute value) of Gaussian states in an efficient manner.
 
Our construction is motivated by and relies on expression~\eqref{eq:threestateexpr}, which 
relates the inner product~$\langle \Psi_1,\Psi_2\rangle$ of two Gaussian states~$\Psi_1,\Psi_2\in\cG_n$
to their inner products~$\langle \Psi_0,\Psi_1\rangle$, $\langle \Psi_0,\Psi_2\rangle$ with a Gaussian reference state~$\Psi_0\in \cG_n$ and their covariance matrices~$\cov_0,\cov_1,\cov_2$. This suggests fixing a reference state~$\Psi_0\in \cG_n$ and using the pair~$(\cov(\Psi),\langle \Psi_0,\Psi\rangle)$ as a classical description of any state~$\ket{\Psi}\in \cG_n$ relevant in the computation. The problem with this idea is that~$\langle \Psi_0,\Psi\rangle$ may vanish, preventing the application of~\eqref{eq:threestateexpr}. To avoid this problem,
we use  -- instead of a single state~$\Psi_0$ -- 
a  (potentially) different reference state for each state~$\Psi$. Specifically, we will show that using number states, i.e., states of the form~\eqref{eq:fermionicBasicStatesDef},
suffices. This motivates the following definition.
\begin{definition}
Let~$\ket{\Psi}=e^{i\theta}U_R\ket{0_F}\in\cG_n$ be a Gaussian state. We call a tuple 
\begin{align}
d=(\cov,x,r)\in  \mathsf{Mat}_{2n\times 2n}(\mathbb{R})\times\{0,1\}^{n}\times \mathbb{C}
\end{align}
a (valid)  description of~$\ket{\Psi}$ if the following three conditions hold:
\begin{enumerate}[(i)]
\item\label{it:descriptionfirst}
$\cov=\cov(\Psi)$ is the covariance matrix of~$\ket{\Psi}$.
\item\label{it:descriptionsecond}
$x\in \{0,1\}^n$ is such that~$\langle x,\Psi\rangle\neq 0$, where~$\ket{x}$ is the number state defined by Eq.~\eqref{eq:fermionicBasicStatesDef}.

In our algorithms we will in fact ensure that~$|\langle x,\Psi\rangle|^2\geq 2^{-n}$, i.e., only a subset of valid descriptions is used. A description~$d=(\Gamma,x,r)$ with this property, i.e., satisfying~$|r|^2\geq 2^{-n}$, will be called a good description. The restriction to good descriptions is necessary to make our algorithms work with finite-precision arithmetic.
\item\label{it:descriptionthird}
$r=\langle x,\Psi\rangle$.
\end{enumerate}
\end{definition}
More explicitly, necessary and sufficient conditions for~$d=(\cov(\Psi),x,r)$ to constitute a description of~$\Psi$ are that
\begin{align}
r\neq 0\qquad\textrm{ and }\qquad |r|^4=2^{-2n}\mathsf{Det}(\cov(\ket{x})+\cov(\Psi))
\end{align}
because of formula~\eqref{eq:innerproductoftwomatrices} for the overlap of two states and because~$\mathsf{Det}(\cdot) = \Pf^2(\cdot)$. 
Here
\begin{align}
\cov(\ket{x})=\bigoplus_{j=1}^n \begin{pmatrix}
0 & (-1)^{x_j}\\
-(-1)^{x_j} & 0 
\end{pmatrix}\label{eq:covariancematrixMx}
\end{align}
is the covariance matrix of~$\ket{x}$.  Since a Gaussian state~$\Psi$ generally has non-zero overlap with more than a single occupation number state~$\ket{x}$, there are several distinct valid descriptions of~$\Psi$. We will denote  the set of descriptions of~$\ket{\Psi}\in \cG_n$ by~$\mathsf{Desc}(\Psi)$. 

We note that  a description~$d=(\cov,x,r)$ uniquely fixes a  Gaussian state~$\Psi\in\cG_n$: The covariance matrix~$\cov$ determines all expectation values, and the global phase of~$\Psi$ is fixed by the overlap~$\langle x,\Psi\rangle$, i.e., by~$r$. Denoting by~$\mathsf{Desc}_n=\bigcup_{\Psi\in\cG_n} \mathsf{Desc}(\Psi)$ the set of all descriptions of fermionic Gaussian~$n$-mode states, this means that we have a function
\begin{align}
\begin{matrix}
\Psi &: \mathsf{Desc}_n & \rightarrow & \cG_n\\
 & d & \mapsto & \Psi(d)\ .
\end{matrix}\label{eq:psifunctiondescr}
\end{align}

The main result of this section shows that expectation values, overlaps, and matrix elements of (Majorana) operators with respect to Gaussian states can be efficiently computed from their classical descriptions.  Furthermore, when evolving a Gaussian state under a Gaussian unitary, the description of the resulting state can be computed efficiently. The same is true for the post-measurement state when applying an occupation number measurement. 

For evolution, we note that it suffices to consider Gaussian unitaries of the form~$U_R$ where~$R\in O(2n)$ belongs to the set of generators~$\mathsf{Gen}(O(2n))$ introduced in Section~\ref{sec:gaussianunitaries}, that is, 
\begin{align}
    \label{eq:Genset}
\mathsf{Gen}(O(2n))=\left\{R_{j,k}(\vartheta)\ |\ j< k\in [2n], \vartheta\in [0,2\pi)\right\}
\cup \{R_j\}_{j=1}^{2n}\ .
\end{align}
Here~$R_{j,k}(\vartheta)$ is a Givens rotation and~$R_j=-\mathsf{diag}(\{(-1)^{\delta_{j,k}}\}_{k=1}^{2n})$ a reflection.  We note that each element of~$\mathsf{Gen}(O(2n))$ can be specified by a tuple~$(j,k,\vartheta)\in [2n]\times [2n]\times [0,2\pi)$ or an index~$j\in [2n]$, respectively. We assume that this parameterization is used in the following algorithms (but leave this implicit). 

To state the properties of our (deterministic) algorithms,  it is convenient to express these as functions.
\begin{theorem}[Overlap, evolution, and measurement]\label{thm:maindescriptalgorithm}
Let~$\Psi(d)\in \cG_n$ be the Gaussian state associated with a description~$d\in\mathsf{Desc}_n$, see Eq.~\eqref{eq:psifunctiondescr}. Then the following holds:
\begin{enumerate}[(i)]
\item
The algorithm~$\Aoverlap:  \mathsf{Desc}_n\times\mathsf{Desc}_n  \rightarrow \mathbb{C}$
given in Fig.~\ref{fig:Aoverlap} has runtime~$O(n^3)$ and satisfies
\begin{align}
\Aoverlap(d_1,d_2)&=\langle \Psi(d_1),\Psi(d_2)\rangle\qquad\textrm{ for all }\qquad d_1,d_2\in\mathsf{Desc}_n\ .
\end{align}
\item
The algorithm~$\Aevolve: \mathsf{Desc}_n\times \mathsf{Gen}(O(2n))\rightarrow\mathsf{Desc}_n$ given in Fig.~\ref{fig:Aevolve} has runtime~$O(n^3)$ and satisfies
\begin{align}
\Psi(\Aevolve(d,R))=U_R\Psi(d)\qquad\textrm{ for all }\qquad d\in\mathsf{Desc}_n\textrm{ and } R\in \mathsf{Gen}(O(2n))\ ,
\end{align}
where~$U_R$ denotes the Gaussian unitary associated with~$R\in O(2n)$.
\item The algorithm~$\Ameasprob: \mathsf{Desc}_n \times [n] \times \sbin \rightarrow \bbR$ given in Fig.~\ref{fig:Ameasprob} has runtime~$O(1)$ and satisfies
\begin{align}
    \Ameasprob(d, j, s) = \|\Pi_j(s)\Psi(d)\|^2
    \qquad
    \textrm{ for all }\qquad d\in \mathsf{Desc}_n, j\in [n], s\in \{0,1\}\ ,
\end{align}
where~$\Pi_j(s)=\frac{1}{2}(I+(-1)^{s}ic_{2j-1}c_{2j})$ is the projection onto the eigenspace of~$a_j^\dagger a_j$ with eigenvalue~$s$. \label{it:thmMeasurementProb}
\item
The algorithm~$\Ameasure:\mathsf{Desc}_n \times [n]\times \{0,1\} \times [0,1]\rightarrow\mathsf{Desc}_n$ given in Fig.~\ref{fig:Ameasure} has runtime~$O(n^3)$.
The algorithm satisfies
\begin{align}
\Psi(\Ameasure(d,j,s, p(d,j,s)))&=\frac{\Pi_j(s)\Psi(d)}{\|\Pi_j(s)\Psi(d)\|}
\end{align}
for all~$d\in \mathsf{Desc}_n$,~$j\in [n]$,~$s\in \{0,1\}$,
with~$p(d,j,s)=\|\Pi_j(s)\Psi(d)\|^2$.
\end{enumerate}
The output of both $\Aevolve$ and $\Apostmeasure$ is a good description for any input. 
\end{theorem}

We argue that descriptions of relevant initial states can be obtained efficiently. 
Clearly, this is the case for any state of the form~$\ket{\Psi}=U_{R_L}\cdots U_{R_1}\ket{0_F}$ obtained by applying a sequence~$\{R_j\}_{j\in [L]}\subset \mathsf{Gen}(O(2n))$ of generators to the vacuum state~$\ket{0_F}$: Here we can use   the algorithm~$\Aevolve$~$L$~times, producing a description of~$\ket{\Psi}$ in time~$O(Ln^3)$.

We we will at times need a description of a state~$\ket{\Psi}$ but do not require  fixing its global phase. This is the case for example when subsequent computational states only involve phase-insensitive expressions, e.g., terms  of the form~$|\langle \Psi,\Phi\rangle|^2$.  Such a description can be found efficiently from the covariance matrix~$\cov$ of~$\ket{\Psi}$. 
Since the phase can be fixed arbitrarily, the problem here is to find~$x\in\sbin^n$ such that~$\langle x, \Psi \rangle$ is non-zero.

\begin{theorem}\label{thm:initialstates}
There is an algorithm~$\Adescribe:\mathsf{Mat}_{2n\times 2n}(\mathbb{R})\rightarrow \mathsf{Desc}_n$ with runtime~$O(n^3)$ such
that for any covariance matrix~$\cov$, the state~$\Psi(\Adescribe(\cov))$ is a Gaussian state with covariance matrix~$\cov$, and $\Adescribe(\cov)$ is a good description. 
\end{theorem}
\noindent For example, consider states of the form~$\ket{\Phi(\pi,y)}=U_{R_\pi}\ket{y}$, where~$R_\pi\in O(2n)$ is a permutation matrix specified by an element~$\pi\in S_{2n}$ and~$y\in \{0,1\}^n$. (Such states are used in Ref.~\cite{bravyiComplexityQuantumImpurity2017a} to give a fast norm estimation algorithm, see Section~\ref{sec:fastnormestimation}.) The covariance matrix of this state is~$\cov(\pi,y)=R_\pi \cov(\ket{y})R_\pi^T$ (with~$\cov(\ket{y})$ defined by Eq.~\eqref{eq:covariancematrixMx}). We thus conclude that 
$\ket{\Psi(\Adescribe(\cov(\pi,y)))}$ is proportional to~$\ket{\Phi(\pi_,y)}$ with a global phase~$e^{i\theta}$ possibly depending  on  the pair~$(\pi,y)$.

The remainder of this section is devoted to the proofs of Theorem~\ref{thm:maindescriptalgorithm} and Theorem~\ref{thm:initialstates}: We describe the algorithms~$\Aevolve, \Aoverlap, \Ameasprob, \Ameasure$ and~$\Adescribe$ in detail, providing pseudocode, and verify that these satisfy the desired properties.

\subsection{Subroutines \label{sec:subroutines}}

Our algorithms  make use of subroutines called~$\Asupport$, $\Arelate$, $\Aoverlaptriple$ and $\Aconvert$ which we describe here.

\begin{figure}[H]
\begin{mdframed}[
    linecolor=black,
    linewidth=0.5pt,
    roundcorner=2pt,
    backgroundcolor=white, 
    userdefinedwidth=\textwidth,
]
\begin{algorithmic}[1]
\Require$\cov\in\mathsf{Mat}_{2n\times 2n}(\mathbb{R})$ covariance matrix of a pure Gaussian state
        \Function{$\Asupport$}{$\cov$}
            \State $\cov^{(0)}\leftarrow \cov$
            \State $x\leftarrow 0^{n}\in \{0,1\}^{n}$.
            \For{$j \leftarrow 1$ to~$n$}
                \Comment{simulate a measurement of~$a_j^\dagger a_j$}
                \State $q_j \leftarrow \frac{1}{2}(I + \cov^{(j-1)}_{2j-1, 2j})$ \label{alg:qjcomput}
                \If{$q_j \geq 1/2$}                 \Comment{choose the higher-probability outcome}\label{alg:algpickhigherone}
                    \State $x_j \leftarrow 0$
                    \State $p_j \leftarrow q_j$
                \Else
                    \State $x_j \leftarrow 1$
                    \State $p_j \leftarrow 1-q_j$
                 \EndIf\label{alg:pickhighertwo}
                 \State{$\cov^{(j)}\leftarrow 0\in \mathsf{Mat}_{2n\times 2n}(\mathbb{R})$} \Comment{covariance matrix of the post-measurement state}\label{alg:Mcomputeone}
                 \State{$\cov^{(j)}_{2j,2j-1}\leftarrow (-1)^{x_j}$}
                  \For{$\ell \leftarrow 1$ to~$n-1$}
                      \For{$k \leftarrow \ell+1$ to~$n$}
                         \If{$(k,\ell)\neq (2j,2j-1)$}
                            \State{$\cov^{(j)}_{k,\ell}\leftarrow \cov^{(j-1)}_{k,\ell}+\frac{(-1)^{x_j}}{2p_j} (\cov^{(j-1)}_{2j-1,\ell}\cov^{(j-1)}_{2j,k}-\cov^{(j-1)}_{2j-1,k}\cov^{(j-1)}_{2j,\ell})$}
                         \EndIf
                      \EndFor
                   \EndFor
                   \State $\cov^{(j)}\leftarrow \cov^{(j)}-(\cov^{(j)})^T$ \label{alg:Mcomputetwo}
            \EndFor
            \State \textbf{return}~$x$
        \EndFunction
\end{algorithmic}
\end{mdframed}
\caption{The algorithm~$\Asupport$: Given the covariance matrix~$\cov$ of a Gaussian state~$\ket{\Psi}$, it computes~$x\in \{0,1\}^n$ such that~$\langle x,\Psi\rangle\neq 0$.\label{fig:Asupport}}
\end{figure}

The subroutine~$\Asupport$ takes as input the covariance matrix~$\cov$ of a Gaussian state~$\Psi$ and produces a string~$x\in \{0,1\}^{2n}$
with the property that~$\langle x,\Psi\rangle\neq 0$. It is given in Fig.~\ref{fig:Asupport}. It has the following properties:

\begin{lemma}
\label{lem:Asupport}
The algorithm~$\Asupport:\mathsf{Mat}_{2n\times 2n}(\mathbb{R})\rightarrow \{0,1\}^{n}$ runs in time~$O(n^3)$. It satisfies
\begin{align}
|\langle \Asupport(\cov(\Psi)),\Psi\rangle|^2\geq 2^{-n}\qquad\textrm{ for every }\qquad \Psi\in\cG_n\ ,\label{eq:lowerboundMPsin}
\end{align}
where~$\cov(\Psi)$ is the covariance matrix of~$\Psi$.
\end{lemma}
\begin{proof}
The main idea of the algorithm is to mimic a measurement in the number state basis executed in a sequential manner. Consider the following process: Suppose we start  with the state~$\Psi^{(0)}=\Psi$, and then 
measure~$a_j^\dagger a_j$ successively for~$j=1,\dots,n$. Let~$P(x_j|x_1\cdots x_{j-1})$ denote the conditional probability of observing the outcome~$x_j\in \{0,1\}$ (when measuring~$a_j^\dagger a_j)$, given that the previous measurements yielded~$(x_1,\ldots,x_{j-1})$. According to Born's rule, this is given by 
\begin{align}
P(x_j|x_1\cdots x_{j-1})&=\langle \Psi^{(j-1)}_{x_1\cdots x_{j-1}},\Pi_j(x_j) \Psi^{(j-1)}_{x_1\cdots x_{j-1}}\rangle
\end{align}
where~$\Psi^{(j-1)}_{x_1\cdots x_{j-1}}$ is the post-measurement state after the first~$(j-1)$~measurements. The probability of observing the sequence~$x\in \{0,1\}^n$ of outcomes then is  
\begin{align}
|\langle x,\Psi\rangle|^2 &=\prod_{j=1}^n P(x_j|x_1\cdots x_{j-1})\label{eq:bayesianrule}
\end{align}
by Bayes' rule.

The algorithm~$\Asupport$ simulates this process: For each~$j\in [n]$, the quantity~$q_j$ computed in line~\ref{alg:qjcomput} is equal to the conditional probability~$P(0|x_{1}\cdots x_{j-1})$ that the~$j$-th measurement results in the outcome~$0$.  Lines~\ref{alg:algpickhigherone}--\ref{alg:pickhighertwo} ensure that the outcome~$x_j\in\{0,1\}$ with higher probability of occurrence is selected at each step, guaranteeing Property~\eqref{eq:lowerboundMPsin} (because of Eq.~\eqref{eq:bayesianrule}).  The matrix~$\cov^{(j)}$ computed in steps~\ref{alg:Mcomputeone}--\ref{alg:Mcomputetwo} is the covariance matrix of the post-measurement state 
$\Psi^{(j)}_{x_1\cdots x_{j}}$.

Each measurement is thus realized in time~$O(n^2)$ yielding the overall complexity of~$O(n^3)$.  
\end{proof}

The algorithm~$\Arelate$ is more straightforward: Given~$x,y\in \{0,1\}^n$,  it outputs~$(\alpha,\vartheta)\in \{0,1\}^{2n}\times\mathbb{R}$ such that~$c(\alpha)\ket{x}=e^{i\vartheta}\ket{y}$. That is, it finds a Majorana monomial~$c(\alpha)$ which maps the basis state~$\ket{x}$ to~$\ket{y}$ up to a phase and computes the corresponding phase.
In Fig.~\ref{fig:A2pseudocode} we give pseudocode for this algorithm.
\begin{figure}[H]
\begin{mdframed}[
    linecolor=black,
    linewidth=0.5pt,
    roundcorner=2pt,
    backgroundcolor=white, 
    userdefinedwidth=\textwidth,
]
\begin{algorithmic}[1]
\Require{$x\in \{0,1\}^n$, $y\in \{0,1\}^n$}
        \Function{$\Arelate$}{$x$,$y$}
           \State $\alpha\leftarrow 0^{2n}\in \{0,1\}^{2n}$.
            \For{$j \leftarrow 1$ to~$n$}
               \State{$\alpha_{2j-1}\leftarrow x_j\oplus y_j$}\label{alg:alphadefinitionline}
            \EndFor
            \State $\vartheta \leftarrow \frac{\pi}{4}|x\oplus y|+\pi\sum_{j=1}^n (x\oplus y)_j \eta_j(x)$\label{alg:varthetadef}
            \State \textbf{return}~$(\alpha, \vartheta)$ 
            \Comment{$(\alpha, \vartheta)$ is such that $c(\alpha)\ket{x}=e^{i\vartheta}\ket{y}$}
        \EndFunction
\end{algorithmic}
\end{mdframed}
\caption{Given~$x,y \in \{0,1\}^n$, the algorithm~$\Arelate$ computes~$\alpha\in \{0,1\}^{2n}$ and~$\vartheta\in\mathbb{R}$ such that 
$c(\alpha)\ket{x}=e^{i\vartheta}\ket{y}$.
 The expression~$\eta_j(x)$ in line~\ref{alg:varthetadef} is defined by Eq.~\eqref{eq:etajdef}, and~$x\oplus y\in \{0,1\}^n$ denotes the bitwise addition modulo two for~$x,y\in \{0,1\}^n$.  }
\label{fig:A2pseudocode}
\end{figure}
\begin{lemma}
The algorithm~$\Arelate:\{0,1\}^{n}\rightarrow\{0,1\}^{2n}\times\mathbb{C}$ runs in time~$O(n)$ and satisfies
\begin{align}
c(\alpha)\ket{x}=e^{i\vartheta}\ket{y}
\textrm{ where }
(\alpha,\vartheta)=\Arelate(x,y)\, \quad\textrm{ for all }\quad x,y\in \{0,1\}^n\ .
\end{align}
\end{lemma}

\begin{proof}
Let~$x,y\in \sbin^n$ be arbitrary.
Define 
\begin{align}
\alpha_{2j-1}&=x_j\oplus y_j\qquad\textrm{ and }\qquad \alpha_{2j}=0\qquad\textrm{ for }\qquad j\in[n]\ ,
\end{align}
as in line~\ref{alg:alphadefinitionline} of algorithm~$\Arelate$. Then
\begin{align}
c(\alpha) \ket{y}&=i^{|\alpha|\cdot (|\alpha|-1)/2} c_1^{x_1\oplus y_1}c_{3}^{x_2\oplus y_2}\cdots c_{2n-1}^{x_n\oplus y_n} \ket{y}\\
&=i^{\left(\sum_{j=1}^n x_j\oplus y_j\right)\cdot \left(\left(\sum_{j=1}^n x_j\oplus y_j\right)-1\right)/2}
(-1)^{\sum_{j=1}^n (x_j\oplus y_j) \eta_j(x)} \ket{y\oplus (x\oplus y)}\\
&=i^{|x\oplus y|\cdot (|x\oplus y|-1)/2} (-1)^{\sum_{j=1}^n (x\oplus y)_j\eta_j(x)}\ket{x}
\end{align}
where in the second identity,  we used that 
\begin{align}
c_{2j-1}\ket{x} &=(-1)^{\eta_j(x)}\ket{x\oplus e_j}
\qquad\textrm{ for all }\qquad x\in \sbin^n\textrm{ and } j\in[n]
\end{align}
because of Eq.~\eqref{eq:ajOnState}.  Because~$i^{|x\oplus y|\cdot (|x\oplus y|-1)/2} (-1)^{\sum_{j=1}^n (x\oplus y)_j \eta_j(x) }=e^{i\vartheta}$ for
\begin{align}
 \vartheta=\frac{\pi}{4} |x\oplus y|\cdot (|x\oplus y|-1)+\pi\sum_{j=1}^n (x\oplus y)_j \eta_j(x)\ ,
\end{align}
comparison with line~\ref{alg:varthetadef} of the algorithm~$\Arelate$ gives the claim. 
\end{proof}

The algorithm~$\Aoverlaptriple$ takes covariance matrices~$\cov_0,\cov_1,\cov_2$ of three Gaussian states $\Phi_0,\Phi_1,\Phi_2$ of the same parity~$\sigma\in\{-1,1\}$ and~$\alpha\in \sbineven^n$, as well as overlaps~$u=\langle \Phi_0,\Phi_1\rangle$, $v=\langle \Phi_1,c(\alpha)\Phi_2\rangle$ (which both have to be non-zero), and computes the overlap~$\langle \Phi_2,\Phi_0\rangle$. It is obtained by direct application of the formula~\eqref{eq:threestatemajoranaexpr}. For completeness, we include pseudocode
in Fig.~\ref{fig:Aoverlaptriple}. Since this algorithm involves computing Pfaffians of matrices that have size linear in~$n$, its runtime is~$O(n^3)$. 

\begin{figure}[H]
\begin{mdframed}[
    linecolor=black,
    linewidth=0.5pt,
    roundcorner=2pt,
    backgroundcolor=white, 
    userdefinedwidth=\textwidth,
]
\begin{algorithmic}[1]
\Require{$\cov_j$ covariance matrix of a Gaussian state~$\Phi_j$ for~$j=0,1,2$}
\Require{$\sigma=\Pf(\cov_0)=\Pf(\cov_1)=\Pf(\cov_2)$}
\Require{$\alpha\in \sbin^n_+$ such that~$\langle \Phi_1,c(\alpha)\Phi_2\rangle \neq 0$}
\Require{$\langle \Phi_0,\Phi_1\rangle\neq 0$}
\Require{$u=\langle \Phi_0,\Phi_1\rangle$ and~$v=\langle \Phi_1,c(\alpha)\Phi_2\rangle$}
        \Function{$\Aoverlaptriple$}{$\cov_0,\cov_1,\cov_2,\alpha,u,v$}
        \State{$R_\alpha\in \mathsf{Mat}_{(6n+|\alpha|)\times (6n+|\alpha|)}(\mathbb{R})$}
         \State $R_\alpha\leftarrow \textrm{ evaluate Eq.~\eqref{eq:Ralphadefinitionm}}$
         \State $o\leftarrow u^{-1}\cdot v^{-1}\cdot \sigma\cdot 4^{-n}i^{n+|\alpha|\cdot (|\alpha|-1)/2}\Pf(R_\alpha)$
         \Comment{compute the overlap $\langle \Phi_2, \Phi_0 \rangle$}
         \State \textbf{return}~$o$ 
        \EndFunction
\end{algorithmic}
\end{mdframed}
\caption{The algorithm~$\Aoverlaptriple$ takes as input the covariance matrices~$\cov_j$ of three Gaussian states~$\Phi_j$, $j=0,1,2$ with identical parity, $\alpha\in \sbineven^n$  and  the overlaps 
$\langle \Phi_0,\Phi_1\rangle$, $\langle \Phi_1,c(\alpha)\Phi_2\rangle$. The latter have to be non-zero. 
The algorithm computes the overlap~$\langle \Phi_2,\Phi_0\rangle$ using Eq.~\eqref{eq:threestatemajoranaexpr}. 
\label{fig:Aoverlaptriple} }
\end{figure}

\noindent Fig.~\ref{fig:aoverlaptriple} 
gives a graphical representation of what the algorithm~$\Aoverlaptriple$ achieves. These graphical representations will be helpful to construct and analyze  other algorithmic building blocks.

\begin{figure}[h]
\centering
\begin{subfigure}[t]{0.48\textwidth}
\centering
    \includegraphics[height=3cm]{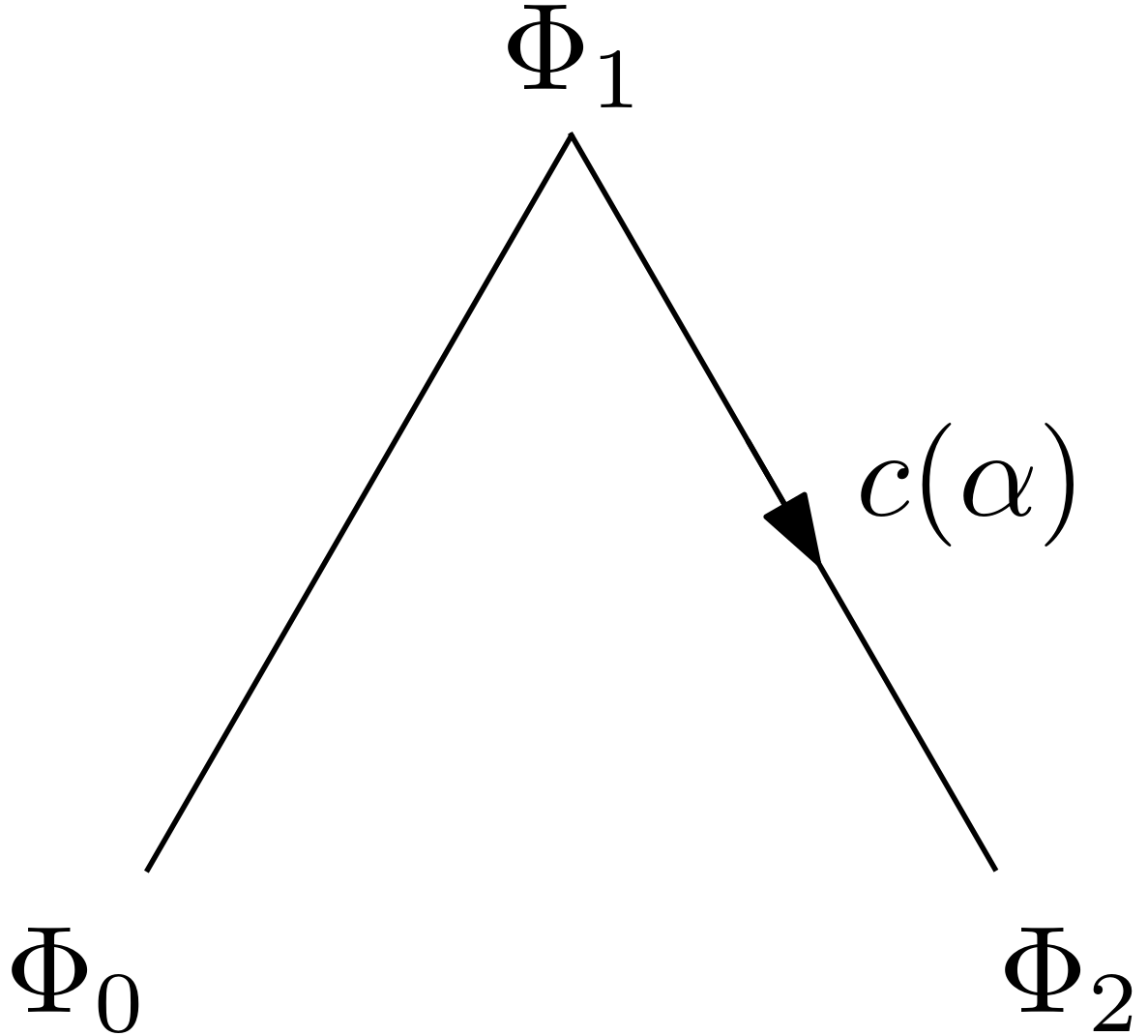}
    \caption{The input to the algorithm~$\Aoverlaptriple$ 
 consists of (descriptions of) three Gaussian states~$\Phi_0,\Phi_1,\Phi_2$
 and~$\alpha\in \{0,1\}^{2n}$, together with
 overlaps~$u=\langle \Phi_0,\Phi_1\rangle$ and~$v=\langle \Phi_1,c(\alpha)\Phi_2\rangle$ that are both non-zero.  
    }
    \label{fig:firstoverlaptriple}
\end{subfigure}
\hfill
\begin{subfigure}[t]{0.48\textwidth}
\centering
\includegraphics[height=3cm]{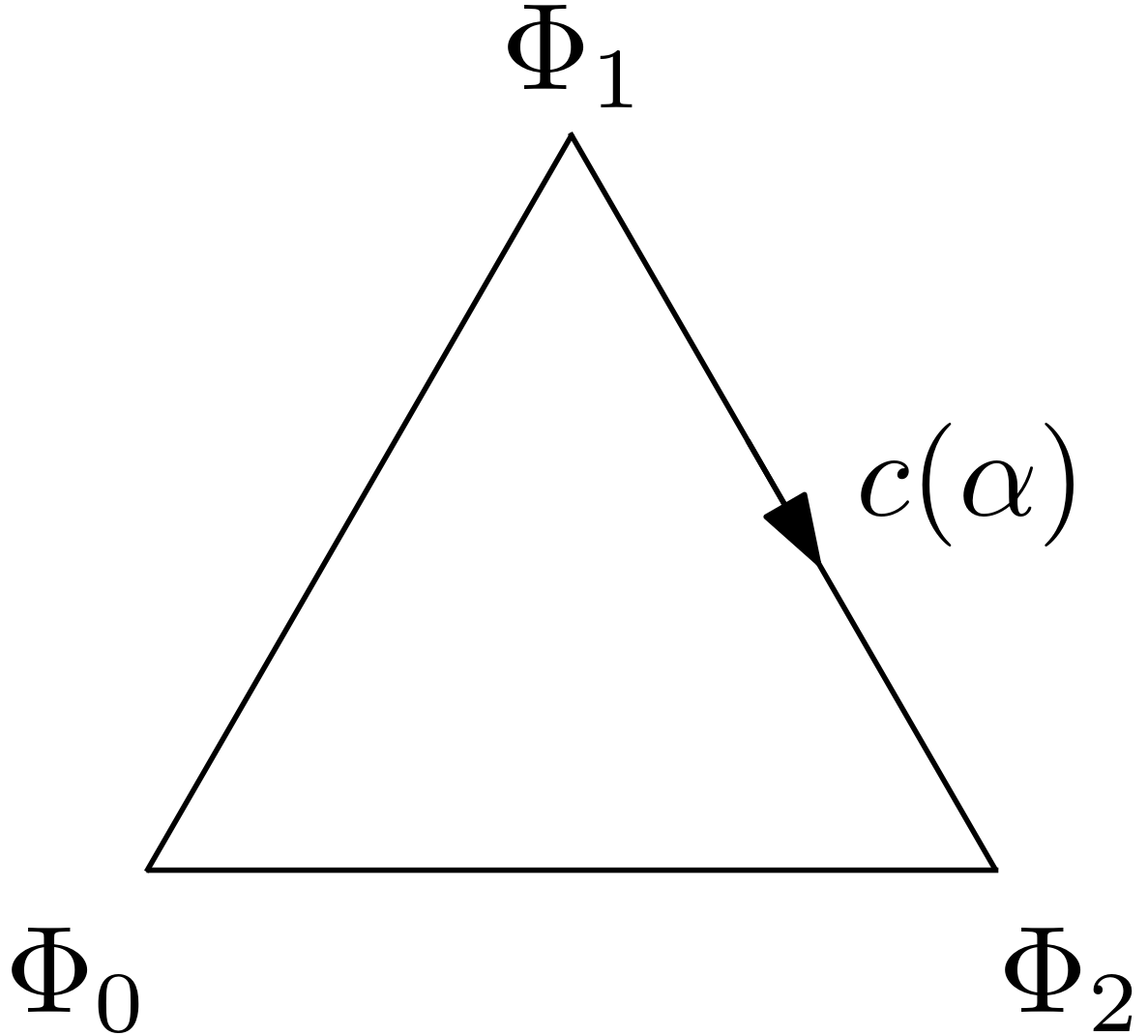}
    \caption{Applying~$\Aoverlaptriple$ provides the inner product~$\langle \Phi_0,\Phi_2\rangle$. In this diagrammatic representation, this completes the triangle   with vertices~$\Phi_0,\Phi_1,\Phi_2$. 
    }
    \label{fig:secondoverlaptriple}
\end{subfigure}

\caption{A graphical representation of the functionality provided by the algorithm~$\Aoverlaptriple$. Solid lines represent inner products that are given / have been computed, and are non-zero. Inner products of the form~$\langle \Phi_1,c(\alpha)\Phi_2\rangle$ are represented by arrows.}
\label{fig:aoverlaptriple}
\end{figure}

The algorithm~$\Aconvert$ takes a description~$d=(\cov,x,r)$ of  a Gaussian state~$\Psi(d)$ and~$y\in \{0,1\}^n$ such that~$\langle y,\Psi(d)\rangle\neq 0$, and outputs a description~$d'=(\cov,y,s)$ of the same state. In other words, it converts a description~$d$ of the state involving the reference state~$\ket{x}$ to a description~$d'$ of the same state but involving a different reference state~$\ket{y}$. In Fig.~\ref{fig:Aconvert} we give pseudocode for this algorithm.
\begin{figure}[H]
\begin{mdframed}[
    linecolor=black,
    linewidth=0.5pt,
    roundcorner=2pt,
    backgroundcolor=white, 
    userdefinedwidth=\textwidth,
]
\begin{algorithmic}[1]
\Require{$d=(\cov,x,r)\in\desc_n$, $y\in\sbin^n$ such that~$\langle y,\Psi(d)\rangle\neq 0$}
        \Function{$\Aconvert$}{$d$,$y$}
      \State{$(\alpha,\vartheta)\leftarrow \Arelate(y,x)$\label{alg:alphathetadef}}
      \Comment{find~$(\alpha, \vartheta)$ such that~$c(\alpha)\ket{y} = e^{i\vartheta} \ket{x}$}
       \State{$\cov_0\leftarrow \cov$, $\cov_1\leftarrow \cov(\ket{x})$, $\cov_2\leftarrow \cov(\ket{y})$\label{alg:triplematricesAconvert}}
       \Comment{covariance matrices of~$\Psi(d)$, $\ket{x}$ and~$\ket{y}$}
          \State{$u\leftarrow\overline{r}$} 
          \Comment{$u=\langle \Psi(d) , x\rangle$}
          \State{$v\leftarrow e^{i\vartheta}$}
                    \Comment{$v=\langle x, c(\alpha)y\rangle$}
           \State $w\leftarrow \Aoverlaptriple(\cov_0,\cov_1,\cov_2,\alpha,u,v)$ \label{alg:triplematricess}
           \Comment{compute the overlap~$\langle y, \Psi(d) \rangle$}
            \State \textbf{return}~$(\cov,y,w)$ 
        \EndFunction
\end{algorithmic}
\end{mdframed}
\caption{The algorithm~$\Aconvert$ takes a description~$d\in\desc_n$  and~$y\in \{0,1\}^n$ such that~$\langle y,\Psi(d)\rangle\neq 0$. It outputs a description~$d'\in\desc_n$ of~$\Psi(d)$ such that the second entry of~$d'$ is equal to~$y$, i.e., $d'=(\cov,y,s)$.
It makes use of the subroutines~$\Arelate$ and~$\Aoverlaptriple$. 
For~$x\in \{0,1\}^n$, $\cov(\ket{x})$ denotes the covariance matrix of the state~$\ket{x}$, see Eq.~\eqref{eq:covariancematrixMx}. 
\label{fig:Aconvert}}
\end{figure}
\noindent The algorithm~$\Aconvert$ is illustrated in Fig.~\ref{fig:aoverlapillustration}.
\begin{figure}[h]
\centering
\begin{subfigure}[t]{0.31\textwidth}
\centering
    \includegraphics[height=3cm]{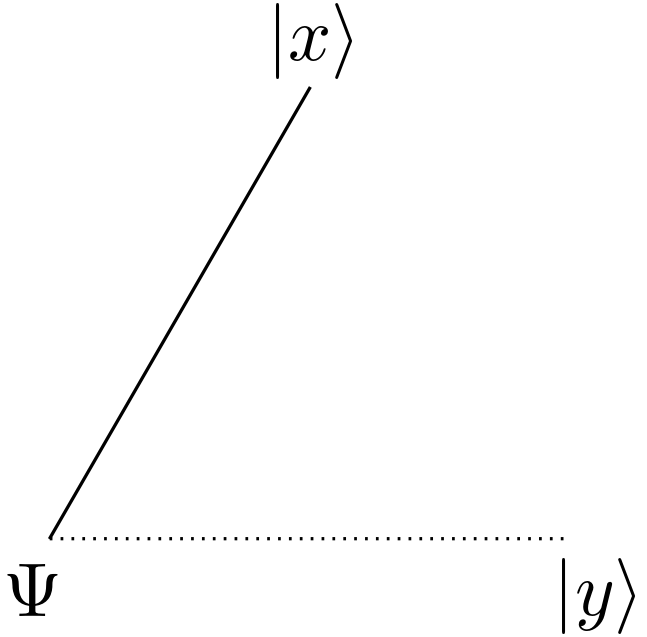}
    \caption{The input to the algorithm~$\Aconvert$ specifies a Gaussian state~$\Psi$, $x\in \{0,1\}^n$ such that~$\langle x,\Psi\rangle\neq 0$, the value~$r=\langle x,\Psi\rangle$ and an element~$y\in \{0,1\}^n$ such that~$\langle y,\Psi\rangle\neq 0$.  The value of~$\langle y,\Psi\rangle$ is not given. 
    }
    \label{fig:converta}
\end{subfigure}
\hfill
\begin{subfigure}[t]{0.31\textwidth}
\centering
\includegraphics[height=3cm]{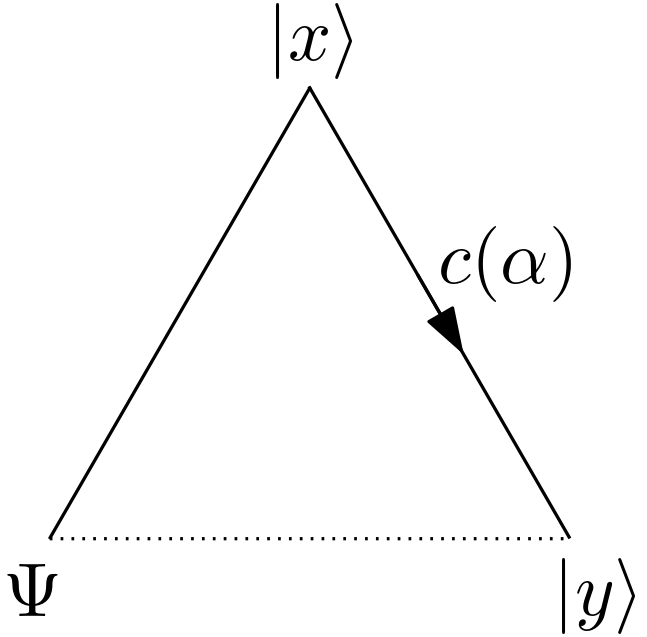}
    \caption{The algorithm  applies the subroutine~$\Arelate$ to find~$(\alpha,\vartheta)$ such that $c(\alpha)\ket{y}=e^{i\vartheta}\ket{x}$. In particular, after this step, the value $\langle x,c(\alpha)y\rangle=e^{i\vartheta}$ is known and it is non-zero. 
    }
    \label{fig:convertb}
\end{subfigure}
\hfill
\begin{subfigure}[t]{0.31\textwidth}
\centering
\includegraphics[height=3cm]{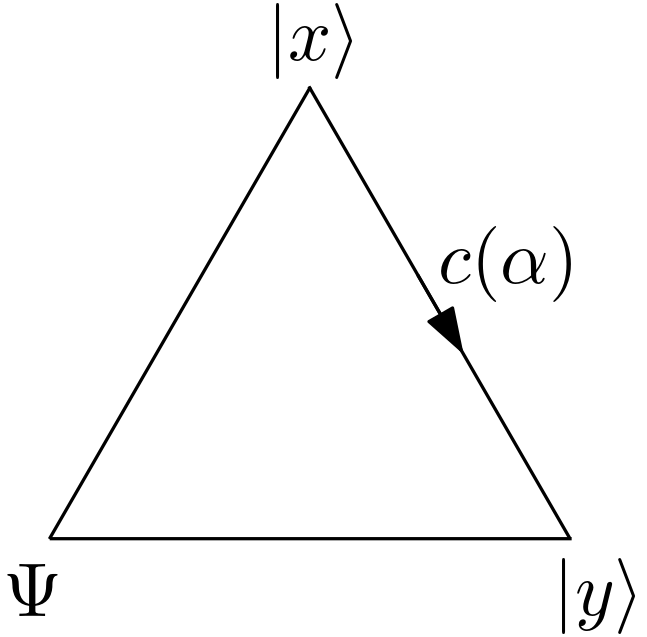}
    \caption{The algorithm then applies the subroutine~$\Aoverlaptriple$ to compute~$w=\langle y, \Psi \rangle$. The triple $(\Gamma,y,w)$ is a valid description of~$\Psi$.}
    \label{fig:convert}
\end{subfigure}       
\caption{An illustration of the algorithm~$\Aconvert$. Dotted lines represent inner products that are non-zero, but that are not provided / have not yet been computed by the algorithm.}
\label{fig:aoverlapillustration}
\end{figure}

\begin{lemma}
The algorithm~$\Aconvert:\desc_n\times \{0,1\}^n\rightarrow\desc_n$ given in Fig.~\ref{fig:Aconvert} runs in time~$O(n^3)$. 
Assume that~$d\in\desc_n$ and~$y\in\sbin^n$ satisfy~$\langle y, \Psi(d)  \rangle \neq 0$. 
Then
\begin{align}
\label{eq:firstpartclaimaconvert0}
    \Psi(\Aconvert(d,y)) = \Psi(d) \ .
\end{align}
Furthermore, denoting the output of~$\Aconvert(d,y)$ by~$d'=(\cov',y',s')$ we have
\begin{align}
y'=y\label{eq:firstpartclaimaconvert}
\end{align}
as well as
\begin{align}
|s'|^2 =  |\langle y', \Psi(d) \rangle|^2 = |\langle y, \Psi(d) \rangle|^2 \ . \label{eq:normpreservationclaimaconvert}
\end{align}
\end{lemma}
\begin{proof}
Let us denote the input to~$\Aconvert$ by~$(d,y)$, where~$d=(\cov,x,r)\in \desc_n$ and~$y\in \{0,1\}^n$. 
Then 
\begin{align}
\langle x,\Psi(d)\rangle \neq 0\label{eq:xpsideneqzero}
\end{align}
since~$d$ is a description of~$\Psi(d)$. 
Furthermore, for~$(\alpha,\vartheta)$ as defined in line~\ref{alg:alphathetadef} we have
\begin{align}
\langle x,c(\alpha)y\rangle =e^{i\vartheta}\neq 0\label{eq:varthetaxalphay}
\end{align}
by definition of the algorithm~$\Arelate$.
In line~\ref{alg:triplematricesAconvert} of~$\Aconvert$, the matrices~$\cov_j$, $j\in[3]$ are the covariance matrices of the states
\begin{align}
(\Phi_0,\Phi_1,\Phi_2)&=(\Psi(d),\ket{x},\ket{y})\ .
\end{align}
We note that Eq.~\eqref{eq:xpsideneqzero} and the assumption~$\langle y, \Psi(d)  \rangle \neq 0$ imply that these three states have identical parity. The value~$w$ computed in line~\ref{alg:triplematricess} using~$\Aoverlaptriple$ is equal to the overlap
\begin{align}
w&=\langle \Phi_2,\Phi_0\rangle=\langle y,\Psi(d)\rangle\ ,\label{eq:sdefinitionphizerotwo}
\end{align}
because
\begin{align}
u&=\overline{r}=\overline{\langle x,\Psi(d)\rangle}=\langle \Psi(d), x \rangle=\langle \Phi_0,\Phi_1\rangle\neq 0\\
v&=e^{i\vartheta}=\langle x,c(\alpha)y\rangle=\langle \Phi_1,c(\alpha)\Phi_2\rangle\neq 0
\end{align}
by Eqs.~\eqref{eq:xpsideneqzero} and~\eqref{eq:varthetaxalphay}.  Eq.~\eqref{eq:sdefinitionphizerotwo} together with the assumption~$\langle y,\Psi(d)\rangle\ \neq 0$ show that the output~$(\cov,y,w)$ is a description of~$\Psi(d)$. This completes the proof of Eq.~\eqref{eq:firstpartclaimaconvert0}.

Eq.~\eqref{eq:normpreservationclaimaconvert} is trivially satisfied because
\begin{align}
    s'=w=\langle y,\Psi(d)\rangle\ .
\end{align}

The complexity of the algorithm is dominated by~$\Aoverlaptriple$, which takes time~$O(n^3)$. 
\end{proof}

\subsection{Computing overlaps and descriptions of evolved/measured states\label{sec:overlapevolvemeasure}}

Based on the subroutines~$\Asupport$, $\Arelate$, $\Aoverlaptriple$ and~$\Aconvert$, we can now describe our main algorithms~$\Aoverlap$, $\Aevolve$, $\Ameasprob$ and~$\Ameasure$ for overlaps, Gaussian unitary evolution, to compute the outcome probability and the post-measurement state when measuring the occupation number, respectively. We give pseudocode for each algorithm and establish the associated claims.

We give pseudocode for the algorithm~$\Aoverlap$ in Fig.~\ref{fig:Aoverlap} and we illustrate it in Fig.~\ref{fig:aoverlap}.

\begin{figure}[H]
\begin{mdframed}[
    linecolor=black,
    linewidth=0.5pt,
    roundcorner=2pt,
    backgroundcolor=white, 
    userdefinedwidth=\textwidth,
]
\begin{algorithmic}[1]
\Require{$d_1=(\cov_1,x_1,r_1), d_2=(\cov_2,x_2,r_2)\in\desc_n$}
        \Function{$\Aoverlap$}{$d_1$,$d_2$}
        \State{$\sigma_1\leftarrow \Pf(\cov_1)$}
        \Comment{ compute the parity~$\sigma_j$ of~$\Psi(d_j)$}
        \State{$\sigma_2\leftarrow \Pf(\cov_2)$}
        \If{$\sigma_1\neq \sigma_2$}\label{alg:checkequalparity}
            \State{\textbf{return} 0}
            \Comment{states with different parities have zero overlap}
        \EndIf
        \State{$(\alpha,\vartheta)\leftarrow \Arelate(x_2,x_1)$\label{alg:alphanewdef}}\label{alg:alphavarthetacomp}
        \Comment{$(\alpha, \vartheta)$ satisfies~$c(\alpha)\ket{x_2} = e^{i\vartheta} \ket{x_1}$}
       \State{$\cov'_0\leftarrow \cov_1$, $\cov'_1\leftarrow \cov(\ket{x_1})$, $\cov'_2\leftarrow \cov_2$\label{alg:triplematricesAconvertnew}}
       \Comment{covariance matrices of~$\Psi(d_1), \ket{x_1}$ and~$\Psi(d_2)$}
          \State{$u\leftarrow\overline{r}_1$}
          \Comment{$u=\langle \Psi(d_1),x_1\rangle$}
          \State{$v\leftarrow e^{i\vartheta}r_2$}
          \Comment{$v=\langle x_1,c(\alpha)\Psi(d_2)\rangle$}
           \State $w\leftarrow \Aoverlaptriple(\cov'_0,\cov'_1,\cov'_2,\alpha,u,v)$ \label{alg:triplematricesnew}
           \Comment{compute the overlap~$\langle \Psi(d_2), \Psi(d_1) \rangle$}
            \State \textbf{return}~$\overline{w}$ 
        \EndFunction
\end{algorithmic}
\end{mdframed}
\caption{The algorithm~$\Aoverlap$ takes descriptions~$d_1,d_2\in\desc_n$ 
and outputs the overlap~$\langle \Psi(d_1),\Psi(d_2)\rangle$.  
\label{fig:Aoverlap}}
\end{figure}

\begin{figure}
\centering
\begin{subfigure}[t]{0.31\textwidth}
\centering
    \includegraphics[height=3cm]{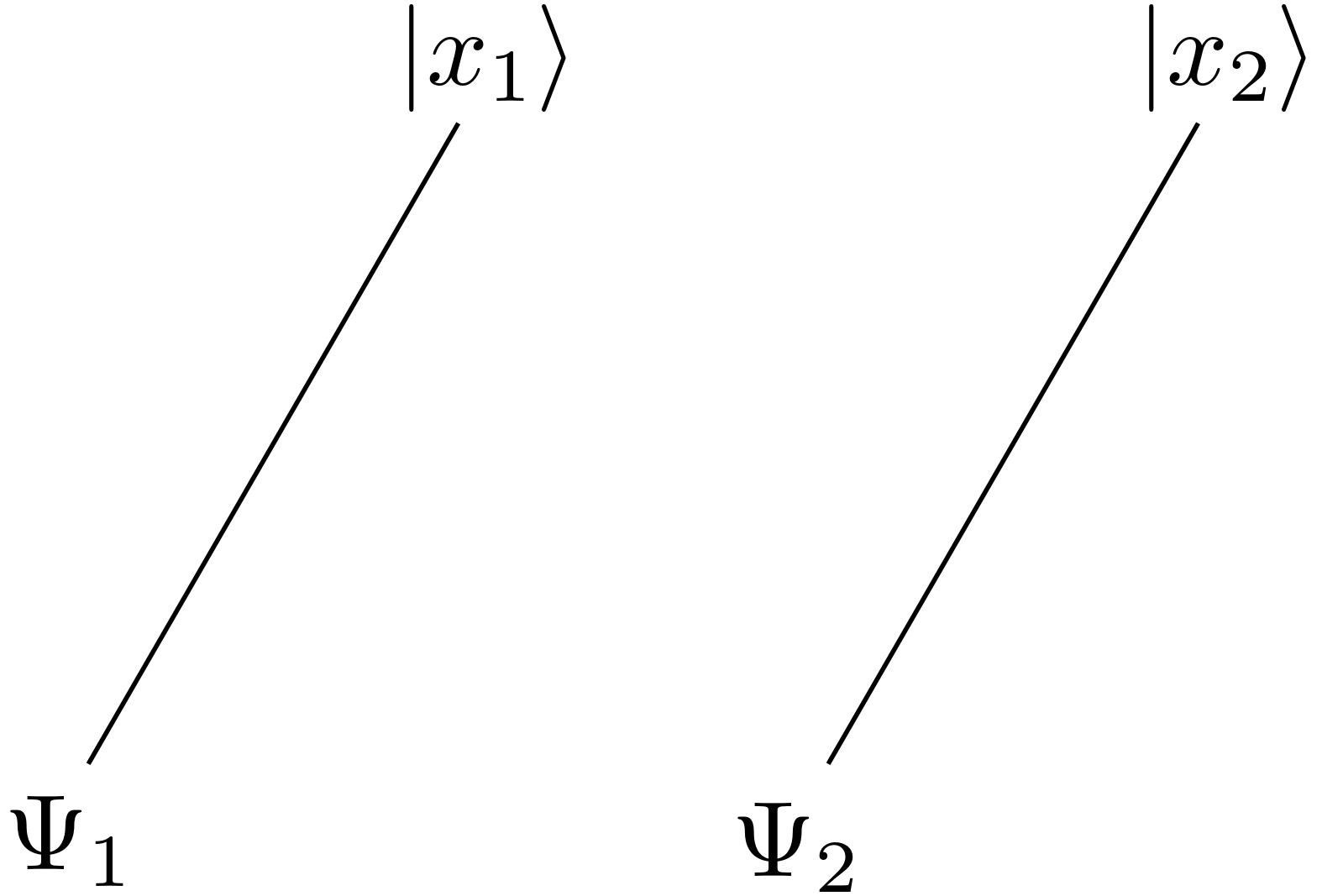}
    \caption{The input to the algorithm~$\Aoverlap$ consists of (descriptions of) two Gaussian~$\Psi_1,\Psi_2$, $x_1,x_2\in\{0,1\}^n$ and the overlaps~$r_j=\langle x_j,\Psi_j\rangle$, $j\in [2]$. The latter are both assumed to be non-zero.
    }
    \label{fig:firstoverlaptriple}
\end{subfigure}
\hfill
\begin{subfigure}[t]{0.31\textwidth}
\centering
\includegraphics[height=3cm]{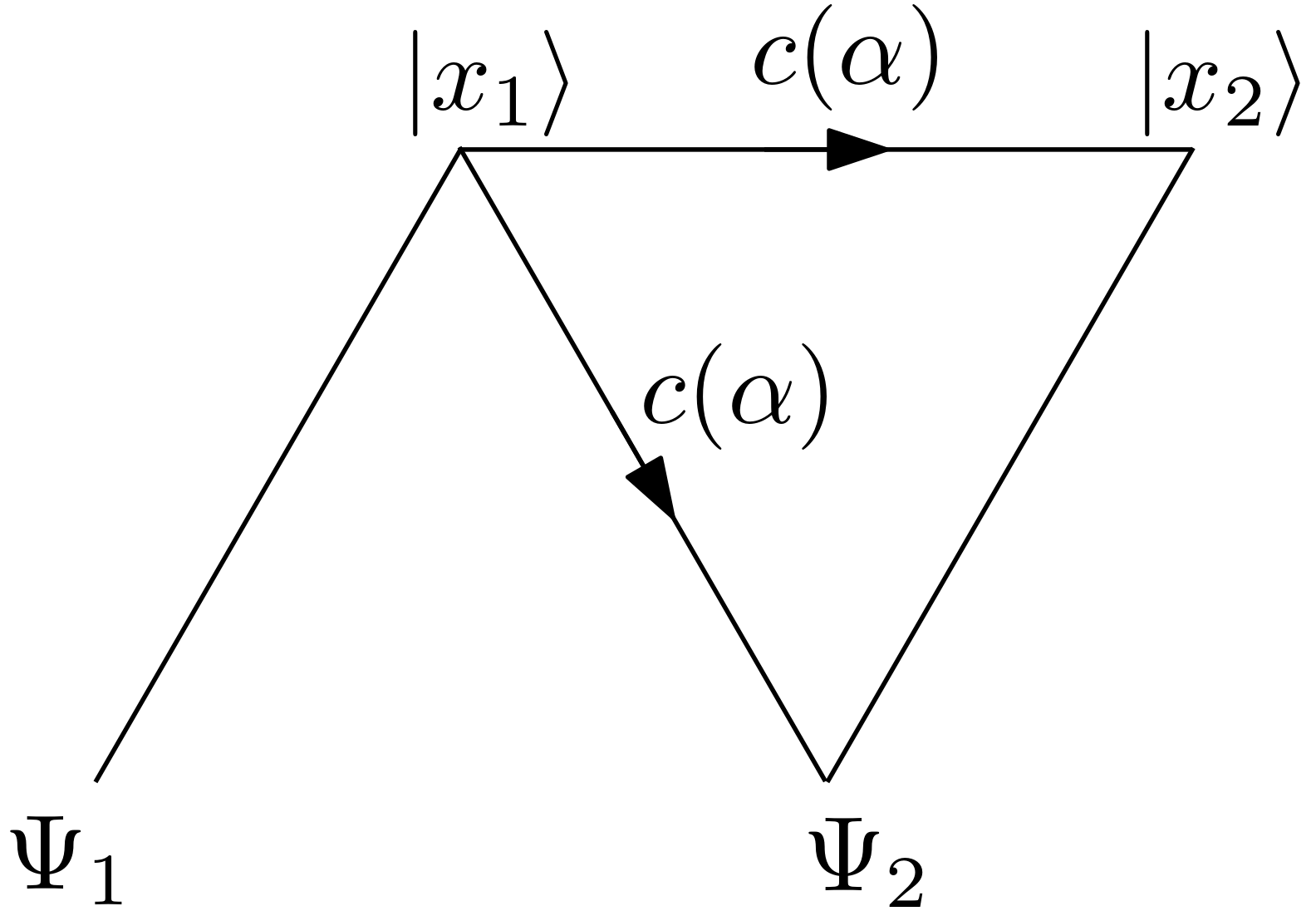}
    \caption{The algorithm uses the subroutine~$\Arelate$ 
    to find~$(\alpha,\vartheta)$ such that $c(\alpha)\ket{x_2}=e^{i\vartheta}\ket{x_1}$. In particular, this means that the value $\langle x_1,c(\alpha)x_2\rangle=e^{i\vartheta}$ is computed, and it is non-zero. Furthermore, this implies that $\langle x_1,c(\alpha)\Psi_2\rangle=\langle c(\alpha)x_1,\Psi_2\rangle=e^{i\vartheta}\langle x_2,\Psi_2\rangle=e^{i\vartheta}r_2$ is also known and non-zero.
    }
    \label{fig:secondoverlaptriple}
\end{subfigure}
\hfill
\begin{subfigure}[t]{0.31\textwidth}
\centering
\includegraphics[height=3cm]{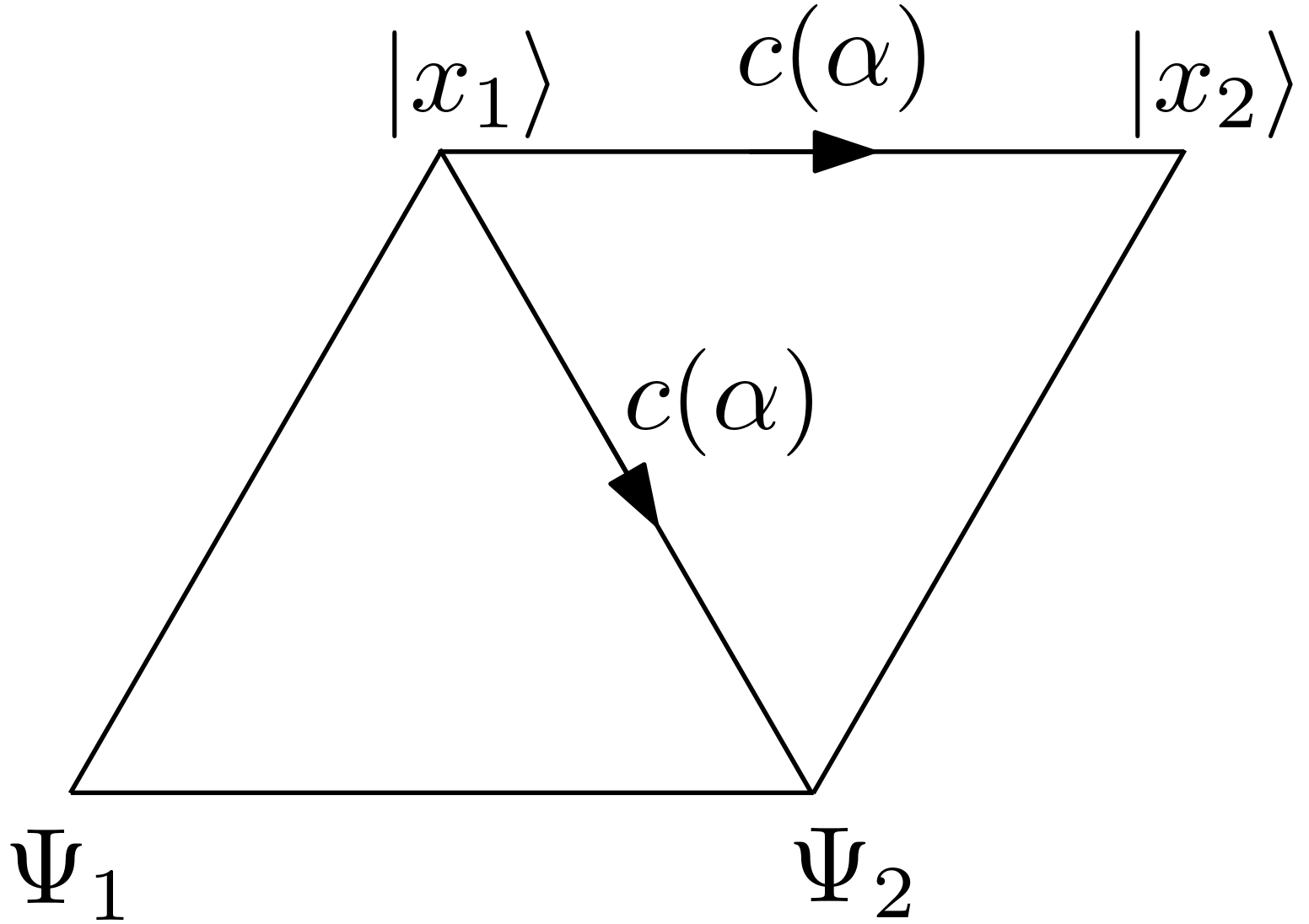}
    \caption{In the last step of the algorithm, the subroutine $\Aoverlaptriple$ is applied
    to complete a triangle: This amounts to computing $w=\langle \Psi_2,\Psi_1\rangle$. The algorithm returns the complex conjugate~$\overline{w}=\langle \Psi_1,\Psi_2\rangle$.}
    \label{fig:secondoverlaptriple}
\end{subfigure}       
\caption{An illustration of the algorithm~$\Aoverlap$.}
\label{fig:aoverlap}
\end{figure}

\begin{lemma}
\label{lem:overlap}
The algorithm~$\Aoverlap:\desc_n\times \desc_n\rightarrow\mathbb{C}$ given in Fig.~\ref{fig:Aoverlap} runs in time~$O(n^3)$. It satisfies
\begin{align}
\Aoverlap(d_1,d_2)&=\langle \Psi(d_1),\Psi(d_2)\rangle\qquad\textrm{ for all }\qquad d_1,d_2\in\desc_n\ .\label{eq:overlapidentitydonedtwo}
\end{align}
\end{lemma}
\begin{proof}
Let~$d_j=(\cov_j,x_j,r_j)\in\desc_n$ for~$j\in[2]$. Then 
\begin{align}
r_j&=\langle x_j,\Psi(d_j)\rangle\neq 0\qquad\textrm{ for }\qquad j\in[2] \ ,
\label{eq:rjdefinitionxjdj}
\end{align}
by assumption.  

Line~\ref{alg:checkequalparity} 
treats the case where~$\Psi(d_1)$ and~$\Psi(d_2)$ have different parity, and are thus orthogonal. Starting from line~\ref{alg:alphavarthetacomp}, we can hence assume that the parities~$\sigma_1, \sigma_2$ of~$\Psi(d_1),\Psi(d_2)$ are identical, $\sigma=\sigma_1=\sigma_2$. By 
Eq.~\eqref{eq:rjdefinitionxjdj}, this implies that both~$\ket{x_1}$ and~$\ket{x_2}$ also have parity~$\sigma$, that is,
\begin{align}    \parity(\ket{x_1})=\parity(\ket{x_2})=\parity(\ket{\Psi(d_1)})=\parity(\ket{\Psi(d_2)})\ .\label{eq:parityequalm}
\end{align}
By definition of~$\Arelate$, the pair~$(\alpha,\vartheta)$ computed in line~\ref{alg:alphanewdef} satisfies 
\begin{align}
c(\alpha)\ket{x_2}=e^{i\vartheta}\ket{x_1}\ .\label{eq:alphax1x2}
\end{align}
Consider the triple of states
\begin{align}
(\Phi_0,\Phi_1,\Phi_2)&=(\Psi(d_1),\ket{x_1},\Psi(d_2))\ .
\end{align}
Then the matrices~$\cov_j'$, $j\in[3]$ defined in line~\ref{alg:triplematricesAconvertnew} of the algorithm are the covariance matrices of~$\Phi_j$, $j\in \{0,1,2\}$. We have 
\begin{align}
u&=\overline{r}_1=\overline{\langle x_1,\Psi(d_1)\rangle}=\langle \Psi(d_1),x_1\rangle=\langle  \Phi_0,\Phi_1\rangle \neq 0\ ,
\end{align}
by Eq.~\eqref{eq:rjdefinitionxjdj}, and similarly
\begin{alignat}{2}
v&=e^{i\vartheta}r_2\\
&= e^{i\vartheta}\langle x_2,\Psi(d_2)\rangle\\
&=\langle e^{-i\vartheta} x_2,\Psi(d_2)\rangle\\
&=\langle c(\alpha)x_1,\Psi(d_2)\rangle&&\qquad\textrm{ with~\eqref{eq:alphax1x2}}\\
&=\langle x_1,c(\alpha)\Psi(d_2)\rangle&&\qquad\textrm{ because~$c(\alpha)$ is self-adjoint}\\
&=\langle \Phi_1,c(\alpha)\Phi_2\rangle\\
&\neq 0&&\qquad\textrm{ because~$r_2\neq 0$ by Eq.~\eqref{eq:rjdefinitionxjdj}\ .} 
\end{alignat}
Furthermore, by Eq.~\eqref{eq:parityequalm} the states~$(\Phi_0,\Phi_1,\Phi_2)$ have identical parity. 
It thus follows from the properties of~$\Aoverlaptriple$ that 
the quantity~$w$ computed in step line~\ref{alg:triplematricesnew} is equal to
\begin{align}
w&=\langle \Psi(d_2),\Psi(d_1)\rangle\ .
\end{align}
Since the output of the algorithm is the complex conjugate~$\overline{w}=\langle \Psi(d_1),\Psi(d_2)\rangle$, this  implies the claim~\eqref{eq:overlapidentitydonedtwo}.

The runtime of~$\Aoverlap$ is dominated by~$\Aoverlaptriple$, hence it is of order~$O(n^3)$. 
\end{proof}

We give pseudocode for the algorithm~$\Aevolve$ in Fig.~\ref{fig:Aevolve} and we illustrate it in Fig.~\ref{fig:aevolvefig}.

\begin{figure}[h]
\begin{mdframed}[
    linecolor=black,
    linewidth=0.5pt,
    roundcorner=2pt,
    backgroundcolor=white, 
    userdefinedwidth=\textwidth,
]
\begin{algorithmic}[1]
\Require{$R\in\mathsf{Gen}(O(2n))$}
\Require{$d=(\cov,x,r)\in\desc_n$}
        \Function{$\Aevolve$}{$R,d$}
           \State $\cov_0\leftarrow R\cov R^T$\label{alg:mzerodefm}
        \Comment{covariance matrix of~$U_R\Psi(d)$}
    \State $y\leftarrow\Asupport(\cov_0)$\label{alg:yReferenceForM0}
\Comment{gives~$y$ such that~$|\langle y,U_R\Psi(d)\rangle |^2\geq 2^{-n}$}
            \If{$R=R_{j,k}(\vartheta)$ 
 for~$j < k\in [2n], \vartheta\in\mathbb{R}$}\label{eq:computationlineone}
 \Comment{$R=R_{j,k}(\vartheta)$ is a Givens rotation}
                \If{$\cos^2(\vartheta/2) \geq 1/2$} \label{alg:gammaIfCond}
                    \State $z \leftarrow x$ \label{alg:xIsZ}
                    \State $s \leftarrow \cos(\vartheta/2)$ \label{alg:sIsCos}
                \Else
                    \State $z \leftarrow x \oplus e_j \oplus e_k$ \label{alg:xIsZOplusOplus}
                    \State $\beta \leftarrow \beta_j(x)+\beta_k(x)$ \label{alg:betakgeqj} 
                    \State $s \leftarrow e^{i\pi\beta} \sin(\vartheta/2)$ \label{alg:sIsPhaseBeta}
                    \EndIf
                \EndIf \label{eq:computationlineEndGivens} 
            \If{$R=R_j$ for~$j\in [2n]$} \label{eq:computationlineBeginReflection}
             \Comment{$R=R_j$ is a reflection}
              \State{$z\leftarrow x \oplus e_j$} \label{alg:RjX}
            \State{$\beta \leftarrow \beta_j(x)$} \label{alg:RjBeta}
              \State{$s\leftarrow e^{i\pi\beta}$} \label{alg:RjS}
            \EndIf\label{eq:computationlineend} \Comment{$(s,z)$ satisfies~$s=\langle z, U_Rx\rangle$, $|s|^2\geq 1/2$}
            \State $(\alpha,\gamma)\leftarrow \Arelate(y,z)$ \label{alg:overlapBeginFinal}
            \Comment{$(\alpha, \gamma)$ satisfies~$c(\alpha)\ket{y} = e^{i\gamma} \ket{z}$}
            \State $\cov_1\leftarrow R\cov(\ket{x})R^T$, $\cov_2\leftarrow \cov(\ket{y})$\label{alg:monetwodef}
            \Comment{covariance matrices of~$U_R\ket{x}$ and~$\ket{y}$}
            \State $u\leftarrow \overline{r}$ \label{alg:evolveU}
            \Comment{$u=\langle U_R\Psi,U_Rx\rangle$}
            \State $v\leftarrow e^{i\gamma}\overline{s}$ \label{alg:evolvev}
            \Comment{$v=\langle U_Rx,c(\alpha)y\rangle$}
            \State $w\leftarrow \Aoverlaptriple(\cov_0,\cov_1,\cov_2,\alpha,u,v)$ \label{alg:overlapEndFinal}
            \Comment{compute the overlap~$\langle y,U_R \Psi(d) \rangle$}
             \State \textbf{return} $(\cov_0,y,w)$ 
             \Comment{return a description of~$U_R \ket{\Psi(d)}$}
        \EndFunction
\end{algorithmic}
\end{mdframed}
\caption{The algorithm~$\Aevolve$ takes a description~$d\in \desc_n$ and an orthogonal matrix~$R \in \mathsf{Gen}(O(2n))$ associated with the Gaussian unitary~$U_R$ and computes a description for the state~$U_R \Psi(d)$. In this algorithm, the functions~$\beta_{s}:\sbin^n\rightarrow\mathbb{R}$ for~$s\in [n]$ are defined as~$\beta_s(x)= \eta_s(x)+\left(x_s - \frac{1}{2}\right) \cdot (s+1)$, $x\in \{0,1\}^n$ with~$\eta_s(x)$ given in Eq.~\eqref{eq:etajdef}.\label{fig:Aevolve}}
\end{figure}

\begin{figure}[h]
\centering
\begin{subfigure}[t]{0.31\textwidth}
\centering
    \includegraphics[height=3cm]{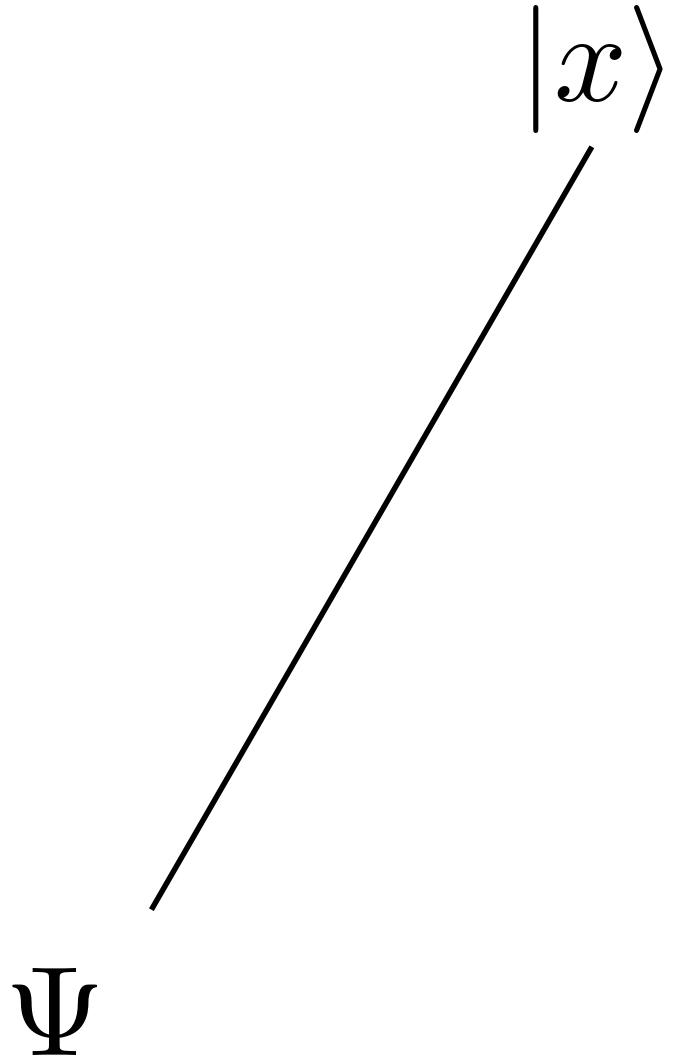}
    \caption{The input to the algorithm~$\Aevolve$  is (the description of) a Gaussian state~$\Psi$, $x\in \{0,1\}^n$ and the non-zero overlap~$r=\langle x,\Psi\rangle$, as well as an element~$R \in \mathsf{Gen}(O(2n))$ associated with a Gaussian unitary~$U_R$.}
    \label{fig:aevolvefig}
\end{subfigure}
\hfill
\begin{subfigure}[t]{0.31\textwidth}
\centering
\includegraphics[height=3cm]{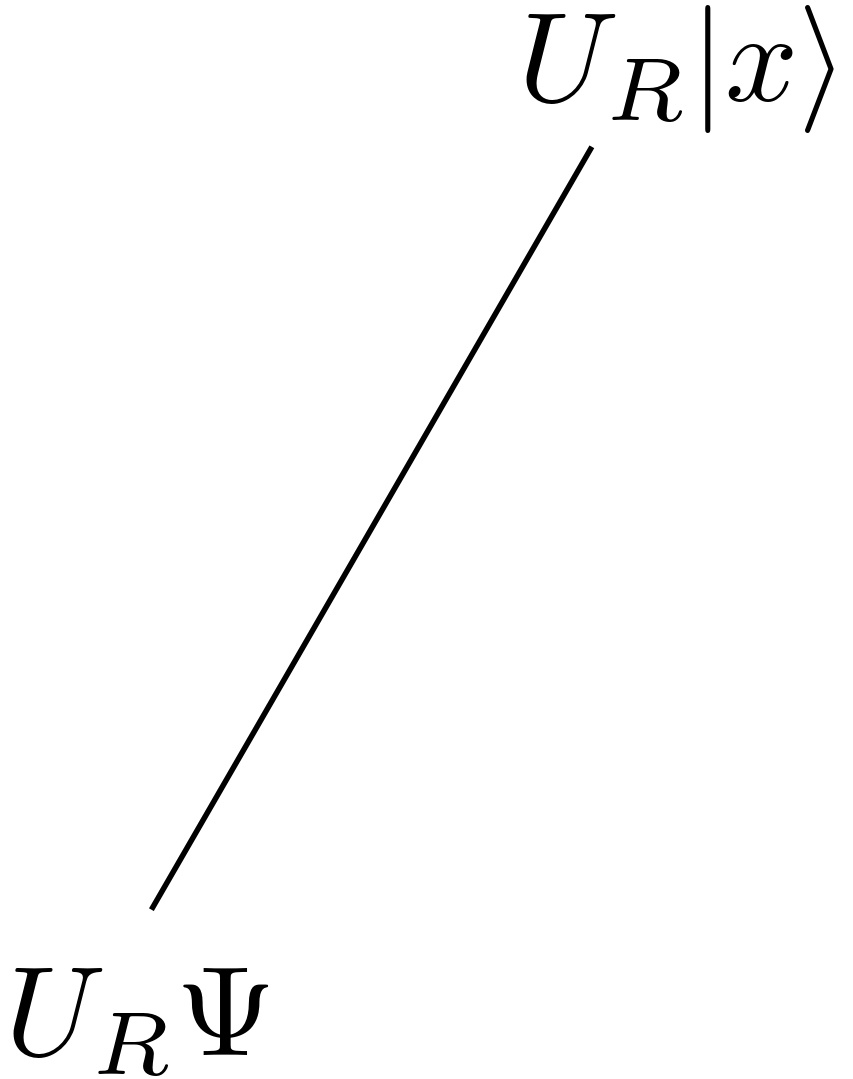}
    \caption{By unitarity of~$U_R$, the input data 
    also provides the inner product~$\langle U_Rx,U_R\Psi\rangle=\langle x,\Psi\rangle=r$, which is non-zero.  
    }
    \label{fig:aevolvefigtwo}
\end{subfigure}
\hfill
\begin{subfigure}[t]{0.31\textwidth}
\centering
\includegraphics[height=3cm]{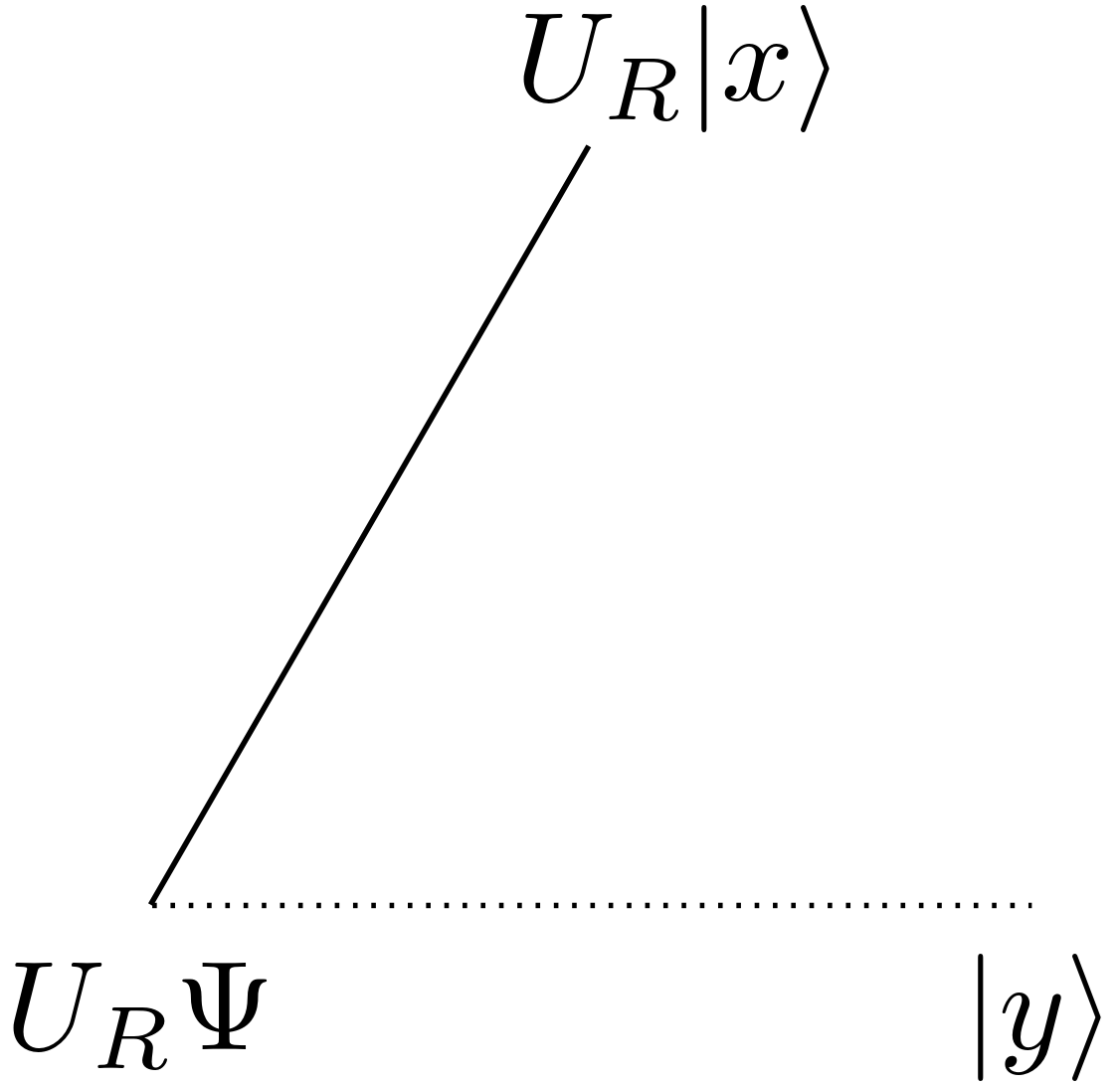}
\caption{The algorithm invokes the subroutine~$\Asupport$ applied to the covariance matrix~$\Gamma_0$ of the evolved state~$U_R\Psi$ in order to find an element~$y\in \{0,1\}^n$ such that~$\langle y,U_R\Psi\rangle\neq 0$. (The value of this inner product is not computed/available at this point.)   
    }
    \label{fig:aevolvefigthree}
\end{subfigure}\\
\vspace{4mm}
\begin{subfigure}[t]{0.31\textwidth}
\centering
    \includegraphics[height=3cm]{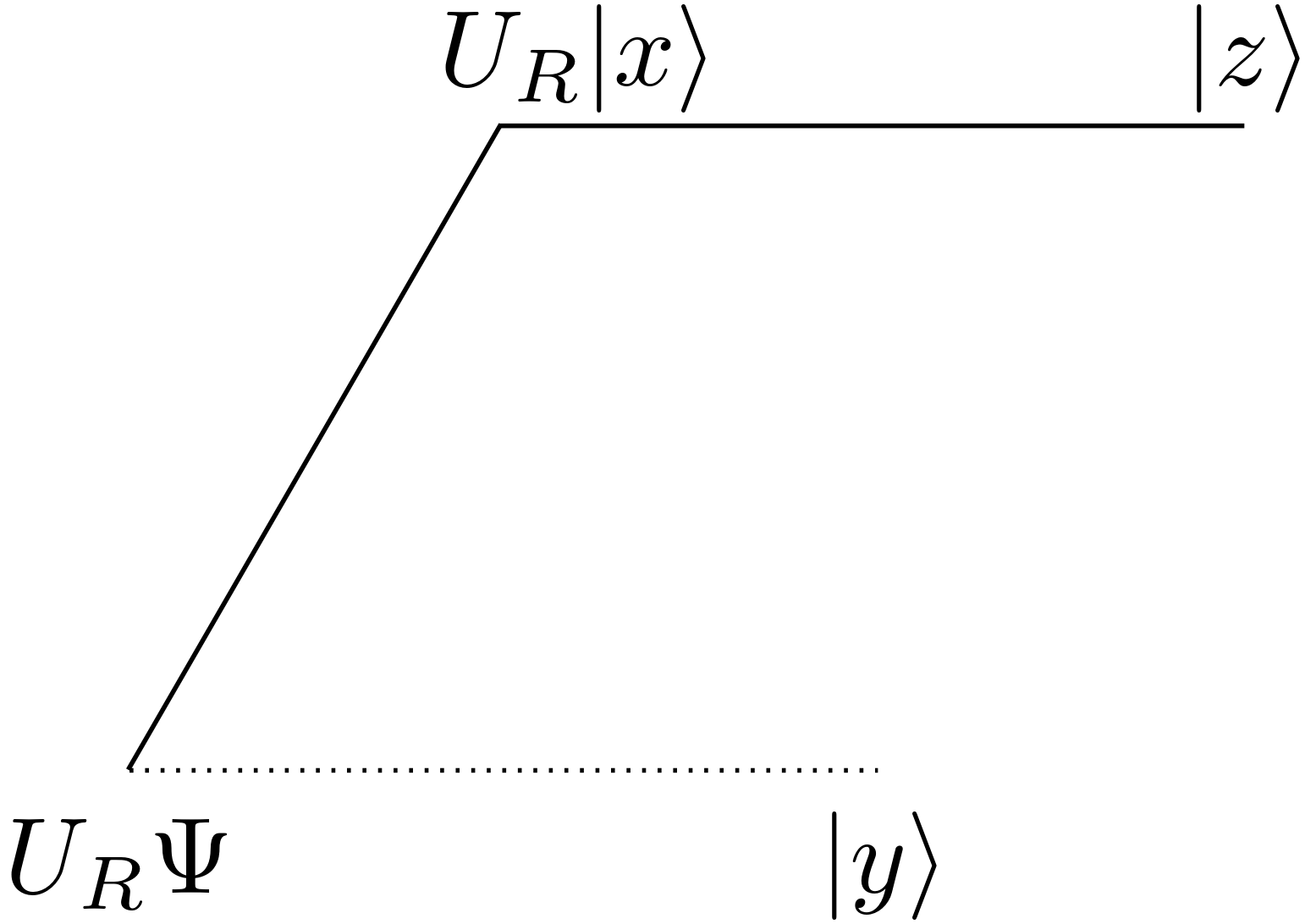}
    \caption{
    Using the fact the action of~$U_R$ for a generator $R\in \mathsf{Gen}(O(2n))$ is local, the algorithm determines an element~$z\in \{0,1\}^n$ such that~$s=\langle z,U_R x\rangle\neq 0$, and computes the value~$s$.
    }
    \label{fig:aevolvefigfour}
\end{subfigure}
\hfill
\begin{subfigure}[t]{0.31\textwidth}
\centering
\includegraphics[height=3cm]{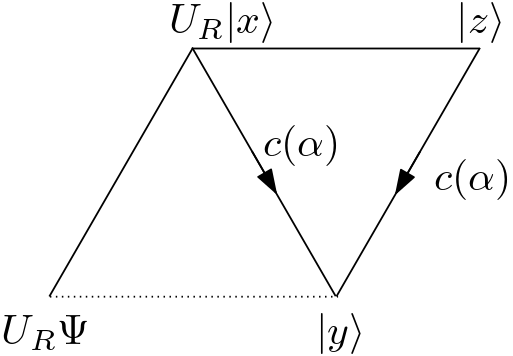}
    \caption{The algorithm then uses the subroutine~$\Arelate$ to find~$(\alpha,\gamma)$ such that $c(\alpha)\ket{y}=e^{i\gamma}\ket{z}$. This means that the inner product $\langle z,c(\alpha)y\rangle=e^{i\gamma}$ is known and non-zero. Since~$\langle z,U_R x\rangle$ is  known and non-zero (as ensured by the previous step~\eqref{fig:aevolvefigfour}), the value~$\langle U_Rx,c(\alpha)y\rangle=e^{i\gamma}\overline{\langle z,U_Rx\rangle}$ is also non-zero and can be computed.}
    \label{fig:aevolvefigfive}
\end{subfigure}
\hfill
\begin{subfigure}[t]{0.31\textwidth}
\centering
\includegraphics[height=3cm]{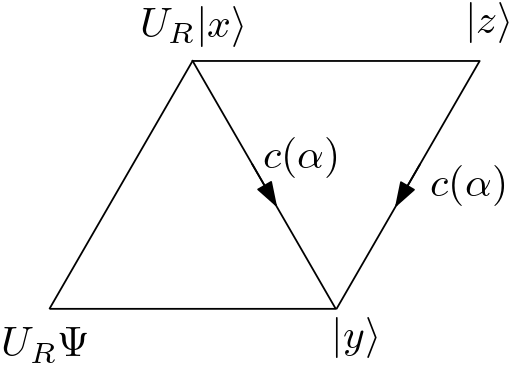}
    \caption{In the last step, the subroutine~$\Aoverlaptriple$ is used to compute the quantity~$w=\langle y,U_R\Psi\rangle$. It is non-zero by step~\eqref{fig:aevolvefigthree}. Thus~$(\Gamma_0,y,w)$ is a valid description of~$U_R\Psi$.}
    \label{fig:aevolvefigsix}
\end{subfigure}\\\
\caption{An illustration of the algorithm~$\Aevolve$. Dotted lines correspond to inner products whose value is non-zero, but has not been computed at that stage of the algorithm. }
\label{fig:aevolvefig}
\end{figure}

\begin{lemma}
    The algorithm~$\Aevolve : \mathsf{Gen}(O(2n)) \times \desc_n \rightarrow \desc_n$ given in Fig.~\ref{fig:Aevolve} runs in time~$O(n^3)$. Consider an arbitrary generator~$R\in \mathsf{Gen}(O(2n))$ and a description~$d \in \desc_n$. 
    Then
    \begin{align}
        \Psi(\Aevolve(R, d)) = U_R\Psi(d)\ ,    \label{eq:claimEvolveOverlap1}
        \end{align}
        that is,  the output of~$\Aevolve$ is a description of the evolved state~$U_R\Psi(d)$. Furthermore,
        denoting the output by~$d'=(\cov', x', r')=\Aevolve(R,d)$ we have  \begin{align}
    |r'|^2=|\lrangle{ x', \Psi(d') }|^2 \geq 2^{-n}\ .\label{eq:lowerboundreturnw}
\end{align}
\end{lemma}

\begin{proof}
Let us denote the input of~$\Aevolve$ by~$(R,d)$ where~$R\in \mathsf{Gen}(O(2n))$ and~$d = (\cov, x, r) \in \desc_n$. 
The state~$U_R \Psi(d)$ has covariance matrix~$\cov_0 = R \cov R^T$ computed in line~\ref{alg:mzerodefm} (see Section~\ref{sec:backgroundFermionicGaussianClassicalSimulation}).
By the properties of~$\Asupport$ (see Lemma~\ref{lem:Asupport}), the state~$\ket{y}$ with~$y =\Asupport(\cov_0) \in\{0,1\}^n$ computed in line~\ref{alg:yReferenceForM0} is such that
\begin{align}
    |\lrangle{y, U_R \Psi(d)}|^2 \geq 2^{-n}\ .\label{eq:ylowerboundrpsid}
\end{align}
In particular, it is non-zero. The remainder of the algorithm computes~$\lrangle{y, U_R \Psi(d)}$.

We first show the following:
\begin{claim}\label{claim:revolutionxm}
Lines~\ref{eq:computationlineone}--\ref{eq:computationlineend} compute~$(z,s)\in \{0,1\}^n\times\mathbb{C}$ such that
\begin{align}
     |\lrangle{z, U_R x}|^2\geq 1/2\label{eq:zurxclaimone}
\end{align}
and
\begin{align}
    s=\langle z,U_Rx\rangle\ .\label{eq:zurxclaimtwo}
\end{align}
\end{claim}
\begin{proof}
Here we are using the fact that for any generator~$R\in \mathsf{Gen}(O(2n))$, the associated Gaussian unitary~$U_R$ has a local action on the mode operators. In particular, we can easily compute the image~$U_R\ket{x}$ of a number state~$\ket{x}$ under~$U_R$. We distinguish two cases:
\begin{enumerate}[(i)]
    \item  $R=R_{j,k}(\vartheta)$, $j < k \in [2n]$, $\vartheta\in [0,2\pi)$ is a Givens-rotation (see Lines~\ref{eq:computationlineone}--\ref{eq:computationlineEndGivens}):
In this case, $R$ is associated with the unitary evolution operator
\begin{align}
    U_{j,k} = \exp(\vartheta/2 c_j c_k)=\cos(\vartheta/2)I+\sin(\vartheta/2)c_jc_k\ .
\end{align}
It maps a basis state~$\ket{x}$, $x\in \{0,1\}^n$, to
\begin{align}
    U_{j,k}(\vartheta) \ket{x} = \cos(\vartheta/2) \ket{x} + e^{i \pi(\beta_{j}(x)+\beta_k(x))}    \sin(\vartheta/2)
    \ket{x \oplus e_j \oplus e_k}\ 
    \label{eq:UjkmActingOnZ}
\end{align}
where we introduced the quantities 
\begin{align}
    \beta_{s}(x)= \eta_s(x)+\left(x_s - \frac{1}{2}\right) \cdot (s+1)\qquad\textrm{ for any }\qquad s\in [n]\ ,
\end{align}
with~$\eta_s(x)$ defined in Eq.~\eqref{eq:etajdef}. To obtain Eq.~\eqref{eq:UjkmActingOnZ}, we used that
\begin{align}
    c_{j} | x \rangle &= 
    e^{i \pi \beta_j(x)}
    \ket{x \oplus e_j}
    \qquad\text{ for all }\quad
    j\in [2n]\quad\textrm{ and }\quad x \in \sbin^n    \label{eq:cjAction}
\end{align}
because
\begin{align}
    c_{2j-1} | x \rangle = (-1)^{\eta_j(x)} | x \oplus e_j \rangle  
    \qquad\text{and}\qquad
    c_{2j} | x \rangle = -i (-1)^{\eta_j(x) + x_j } | x \oplus e_j \rangle \ .
\end{align} 
Eq.~\eqref{eq:UjkmActingOnZ} motivates the following case distinction:
\begin{enumerate}[(a)]
\item $\cos^2(\vartheta/2) \geq 1/2$ (see Lines~\ref{alg:xIsZ}--\ref{alg:sIsCos}): Here~$\ket{x}$ has higher amplitude than~$\ket{x\oplus e_j\oplus e_k}$ in the state~$U_{j,k}(\vartheta)\ket{x}$. The algorithm picks~$z=x$ (Line~\ref{alg:xIsZ}) and sets
$s=\cos(\vartheta/2)$ (line~\ref{alg:sIsCos}). In particular, comparing with~\eqref{eq:UjkmActingOnZ}, it follows immediately that the claims~\eqref{eq:zurxclaimone} and~\eqref{eq:zurxclaimtwo} are satisfied.
\item $\cos^2(\vartheta/2) < 1/2$ (see Lines~\ref{alg:xIsZOplusOplus}--\ref{alg:sIsPhaseBeta}): 
In this case the algorithm ensures that 
\begin{alignat}{2}
    z &= x \oplus e_j \oplus e_k
    &&\text{ by Line~\ref{alg:xIsZOplusOplus}}   \label{eq:zejekdefm}\\
    \beta &= \beta_{j}(x)+\beta_k(x) 
    &&\text{ by Line~\ref{alg:betakgeqj}}\\
    s &=e^{i\pi (\beta_{j}(x)+\beta_k(x))} \sin(\vartheta/2)
    \qquad &&\text{ by Lines~\ref{alg:betakgeqj} and~\ref{alg:sIsPhaseBeta}}\ .\label{eq:sdefxkmn}
\end{alignat}
Because~$\cos^2(\vartheta/2)+\sin^2(\vartheta/2)=1$ we have
\begin{align}
    |s|^2 &\geq \frac{1}{2}
\end{align}
by the assumption that~$\cos^2(\theta/2)<1/2$.  To prove the claims~\eqref{eq:zurxclaimone} and~\eqref{eq:zurxclaimtwo}, it thus suffices to show the second claim~\eqref{eq:zurxclaimtwo}. But this again follows from~\eqref{eq:UjkmActingOnZ} and the definitions~\eqref{eq:zejekdefm} of~$z$ and~\eqref{eq:sdefxkmn} of~$s$, i.e., we have~$s=\langle z,U_{j,k}(\vartheta)x\rangle$.
\end{enumerate}
\item  $R=R_j$, $j\in [2n]$ is a reflection (see Lines~\ref{eq:computationlineBeginReflection}--\ref{eq:computationlineend}):
Here~$R$ is associated with the unitary evolution operator~
\begin{align}
    U_{j}=c_j\ .
\end{align} Its action on~$\ket{x}$, $x\in\sbin^n$, is described by Eq.~\eqref{eq:cjAction}, i.e., we have
\begin{align}
 U_{j} \ket{x}&= 
    e^{i \pi \beta_j(x) }
    \ket{x \oplus e_j}\ . 
\end{align}
This state is proportional to~$\ket{x\oplus e_j}$, showing that the choice
\begin{alignat}{2}
   z &= x\oplus e_j
   \qquad&&\text{(Line~\ref{alg:RjX})}\\
    \beta &= \beta_j(x)
    &&(\text{Line~\ref{alg:RjBeta}})\\
    s &=e^{i\pi\beta}
    &&\text{(Line~\ref{alg:RjS})}
\end{alignat}
indeed ensures that the claims~\eqref{eq:zurxclaimone} and~\eqref{eq:zurxclaimtwo} are satisfied.
\end{enumerate}
\end{proof}

Equipped with Claim~\ref{claim:revolutionxm}, we can  show that the algorithm~$\Aevolve$ has the desired functionality. The matrix~$\Gamma_0$ computed in Line~\ref{alg:mzerodefm} 
is the covariance matrix of the evolved state~$U_R \Psi(d)$, whereas~$\Gamma_1,\Gamma_2$ computed in Line~\ref{alg:monetwodef} 
are the covariance matrices of~$U_R\ket{x}$ and~$\ket{y}$, respectively. Thus~$\Aoverlaptriple$ 
in Line~\ref{alg:overlapEndFinal} is invoked on the triple of states~\begin{align}
    (\Phi_0,\Phi_1,\Phi_2)&=(U_R\Psi(d), U_R\ket{x},\ket{y}) \ .
\end{align}
To check that the requirements of~$\Aoverlaptriple$ are satisfied, first observe that
\begin{alignat}{2}
    u &= \overline{r}
    &&\textrm{ by Line~\ref{alg:evolveU}} \\
    &= \lrangle{\Psi(d), x}
    &&\textrm{ by definition of~$r$}\\
    &= \lrangle{ U_R \Psi(d), U_R x}
    \qquad&&\textrm{ by unitarity of~$U_R$}\\
    &= \lrangle{\Phi_0, \Phi_1}\ .
\end{alignat}
Furthermore, this is non-zero because~$r$ (part of the input) is non-zero by definition of the description~$d = (\cov_0, x, r)$ of~$\Psi(d)$.

By the defining property of the subroutine~$\Arelate$, Line~\ref{alg:overlapBeginFinal} of the algorithm computes~$(\alpha, \gamma) \in \sbin^{2n} \times [0, 2\pi)$ such that
\begin{align}
    c(\alpha)\ket{y} = e^{i\gamma} \ket{z}\ .\label{eq:cxvarthetay}
\end{align}
We also have 
\begin{alignat}{2}
    v &= e^{i \gamma} \overline{s}
    &&\textrm{   by line~\ref{alg:evolveU}}\\
    &= e^{i \gamma} \overline{\lrangle{z,U_R x}} 
    &&\textrm{ by~\eqref{eq:zurxclaimtwo}}\\
    &= \lrangle{U_R x, c(\alpha) y}  
    \qquad&&\textrm{ by~\eqref{eq:cxvarthetay}} \\
    &= \lrangle{\Phi_1, c(\alpha) \Phi_2}\ .
\end{alignat}
Because~$|s|^2\geq 1/2$ (see Claim~\ref{claim:revolutionxm}), this is non-zero. 

We conclude from the the properties of~$\Aoverlaptriple$  that 
\begin{align}
    w &= \lrangle{\Phi_2, \Phi_0}\qquad\qquad\textrm{ see Line~\ref{alg:overlapEndFinal}} \\
    &= \lrangle{y,U_R \Psi(d)} \ .
\end{align}
By construction of~$y$ using~$\Asupport$, we have
\begin{align}
    |w|^2 &\geq 2^{-n}\ ,
\end{align}
see Eq.~\eqref{eq:ylowerboundrpsid}. In particular, we conclude that the triple~$(\Gamma_0,y,w)$ is a valid description of~$U_R\Psi(d)$ with the desired property~\eqref{eq:lowerboundreturnw}. This is what the algorithm returns.
    
The runtime of the algorithm~$\Aevolve$ is dominated by the runtime~$O(n^3)$ of the algorithm $\Aoverlaptriple$.

\end{proof}

We give pseudocode for the algorithm~$\Ameasprob$ in Fig.~\ref{fig:Ameasprob}.

\begin{figure}[H]
\begin{mdframed}[
    linecolor=black,
    linewidth=0.5pt,
    roundcorner=2pt,
    backgroundcolor=white, 
    userdefinedwidth=\textwidth,
]
\begin{algorithmic}[1]
        \Require{$d = (\cov, x, r) \in\desc_n$}
        \Require $j \in [n]$
        \Require $s \in \sbin$
        \Function{$\Ameasprob$}{$d, j, s$}
        \State \textbf{return} $\frac{1}{2} (1+(-1)^s \cov_{2j-1,2j})$ \label{alg:measprobout}
        \EndFunction
\end{algorithmic}
\end{mdframed}
\caption{The subroutine~$\Ameasprob$ takes as input a description~$d=(\cov,x,r) \in \desc_n$ of a Gaussian state~$\Psi(d)$, an integer~$j\in[n]$ and a bit~$s\in\sbin$. It outputs the probability of obtaining the measurement outcome~$s$ when measuring the occupation number operator~$a_j^\dagger a_j$. The outcome probability does not depends on the global phase of~$\Psi(d)$ (which is determined by its reference state~$x$ and the overlap~$r$), but only on its covariance matrix~$\cov$.
\label{fig:Ameasprob}}
\end{figure}

\begin{lemma}
The algorithm~$\Ameasprob: \desc_n \times [n] \times \sbin \rightarrow \bbR$ given in Fig.~\ref{fig:Ameasprob} runs in time~$O(1)$. 
It satisfies
\begin{align}
    \Ameasprob(d, j, s) = \lrangle{\Psi(d), \Pi_j(s) \Psi(d)}
    \qquad
    \textrm{ for all }\qquad d \in \mathsf{Desc}_n, j\in [n], s\in \{0,1\}\ ,
\end{align}
where~$\Pi_j(s)=\frac{1}{2}(I+(-1)^{s}ic_{2j-1}c_{2j})$ is the projection onto the eigenvalue-$s$  eigenspace of~$a_j^\dagger a_j$. 
\end{lemma}

\begin{proof}
We denote the input to~$\Ameasprob$ by~$(d, j, s)$ where~$d=(\cov,x,r) \in \desc_n$ is a description of a state~$\Psi(d)$, $j\in[n]$ and~$s\in\sbin$.
Given the state~$\Psi(d)$, the probability
of obtaining measurement outcome~$s$ when measuring the occupation number operator~$a_j^\dagger a_j$ is given by Eq.~\eqref{eq:probNumberMeas}. This is the output of the algorithm in line~\ref{alg:measprobout} and gives the claim. 
Computing line~\ref{alg:measprobout} requires a constant number of arithmetic operations, giving the runtime~$O(1)$.
\end{proof}

We give pseudocode for the algorithm~$\Ameasure$ in Fig.~\ref{fig:Ameasure} and we illustrate it in Fig.~\ref{fig:ameasurefig}.

\begin{figure}[h]
\begin{mdframed}[
    linecolor=black,
    linewidth=0.5pt,
    roundcorner=2pt,
    backgroundcolor=white, 
    userdefinedwidth=\textwidth,
]
\begin{algorithmic}[1]
        \Require $d = (\cov, x, r)\in \desc_n$
        \Require $j \in [n]$
        \Require $s \in \sbin$
        \Require $p = \| \Pi_j(s) \Psi(d)\|^2 \in [0,1]$
        \Comment{probability of outcome~$s$ when measuring~$a_j^\dagger a_j$}
        \Function{$\Ameasure$}{$d, j, s, p$}
            \State{$\cov'\leftarrow 0\in \mathsf{Mat}_{2n\times 2n}(\mathbb{R})$} \label{alg:Mbegin}
            \Comment{compute covariance matrix of post-measurement state~$\Psi'$}
            \State{$\cov'_{2j,2j-1}\leftarrow (-1)^{s}$}
            \For{$\ell \leftarrow 1$ to~$n-1$}
                \For{$k \leftarrow \ell+1$ to~$n$}
                    \If{$(k,\ell)\neq (2j,2j-1)$}
                    \State{$\cov'_{k, \ell}\leftarrow      \cov_{k,\ell}+\frac{(-1)^{s}}{2p} (\cov_{2j-1,\ell}\cov_{2j,k}-\cov_{2j-1,k}\cov_{2j,\ell})$} \label{alg:Mend}
                    \EndIf
                \EndFor
            \EndFor   
        \State $\cov'\leftarrow \cov'-(\cov')^T$
        \State $y \leftarrow \Asupport(\cov')$ \label{alg:measureY}
        \Comment{find~$y$ such that~$|\langle y, \Psi'\rangle|^2\geq 2^{-n}$}
        \State $(\alpha,\vartheta)\leftarrow \Arelate(y,x)$ \label{alg:measureAlphaVartheta}
        \Comment{$(\alpha,\vartheta)$ are such that~$c(\alpha)\ket{y}=e^{i\vartheta}\ket{x}$}
        \State $\cov_0 \leftarrow \cov$, $\cov_1 \leftarrow \cov(|x\rangle)$, $\cov_2 \leftarrow \cov(|y\rangle)$ \label{alg:measureSetMs}
        \Comment{covariance matrices of~$\Psi(d)$, $\ket{x}$ and~$\ket{y}$}
        \State $u \leftarrow \overline{r}$ \label{alg:measureU}
        \Comment{$u=\langle \Psi(d),x\rangle$}
        \State $v \leftarrow e^{i\vartheta}$ \label{alg:measureV}
        \Comment{$v=\langle x,c(\alpha)y\rangle$}
        \State $w \leftarrow \Aoverlaptriple(\cov_0, \cov_1, \cov_2, \alpha, u, v)~$ \label{alg:measureW}
        \Comment{$w=\langle y,\Psi(d)\rangle$}
        \State \textbf{return} $(\cov', y, w/\sqrt{p})$
        \Comment{return a description of~$\Psi'$}
        \EndFunction
\end{algorithmic}
\end{mdframed}
\caption{The algorithm~$\Ameasure$ takes as input a description~$d\in \desc_n$, an integer~$j\in[n]$, a bit~$s\in\sbin$ and a real number~$p \in [0,1]$.
Assuming~$p=\| \Pi_j(s) \Psi(d)\|^2$, the algorithm outputs a description of the post-measurement state~$\Psi'=(\Pi_j(s)\Psi(d)) / \|\Pi_j(s)\Psi(d)\|$ when measuring the number operator~$a_j^\dagger a_j$ and obtaining the outcome~$s$ where~$s\in\{0,1\}$. Here, $\Pi_j(s)=(I+(-1)^{s}ic_{2j-1}c_{2j})/2$ is the projection onto the eigenvalue~$s$ eigenspace of~$a_j^\dagger a_j$. \label{fig:Ameasure}}
\end{figure}

\begin{figure}[h]
\centering
\begin{subfigure}[t]{0.31\textwidth}
\centering
    \includegraphics[height=3cm]{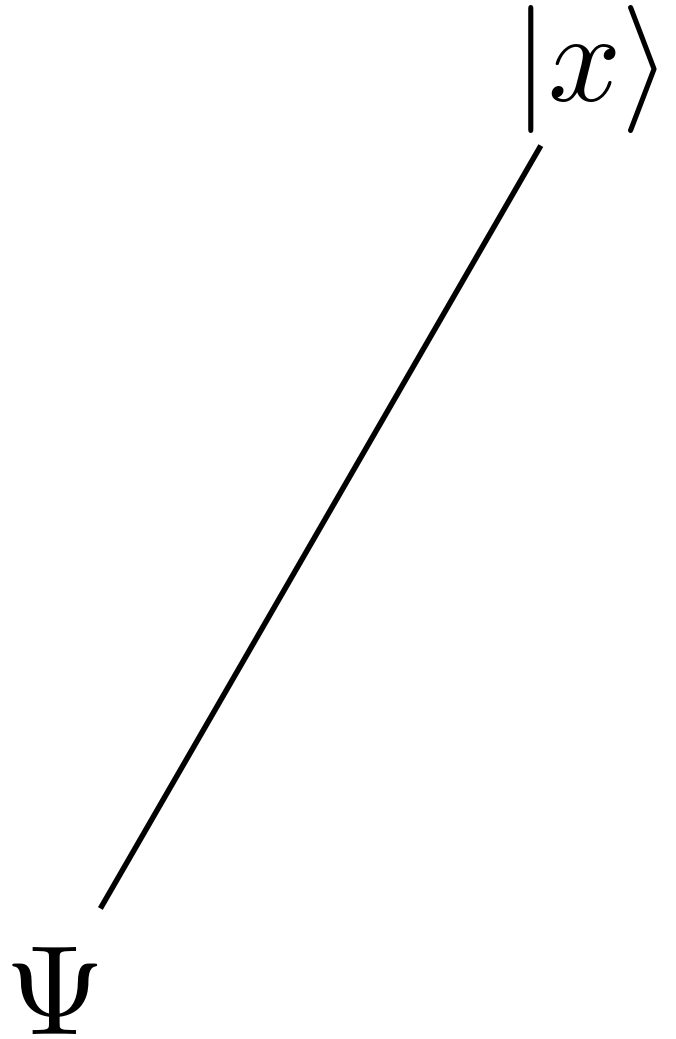}
    \caption{The algorithm~$\Ameasure$ is given the covariance matrix~$\Gamma$ of a Gaussian state~$\Psi$, $x\in \{0,1\}^n$ as well as the value~$r=\langle x,\Psi\rangle$ (which is non-zero).   It is additionally given the probability~$p=\|\Pi_j(s)\Psi\|^2$ of observing the outcome~$s$.
    }
    \label{fig:postaevolvefig}
\end{subfigure}
\hfill
\begin{subfigure}[t]{0.31\textwidth}
\centering
\includegraphics[height=3cm]{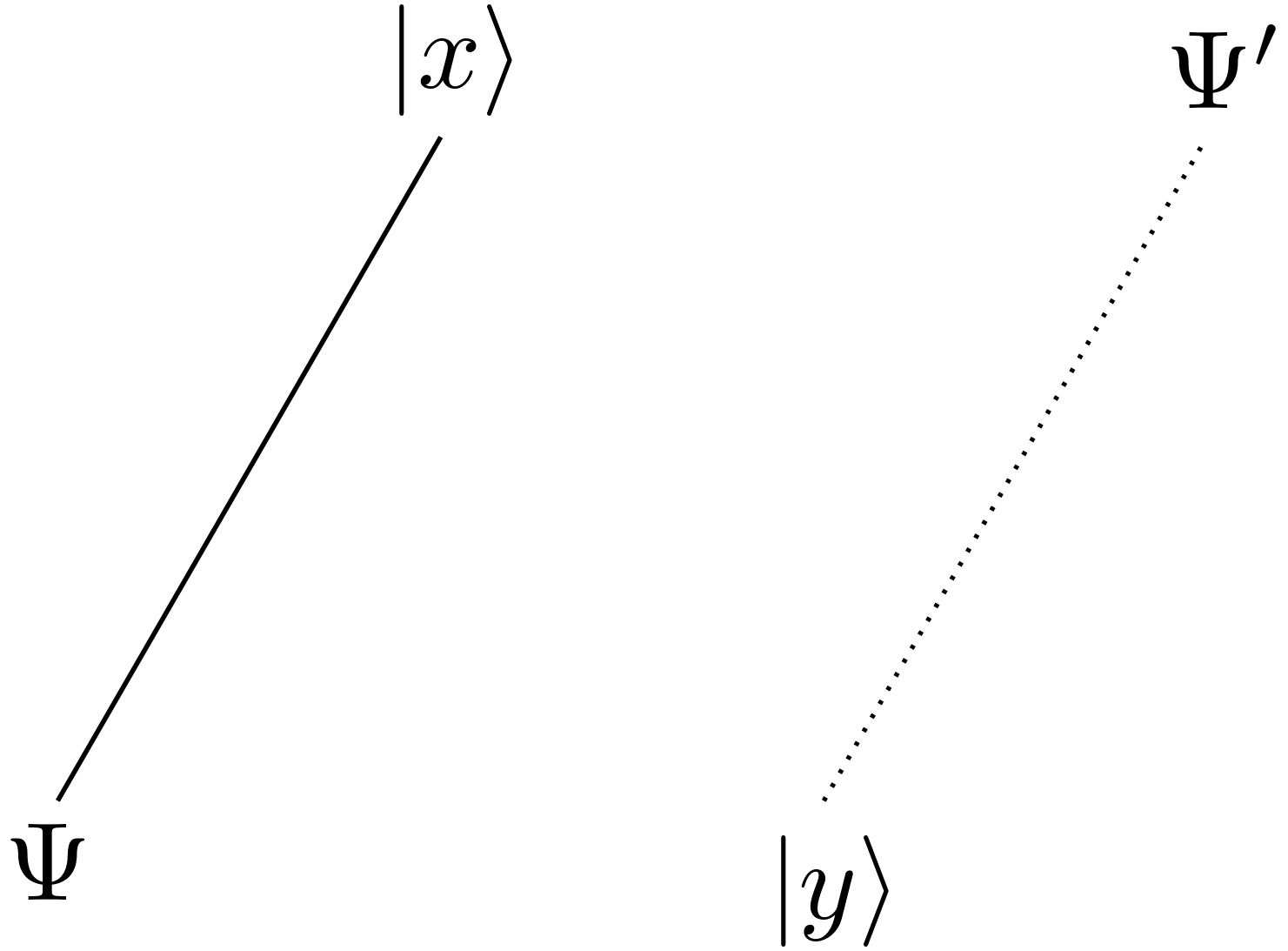}
    \caption{After computing the covariance matrix~$\Gamma'$ of the post-measurement state~$\Psi'$, the algorithm uses the subroutine~$\Asupport$ to find an element~$y\in \{0,1\}^n$ such that~$\langle y,\Psi'\rangle\neq 0$. (The value of this inner product is not computed at this point.)
    }
    \label{fig:postaevolvefigtwo}
\end{subfigure}
\hfill
\begin{subfigure}[t]{0.31\textwidth}
\centering
\includegraphics[height=3cm]{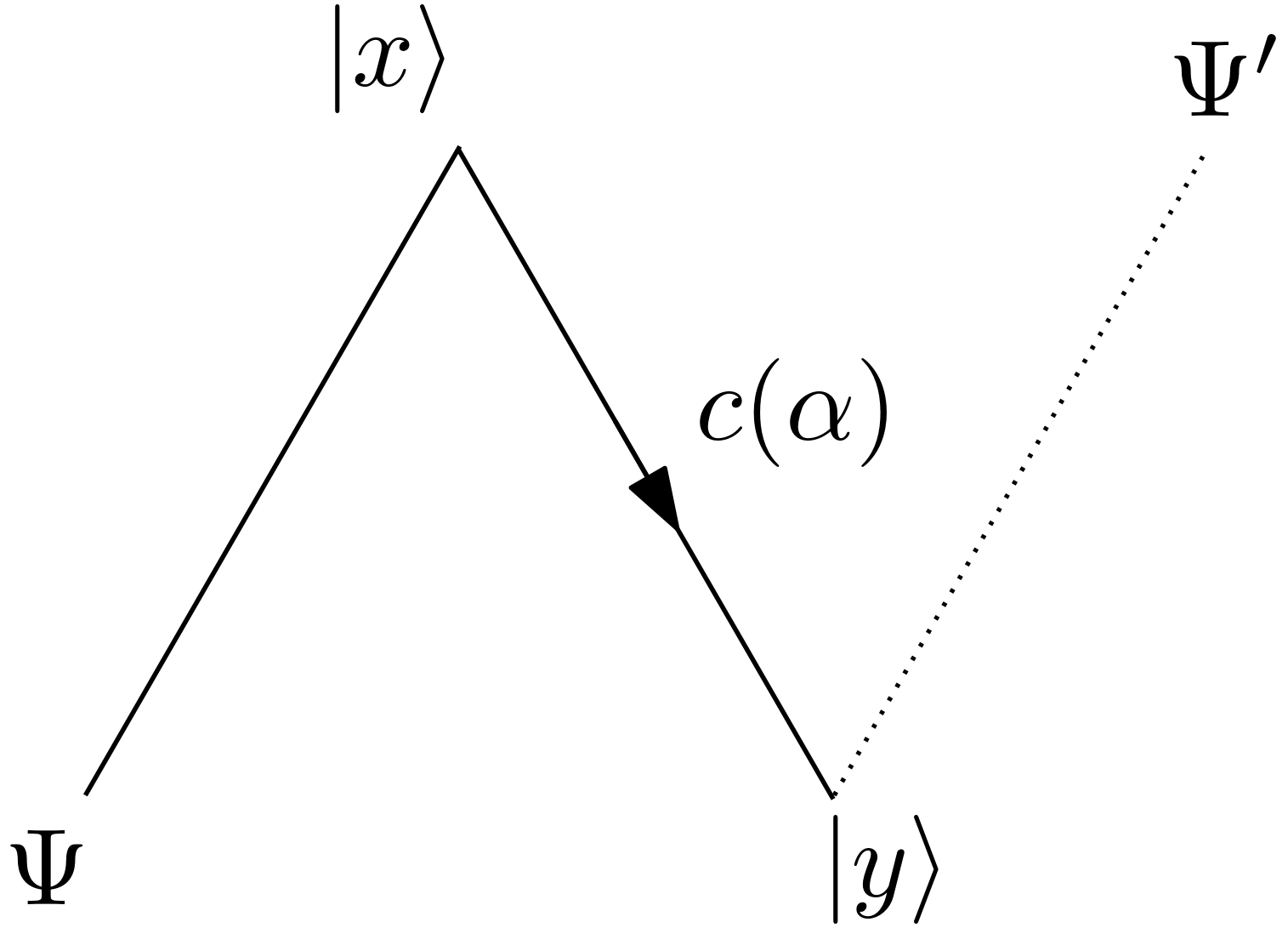}
    \caption{The subroutine $\Arelate$ is then used to find $(\alpha,\vartheta)$ such that $c(\alpha)\ket{y}=e^{i\vartheta} \ket{x}$.   
    }
    \label{fig:postaevolvefigthree}
\end{subfigure}\\ 
\vspace{4mm}
\begin{subfigure}[t]{0.31\textwidth}
\centering
    \includegraphics[height=3cm]{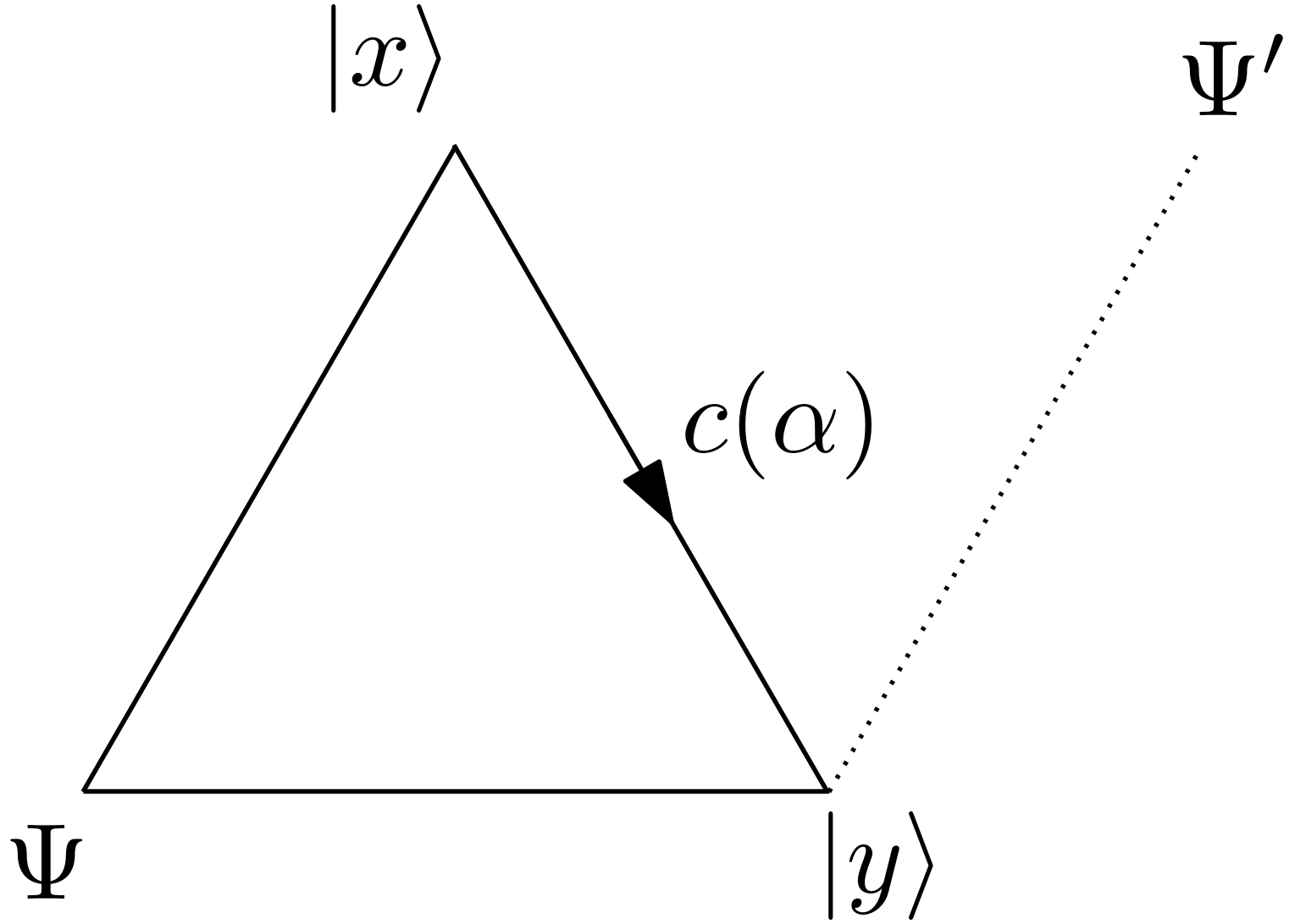}
    \caption{The algorithm~$\Aoverlaptriple$ is used to compute the overlap~$w=\langle y,\Psi\rangle$.}
    \label{fig:postaevolvefigfour}
\end{subfigure}
\hfill
\begin{subfigure}[t]{0.31\textwidth}
\centering
\includegraphics[height=3cm]{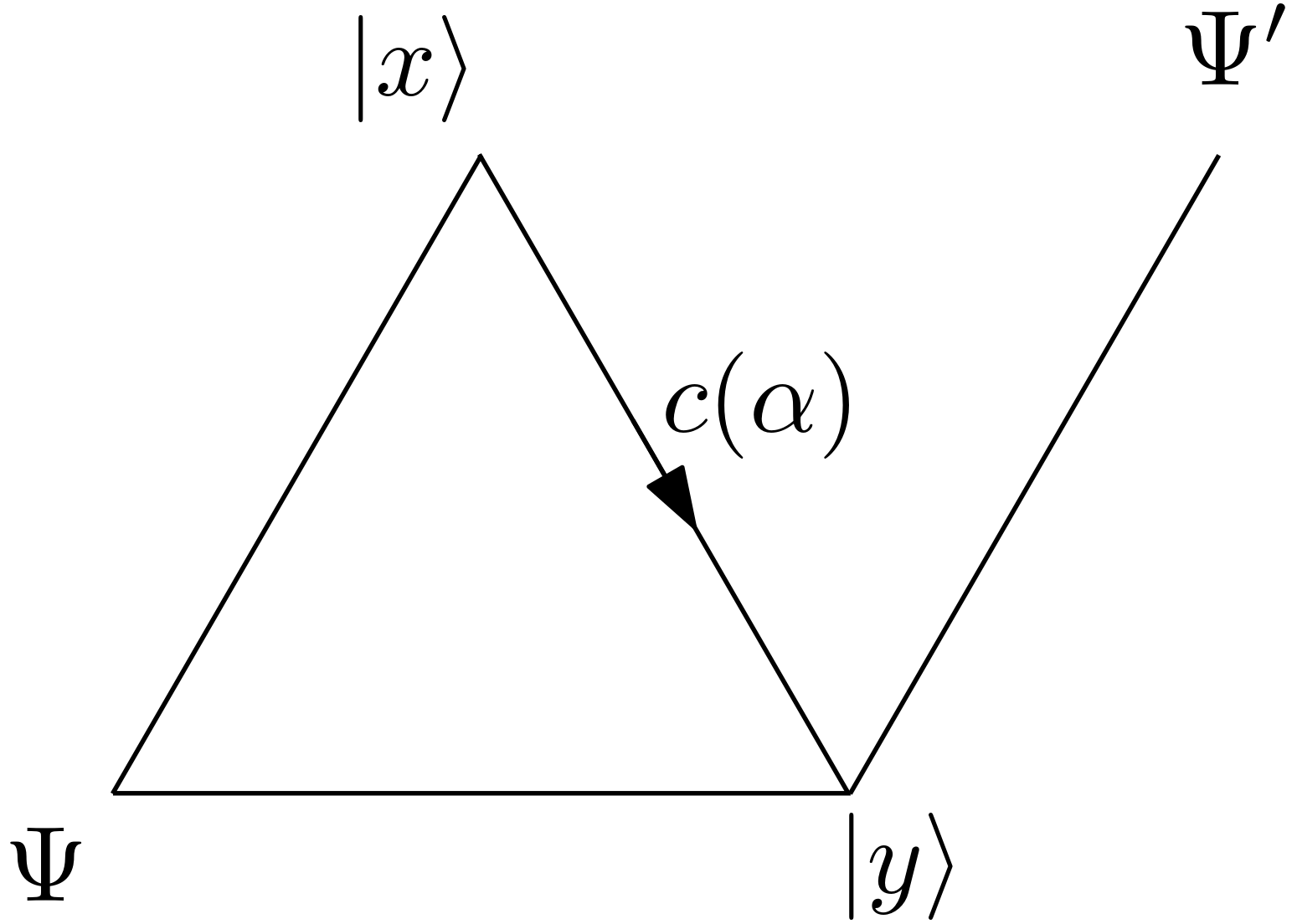}
    \caption{As argued in the proof of Lemma~\ref{lem:measureclaim}, the inner product~$\langle y,\Psi'\rangle$ can be computed from~$w$ and the probability~$p$: it is equal to~$\langle y,\Psi'\rangle=w/\sqrt{p}$. Thus~$(\Gamma',y,w,\sqrt{p})$ is a description of~$\Psi'$.
    }
    \label{fig:postaevolvefigfive}
\end{subfigure}
\hfill
\begin{subfigure}[t]{0.31\textwidth}
\centering
\hfill
\end{subfigure}
\caption{An illustration of the algorithm~$\Ameasure$, which computes a description of the post-measurement state~$\Psi'=\frac{\Pi_j(s)\Psi}{\|\Pi_j(s)\Psi\|}$ given a description of~$\Psi$.
}
\label{fig:ameasurefig}
\end{figure}

\begin{lemma}\label{lem:measureclaim}
The algorithm~$\Ameasure : \desc_n \times [n] \times \sbin \times [0,1] \rightarrow \desc_n$ given in Fig.~\ref{fig:Ameasure} runs in time~$O(n^3)$. Let~$d\in \desc_n$, $j\in [n]$ and~$s\in \{0,1\}$ be arbitrary. Let~$\Pi_j(s)=\frac{1}{2}(I+(-1)^{s}ic_{2j-1}c_{2j})$ be the projection onto the eigenvalue-$s$ eigenspace of~$a_j^\dagger a_j$ and let~$p=\| \Pi_j(s) \Psi(d) \|^2$.
Then
\begin{align}
\Psi(\Ameasure(d,j,s, p)&=\frac{\Pi_j(s)\Psi(d)}{\|\Pi_j(s)\Psi(d)\|}\ ,\label{lem:measureClaim1}
\end{align}
that is, $\Ameasure$ computes a description of the post-measurement state when measuring~$a_j^\dagger a_j$ and obtaining outcome~$s$. Denoting the output of the algorithm by
\begin{align}
d'=(\cov',x',r')=\Ameasure(d,j,s,p) \ ,   
\end{align}
we further have 
\begin{align}
|r'|^2 = |\lrangle{ x', \Psi(d')}|^2\geq 2^{-n}\ .\label{eq:rprimepostmeasuredjs}
\end{align}
\end{lemma}

\begin{proof}
We denote the input to~$\Ameasure$ by~$(d, j, s, p)$, where~$d\in\desc_n$, $j\in[n]$, $s\in\sbin$ and~$p\in[0,1]$.
For brevity, let us denote the post-measurement state when measuring the observable~$a_j^\dagger a_j$ and obtaining outcome~$s$ by
\begin{align}
\Psi'&=\frac{\Pi_j(s)\Psi(d)}{\|\Pi_j(s)\Psi(d)\|}\ .
\end{align}

In lines~\ref{alg:Mbegin}--\ref{alg:Mend}, the algorithm~$\Apostmeasure$ computes the covariance matrix~$\cov'$ of~$\Psi'$ according to Eq.~\eqref{eq:MpostMeasurement}. In line~\ref{alg:measureY} the algorithm uses~$\Asupport$ to find~$y\in\{0,1\}^n$ such that 
\begin{align}
\left|\langle y, \Psi' \rangle\right|^2&\geq 2^{-n}\ . 
\label{eq:overlapYPostMeasState}
\end{align} 
Line~\ref{alg:measureAlphaVartheta}
provides~$(\alpha,\vartheta)\in \{0,1\}^{2n}\times\mathbb{R}$ such that
\begin{align}
c(\alpha)\ket{y}=e^{i\vartheta}\ket{x}\ .\label{eq:alphaytranslationproperty}
\end{align}
Line~\ref{alg:measureSetMs} of the algorithm sets the matrices~$(\Gamma_0,\Gamma_1,\Gamma_2)$ equal to the covariance matrices of the three states
\begin{align}
    (\Phi_0,\Phi_1,\Phi_2)=(\Psi(d), \ket{x},\ket{y})\ .
\end{align}
We check that the conditions for applying~$\Aoverlaptriple$ in Line~\eqref{alg:measureW} are satisfied. We have 
\begin{align}
    u &= \overline{r}\qquad\qquad\qquad\textrm{ by Line~\ref{alg:measureU}}\\
    &= \lrangle{\Psi(d), x}\qquad\ \ \textrm{ by definition of~$r$}\\ 
    &= \lrangle{\Phi_0, \Phi_1}\ ,
\end{align}
and this is non-zero because~$d=(\Gamma,x,r)$ is a valid description (hence~$r\neq 0$). 
Similarly as before, we also have 
\begin{align}
    v &= e^{i \vartheta}\qquad\qquad\qquad\textrm{ by Line~\ref{alg:measureV}}\\
    &= \lrangle{x, c(\alpha) y}\qquad\ \ \ \textrm{ by Eq.~\eqref{eq:alphaytranslationproperty}}\\
    &= \lrangle{\Phi_1, c(\alpha) \Phi_2} \ .
\end{align}
In particular, this is also non-zero. 
The requirements to run~$\Aoverlaptriple$ in Line~\ref{alg:measureW} are therefore met, and Line~\ref{alg:measureW} returns
\begin{align}
    w &= \lrangle{ \Phi_2, \Phi_0 } \\
    &= \lrangle{ y,\Psi(d) } \ .\label{eq:wypsiexpressiontwo}
\end{align}

It remains to show that~$(\cov', y, w/\sqrt{p})$ (the expression returned by the algorithm) is a valid description of the post-measurement state~$\Psi'$, and to establish the  bound
\begin{align}
    |w/\sqrt{p}|^2&\geq 2^{-n}\label{eq:wpsiprimelowerbound}
\end{align} in order to prove Eq.~\eqref{eq:rprimepostmeasuredjs}.

Inserting~$\Psi'=\Pi_j(s)\Psi/\sqrt{p}$, Eq.~\eqref{eq:overlapYPostMeasState} implies that
\begin{align}
2^{-n} &\leq |\langle y,\Psi'\rangle|^2=\frac{1}{p}|\langle y,\Pi_j(s)\Psi\rangle|^2=\frac{1}{p}|\langle \Pi_j(s)y,\Psi\rangle|^2\leq\frac{1}{p} \|\Pi_j(s)y\|^2\cdot \|\Psi\|^2\ \label{eq:firstinequalitynfw}
    \end{align}
    because~$\Pi_j(s)$ is  self-adjoint and with the Cauchy-Schwarz inequality. 
    In particular, we have $\Pi_j(s)y\neq 0$ and thus 
    \begin{align}
        \Pi_j(s)y&=y\label{eq:pisyeigenvalueeq}
    \end{align} since any number state~$\ket{y}$ is an eigenvector of the projection~$\Pi_j(s)$. Inserting~\eqref{eq:pisyeigenvalueeq} into~\eqref{eq:firstinequalitynfw} and using~\eqref{eq:wypsiexpressiontwo}  we obtain the bound
    \begin{align}
        2^{-n} \leq  \frac{1}{p}|\langle y, \Psi\rangle|^2=\frac{|w|^2}{p}\ ,
    \end{align}
    establishing~\eqref{eq:wpsiprimelowerbound}. Eq.~\eqref{eq:pisyeigenvalueeq} and the self-adjointness of~$\Pi_j(s)$ also imply that
    \begin{align}
    \langle y,\Psi'\rangle &=\frac{1}{\sqrt{p}}\langle y,\Pi_j(s)\Psi\rangle=\frac{1}{\sqrt{p}}\langle y,\Psi\rangle=\frac{w}{\sqrt{p}}\ .
    \end{align}
    Since~$\Gamma'$ is the covariance matrix of~$\Psi'$ and~$p=\|\Pi_j(s) \Psi(d)\|^2$ is the probability of obtaining outcome~$s$ when measuring~$a_j^\dagger a_j$, this shows that~$(\Gamma',y,w/\sqrt{p})$ is a valid description of~$\Psi'$ as claimed.

The complexity of the algorithm is dominated by~$\Aoverlaptriple$, which takes time~$O(n^3)$. 

\end{proof}

\subsection{Initial states for computation}

Using the algorithm~$\Aevolve$, it is straightforward to generate a description of a state that is obtained by applying a sequence of Gaussian unitaries (generators) to the vacuum state. This is all that is typically needed to describe initial states.

In cases where we do not need to fix the overall phase, we can generate a description from the covariance matrix. The algorithm~$\Adescribe$ takes as input the covariance matrix~$\cov$ of a Gaussian state~$\Psi$ and
outputs a description~$d\in\desc_n$ of a Gaussian state which is equal to~$\Psi$ up to a global phase.
It is given in Fig.~\ref{fig:Adescribe} and it simply uses the subroutine~$\Asupport$ and Eq.~\eqref{eq:innerproductoftwomatrices}.

\begin{figure}[H]
\begin{mdframed}[
    linecolor=black,
    linewidth=0.5pt,
    roundcorner=2pt,
    backgroundcolor=white, 
    userdefinedwidth=\textwidth,
]
\begin{algorithmic}[1]
\Require $\cov\in\mathsf{Mat}_{n\times n}(\mathbb{R})$ covariance matrix of a pure Gaussian state~$\Psi$
        \Function{$\Adescribe$}{$\cov$}
            \State $y\in \{0,1\}^n$
            \State $y\leftarrow \Asupport(\cov)$\label{alg:eqasupport}
            \Comment{find~$y$ such that~$|\langle y,\Psi\rangle|^2\geq 2^{-n}$}
            \State $\sigma \leftarrow \Pf(\Gamma)$
            \State $r \leftarrow \sqrt{\sigma 2^{-n} \Pf( \cov(\ket{y}) + \cov )}$\label{lem:rcomputationmline}
            \Comment{compute~$|\langle y,\Psi\rangle|$}
            \State \textbf{return} $(\cov,y,r)$
        \EndFunction
\end{algorithmic}
\end{mdframed}
\caption{The algorithm~$\Adescribe$: Given the covariance matrix~$\cov$ of a Gaussian state~$\ket{\Psi}$, it outputs~$d\in \desc_n$  such that~$|\langle \Psi(d),\Psi\rangle|=1$.\label{fig:Adescribe}}
\end{figure}

\begin{lemma}
The algorithm~$\Adescribe:\mathsf{Mat}_{2n\times 2n}(\mathbb{R})\rightarrow \desc_n$ runs in time~$O(n^3)$. Its output is such that for every covariance matrix~$\cov$, the state~$\Psi(\Adescribe(\cov))$ is a Gaussian state with covariance matrix~$\cov$.
We have 
\begin{align}
    |r|^2=|\langle x, \Psi(d)\rangle|^2\geq 2^{-n}\label{eq:mhmxexpr}
\end{align} for~$d=(\Gamma, x, r) = \Psi(\Adescribe(\Gamma))$.
\end{lemma}
\begin{proof}
Let~$\cov\in\mathsf{Mat}_{2n\times 2n}(\mathbb{R})$ be a covariance matrix and let~$\Psi$ be a Gaussian state with covariance matrix~$\cov$. By definition of the algorithm~$\Asupport$, the  value~$y\in \{0,1\}^n$ computed in  line~\ref{alg:eqasupport} satisfies
\begin{align}
|\langle y,\Psi\rangle|^2 &\geq 2^{-n}\ .
\end{align}
By Eq.~\eqref{eq:innerproductoftwomatrices}, the value~$r$ computed in Line~\ref{lem:rcomputationmline}
satisfies
\begin{align}
    r&=|\langle y,\Psi\rangle|\ .
\end{align}
In particular, there is an angle~$\vartheta\in [0,2\pi)$ such that
\begin{align}
    r&=\sqrt{\sigma 2^{-n} \Pf( \cov(\ket{y}) + \cov )}=\langle y,e^{i\vartheta}\Psi\rangle\ .
\end{align}
It follows immediately that~$d=(\Gamma,y,r)$ is a valid description  of the Gaussian state~$e^{i\vartheta}\Psi=\Psi(d)$ with the required property~\eqref{eq:mhmxexpr}.
\end{proof}

\section{Classical simulation of fermionic Gaussian  circuits with non-Gaussian initial states\label{sec:classicalsimulationalgorithms}}

In this section, we argue that the techniques developed in Section~\ref{sec:trackingphases} to describe fermionic Gaussian states (including relative phases)  give rise to efficient classical simulation algorithms for computations composed of non-Gaussian initial states, Gaussian unitaries and occupation number measurements. Specifically, we argue that algorithms developed in the context of stabilizer circuits can immediately be translated to this fermionic setup. Furthermore, this translation maintains runtime bounds when the stabilizer extent is replaced by the fermionic Gaussian extent. Because of the generality of this adaptation procedure 
-- it being applicable to a variety of simulation algorithms both for strong and weak simulation -- we restrict our attention to the key substitutions.

Our algorithms apply to the efficient classical simulation of fermionic circuits of the following form, involving~$n$~fermions.
\begin{enumerate}[(i)]
\item\label{it:generalinitialstatenongaussian}
The initial state~$\Psi^{(0)}=\Psi$ is a possibly non-Gaussian state~$\Psi$. We assume that its fermionic Gaussian extent~$\xi(\Psi)$ and a corresponding optimal decomposition into a superposition of Gaussian states is known. This is the case for example for any four-fermion state, or a tensor product of two four-fermion states, see Section~\ref{sec:gaussianoverlap}. Alternatively, we may assume that an upper bound~$\overline{\xi}(\Psi)\geq \xi(\Psi)$ and a corresponding decomposition of~$\Psi$ achieving this value is known: In this case, runtime upper bounds will depend on~$\overline{\xi}(\Psi)$ instead of~$\xi(\Psi)$. 
\item
The computation proceeds in a sequence of timesteps. At each step~$t\in[T]$, one of the following is performed:
\begin{enumerate}[(a)]
\item\label{it:gaussianunitaryevolutionoperator}
A Gaussian unitary~$U_R$, $R\in\mathsf{Gen}(O(2n))$ is applied to the state. Here the choice of~$R$ may depend (in an efficiently computable manner) on measurement results obtained previously. We will leave this dependence implicit  and do not take it into account in our runtime estimates, as it will typically depend heavily on the circuit considered. 
\item
An occupation number measurement, i.e., measurement of the operator~$a_j^\dagger a_j$ for some~$j\in [n]$ is performed, yielding a measurement outcome~$s\in \{0,1\}$ and a corresponding post-measurement state. The choice of the mode~$j\in [n]$ to be measured may again depend (in an efficient manner) on the measurement outcomes already obtained.
\end{enumerate}
\end{enumerate}
We note that the restriction to the set of Gaussian unitaries associated with generators of~$O(2n)$ in~\eqref{it:gaussianunitaryevolutionoperator} incurs no loss of generality at the cost of possibly increasing~$T$ by a factor of order~$O(n^2)$ and an additive term in the runtime of order~$O(n^3)$ since a decomposition of an arbitrary element~$R\in O(2n)$ of the form~\eqref{eq:soL} as a product of~$L\leq O(n^2)$ generators can be found in time~$O(n^3)$, see the discussion below Theorem~\ref{thm:maindescriptalgorithm}.

The use of arbitrary initial states~$\Psi$ in~\eqref{it:generalinitialstatenongaussian} allows us to model, in particular, the application of certain ``magic gates'' using so-called gadgets. These can be realized by using non-Gaussian auxiliary states combined with Gaussian operations, see e.g.,~\cite{PhysRevA.73.042313,PhysRevLett.123.080503}. Since all~$1$-, $2$- and~$3$-fermion states are Gaussian~\cite{meloPowerNoisyFermionic2013a}, $4$-fermion states provide the smallest non-trivial examples; these will also be our main focus in Section~\ref{sec:gaussianoverlap}. We refer to e.g.,~\cite{PhysRevA.73.042313,PhysRevLett.123.080503} for a discussion of these constructions.

We proceed as follows: In Section~\ref{sec:classicalsimulationalgorithms}, we formulate  in general terms how simulation algorithms for a model can be generalized to initial states that are superpositions: This follows known approaches for stabilizer circuits augmented by magic states. 
In Section~\ref{sec:sparsification} we review the relationship between the~$\cD$-extent and the~$\cD$-rank  defined by a dictionary~$\cD$.
In Section~\ref{sec:fastnormestimation} we discuss fast algorithms for estimating norms of superpositions of dictionary states.  
In Section~\ref{sec:classicalsimulationfermionic} we apply these constructions to our setup.

\subsection{Extending simulation algorithms to superpositions\label{sec:extendingtosuperpos}}

Here we discuss how to extend simulation algorithms for an efficiently simulable model $(\cD,\cE,\cM)$ in such a way that the resulting extended algorithms~$(\chi\Aevolve,\chi\Ameasprob,\allowbreak\chi\Apostmeasure)$
work with any initial state~$\Psi$ which is a superposition of~$\chi$ elements of~$\cD$ (i.e., has~$\cD$-rank bounded by~$\chi_{\cD}(\Psi)\leq \chi$). Our discussion is standard and is included only for the reader's convenience: It follows that for stabilizer states as discussed in~\cite{bravyiSimulationQuantumCircuits2019a}.

Recall that the dictionary~$\cD$ is a set of states, $\cE$ a set of operations and~$\cM$ a set of measurements.  In addition to the subroutines~$\Aevolve,\Ameasprob$ and~$\Apostmeasure$ 
for evolution and measurement associated with~$(\cD,\cE,\cM)$, the construction discussed here requires
an efficient  algorithm~$\Aoverlap$ which computes inner products~$\langle \Psi(d_1),\Psi(d_2)\rangle$ from descriptions~$(d_1,d_2)\in\desc^2_n$.  This means that the description~$d\in \desc_n$ of a state~$\Psi(d)$ must include phase information. For Gaussian states, the covariance matrix formalism has to be extended as discussed in Section~\ref{sec:trackingphases}.

Our goal is to find classical simulation algorithms for circuits of the following form:
\begin{enumerate}[(i)]
\item
The initial state~$\Psi^{(0)}=\Psi$ is a superposition of the form
\begin{align}
    \label{eq:decompositionIntoChiElems}
    \Psi= \sum_{j=1}^{\chi} \gamma_j \varphi_j \ .    
\end{align}
of~$\chi$ states~$\{\varphi_j\}_{j=1}^\chi\subset\cD$ with complex coefficients~$\{\gamma_j\}_{j=1}^\chi$. We assume that this decomposition is explicitly provided as an input to the classical algorithm in the form of a
$\chi$-tuple~$\{(\gamma_j,d_j)\}_{j=1}^{\cD}$, where~$d_j$ is an efficient classical descriptions of the state~$\varphi_j$.
\item
In each timestep~$t\in[T]$,
\begin{enumerate}[(a)]
\item\label{it:aevolvestep}
either an evolution operation~$E\in\cE$, or
\item
a measurement~$M\in\cM$ 
\end{enumerate}
is applied to the state. We assume that corresponding efficient descriptions of~$E$ respectively~$M$ are given to the classical simulation algorithm. 
\end{enumerate}
The algorithms~$(\Aevolve,\Ameasprob,\Apostmeasure,\Aoverlap)$ associated with the model \linebreak $(\cD,\cE,\cM)$ then immediately give rise to algorithms~$(\chi\Aevolve,\chi\Ameasprob,\chi\Apostmeasure)$ for simulating a more general circuit: At each time step~$t\in [T]$, the resulting algorithm maintains the data~$\{\gamma_j^{(t)},d_j^{(t))}\}_{j=1}^{\cD}$ describing the instantaneous state~$\Psi^{(t)}$ after step~$t$ as a linear combination
\begin{align}
\Psi^{(t)}&=\sum_{j=1}^\chi \gamma_j^{(t)}\Psi(d_j^{(t)})\ 
\end{align}
of vectors belonging to the dictionary~$\cD$, and the subroutines~$(\chi\Aevolve,\chi\Ameasprob, \allowbreak\chi\Apostmeasure)$  are used to successively update this description (respectively sample from corresponding measurement outcomes).

Before describing the extended routines~$\chi\Aevolve,\chi\Ameasprob,\chi\Apostmeasure$ in more detail, it is convenient to introduce a subroutine~$\chi\Anorm$ which takes as input a tuple $\{(\gamma_j,d_j)\}_{j=1}^\chi\in (\mathbb{C}\times \desc_n)^\chi$
and outputs the squared norm~$\|\sum_{j=1}^\chi \gamma_j \Psi(d_j)\|^2$. It is clear that such a primitive can be realized naively by using the algorithm~$\Aoverlap$ for computing inner products. This naive implementation, which we refer to as~$\chi\Anaivenorm$, requires time
\begin{align}
\mathsf{time}(\chi\Anaivenorm)&=\chi^2 \mathsf{time}(\Aoverlap)\ .
\end{align}
Let us now describe the procedures~$\chi\Aevolve,\chi\Ameasprob$ and~$\chi\Apostmeasure$, building on a (general) norm computation subroutine~$\chi\Anorm$.
\begin{enumerate}[(a)]
\item
if an evolution operation~$E\in E$  with description~$d_E$ is applied at time~$t$, then the description is updated by
setting
\begin{align}
\gamma_{j}^{(t)}&=\gamma_{j}^{(t-1)}\qquad\textrm{ and }\qquad d_j^{(t)}=\Aevolve(d_E,d_j^{(t-1)})\qquad\textrm{ for }\qquad j\in [\chi]\ .
\end{align}
This defines the algorithm~$\chi\Aevolve$.
The runtime of this update is~
\begin{align}
    \mathsf{time}(\chi\Aevolve)=\chi\cdot \mathsf{time}(\Aevolve)\ .
\end{align}

\item
if a (projective) measurement~$M=\{M_s\}_{s\in\cM}\in\cM$ with description~$d_M$ is applied to the state at time~$t$, then
the probability of obtaining~$s\in\cM$ is given by
\begin{align}
p(s|\Psi^{(t-1)})&= \|M_s\Psi^{(t-1)}\|^2=\left\|\sum_{j=1}^\chi \gamma_j^{(t-1)}\sqrt{p(s|\Psi^{(t-1)}_j)} \Psi^{(t-1)}_j(M,s)\right\|^2\ .
\end{align}
Here the probability~$p(s|\Psi^{(t-1)}_j)=\| M_s \Psi_j^{(t-1)} \|^2=\Ameasprob(d^{(t-1)}_j,d_{M},s)$ of obtaining outcome~$s$ when measuring~$\Psi^{(t-1)}_j$ can be efficiently obtained from the description~$d^{(t-1)}_j$ of~$\Psi_j^{(t-1)}$ and the description~$d_M$ of~$M$. 
(Summands~$j$ where the probability~$p(s|\Psi_j^{(t-1)})$ vanishes can be omitted from this sum.) 
Similarly, 
a description~$d_j(s)=\Apostmeasure(d^{(t-1)}_j,d_M,s)$ of the (normalized) post-measurement state~$\Psi^{(t-1)}_j(M,s)=\frac{1}{\sqrt{p(s|\Psi^{(t-1)}_j)}}M_s\Psi^{(t-1)}_j$  (when measuring~$\Psi^{(t-1)}_j$) can be obtained 
efficiently.
 In particular, setting~$\tilde{\gamma_j}=\gamma^{(t-1)}_j \sqrt{p(s|\Psi^{(t-1)}_j)}$, 
we conclude that  the outcome probability 
\begin{align}
p(s|\Psi^{(t-1)})&=\left\|\sum_{j=1}^\chi \tilde{\gamma}_j \Psi(d_j(s))\right\|^2\  \label{eq:squarenormcomputationb}
\end{align}
is the squared norm of a superposition of elements from~$\cD$.  This expression (together with the norm computation routine~$\chi\Anorm$) defines the algorithm~$\chi\Ameasprob$. In particular, given~$\{\tilde{\gamma}_j, d_j(s),p(s|\Psi_j^{(t-1)})\}_{j=1}^\chi$, the probability~$p(s|\Psi)$ can be evaluated (exactly) in runtime~$\mathsf{time}(\chi\Anorm)$. Since~$\chi\Ameasprob$ first has to compute the descriptions~$d_j(s)$ of the post-measurement states~$\Psi_j^{(t-1)}(M,s)$
and the probabilities~$\{p(s|\Psi_j^{(t-1)})\}_{j=1}^\chi$, its runtime is
\begin{align}
    \mathsf{time}(\chi\Ameasprob)= \mathsf{time}(\chi\Anorm)+\chi\cdot (\runtime(\Ameasprob) + \runtime(\Ameasure))\ .
\end{align} 
 One can easily verify that the post-measurement state
after time step~$t$ is given by
\begin{align}
\Psi^{(t)}&=\sum_{j=1}^\chi \gamma_j^{(t)}\Psi(d_j^{(t)})\ ,
\end{align}
where
\begin{align}
\gamma_j^{(t)}=\frac{\tilde{\gamma}_j}{\sqrt{p(s|\Psi^{(t-1)})}}\qquad\textrm{ and }\qquad d_j^{(t)}=d_j(s)\ .
\end{align} 
In particular, this means that (similarly as for~$\chi\Ameasprob$) we have an algorithm $\chi\Apostmeasure$ which given~$\{(\gamma_j,d_j^{(t-1)})\}_{j=1}^\chi$ and~$p(s|\Psi^{(t-1)})$, computes a description of the post-measurement state in time
\begin{align}
    \runtime(\chi\Apostmeasure)=\chi\cdot (\runtime(\Apostmeasure)+\runtime(\Ameasprob))\ .
\end{align}
Given the ability to compute~$p(s|\Psi)$ and assuming, e.g., that the number~$|\cM|$ of measurement outcomes is constant, one can then sample from this distribution (when the goal is to perform weak simulation) to get an outcome~$s\in\cM$.
\end{enumerate}

Using the naive algorithm~$\chi\Anaivenorm$ for~$\chi\Anorm$ gives 
runtimes
\begin{equation}
\begin{aligned}
\mathsf{time}(\chi\Aevolve)&=\chi\cdot\runtime(\Aevolve)\\
\mathsf{time}(\chi\Ameasprob)&=\chi^2\cdot \mathsf{time}(\Aoverlap) + \chi\cdot (\runtime(\Apostmeasure)+\runtime(\Ameasprob))\\
\mathsf{time}(\chi\Apostmeasure)&=\chi\cdot (\runtime(\Apostmeasure)+\runtime(\Ameasprob))
\end{aligned} \ .
\label{eq:naivealgorithmruntimes}
\end{equation}
As a function of~$\chi$, this is dominated by the computation of the squared norm~\eqref{eq:squarenormcomputationb} in~$\chi\Ameasprob$ which takes time~$O(\chi^2)$.

\subsection{Sparsification: Relating~$\cD$-extent to approximate~$\cD$-rank\label{sec:sparsification}}

Algorithms whose complexity depends on the~$\cD$-extent~$\xi_\cD(\Psi)$ instead of the (exact)~$\cD$-rank~$\chi_\cD(\Psi)$ (see Eq.~\eqref{eq:drankdefinition}) of the initial state~$\Psi$ can be obtained as follows. The idea consists in  replacing~$\Psi$ by a state~$\tilde{\Psi}$ which is~$\delta$-close to~$\Psi$ and has bounded~$\cD$-rank.
More precisely, it relies on the following result which connects the~$\cD$-extent~$\xi_\cD(\Psi)$ to the approximate~$\cD$-rank~$\chi_\cD^\delta(\Psi)$  defined in Eq.~\eqref{eq:approximaterankdefinition}.

\begin{theorem}[Theorem~1 in~\cite{bravyiSimulationQuantumCircuits2019a}]\label{thm:approximateextentexact}
Suppose~$\Psi=\sum_{j=1}^m \gamma_j \varphi_j$ is a decomposition of a normalized vector~$\Psi$ into a superposition of elements~$\{\varphi_j\}_{j=1}^m$ belonging to the dictionary~$\cD$. Then 
\begin{align}
\chi_\cD^\delta(\Psi)\leq 1+\|\gamma\|_1^2/\delta^2
\end{align}
where~$\|\gamma\|_1=\sum_{j=1}^m |\gamma_j|$ is the~$1$-norm of~$\gamma$. In particular, we have the relationship
\begin{align}
\chi_\cD^\delta(\Psi) &\leq 1+\xi_\cD(\Psi)/\delta^2\ .
\end{align}
\end{theorem}
In~\cite{bravyiSimulationQuantumCircuits2019a}, this result was established for the dictionary~$\cD=\stab_n$ consisting of~$n$-qubit stabilizer states. Inspection of the proof immediately shows that the statement is applicable to any dictionary~$\cD$ (independently of, e.g., whether or not it is finite). In particular, Theorem~\ref{thm:approximateextentexact} implies that in runtime upper bounds, the quantity~$\chi_\cD$ can always be replaced by (the potentially much smaller quantity)~$\xi_\cD(\Psi)/\delta^2$, at the cost of introducing an~$O(\delta)$-error in~$L^1$-distance in the sampled distribution. For example, using the naive norm estimation algorithm (i.e., inserting into~\eqref{eq:naivealgorithmruntimes}), this gives a quadratic scaling (for computing output probabilities) in~$\xi_{\cD}(\Psi)$. Note that here we are assuming that a decomposition of~$\Psi$ with squared~$L^1$-norm~$\|\gamma\|_1^2$ of coefficients achieving~$\xi_\cD(\Psi)$ is given.

\subsection{Fast norm estimation and approximate simulation \label{sec:fastnormestimation}}

In pioneering work, Bravyi and Gosset~\cite{bravyiImprovedClassicalSimulation2016} improved upon the  sketched algorithm in the case of stabilizer states. This was achieved  by replacing the~$O(\chi^2)$-runtime (naive) estimation algorithm~$\chi\Anaivenorm$ for the norm of a  superposition of stabilizer states by a probabilistic algorithm~$\chi\Afastnorm$. With success probability at least~$1-p_f$, the  algorithm~$\chi\Afastnorm$  provides 
an estimate~$\hat{N}$ of the squared norm~$N=\left\|\Psi\right\|^2$ of a superposition~$\Psi=\sum_{j=1}^\chi \gamma_j \varphi_j$ of~$n$-qubit stabilizer states~$\{\varphi_j\}_{j=1}^\chi$ with multiplicative error~$\epsilon$ (i.e., $\hat{N}\in [(1-\epsilon)N,(1+\epsilon)N])$, and has runtime~$O(\chi\cdot n^3\epsilon^{-2}p_f^{-1})$
 The key observation underlying the algorithm is the fact that the norm of interest can be expressed as 
\begin{align}
N &=2^{n}\ExpE_{\Theta}\left[|\langle \Theta,\Psi\rangle|^2\right]\ ,\label{eq:thetajpsioverlapexpectation}
\end{align}
i.e., it is proportional to the expected (squared) overlap of~$\Psi$ with a state~$\ket{\Theta}$ drawn uniformly at random from the set of~$n$-qubit stabilizer states. Given~$\{\gamma_j,\varphi_j\}_{j=1}^\chi$, this algorithm
proceeds by taking~$R=\lceil p_f^{-1}\epsilon^{-2} \rceil$ stabilizer states~$\Theta_1,\ldots,\Theta_R$ chosen uniformly from the set of all stabilizer states, and producing the estimate
\begin{align}
    \hat{N}=\frac{2^n}{R}\sum_{k=1}^R |\langle \Theta_k,\Psi\rangle|^2\ \label{eq:thetajpsioverlap}
\end{align}
of~$N$.
Importantly, expression~\eqref{eq:thetajpsioverlap} can be computed from (the descriptions of)~$\{\Theta_k\}_{k=1}^R$, $\{\varphi_j\}_{j=1}^\chi$ and the coefficients~$\{\gamma_j\}_{j=1}^\chi$, using~$\chi\cdot R$ calls of a subroutine~$\Aoverlap$ which computes the overlap of two stabilizer states (including phases). This is because each summand in~\eqref{eq:thetajpsioverlap} can be written as a sum 
\begin{align}
    |\langle \Theta_k,\Psi\rangle| &=\left|\sum_{j=1}^\chi \gamma_j \langle \Theta_k,\varphi_j\rangle\right|\label{eq:expressionoverlamindep}
\end{align}
of~$\chi$ such products. The resulting runtime of this norm estimation algorithm is thus~$O(\chi\cdot R\cdot \mathsf{time}(\Aoverlap))$ which amounts to the  claimed runtime~$O(\chi\cdot n^3\epsilon^{-2}p_f^{-1})$.

We note that a  similar reasoning can be applied to any situation where the norm of a superposition of dictionary elements of interest can be expressed as in Eq.~\eqref{eq:thetajpsioverlapexpectation} as the expected overlap of the inner product of~$\Psi$ with a state~$\Theta$ randomly chosen according to a suitable distribution over dictionary states. Specifically, as derived in Appendix~B of~\cite{bravyiComplexityQuantumImpurity2017a} and discussed below (see Section~\ref{sec:classicalsimulationfermionic}), this is the case for the set of fermionic Gaussian states.
The corresponding norm estimation algorithm then has a runtime of the form
\begin{equation}
\begin{aligned}
    \runtime(\chi\Afastnorm)
    &=O(\chi\cdot R\cdot \mathsf{time}(\Aoverlap) +R\cdot\runtime(\Asamplestate)) \\
    &=O(p_f^{-1} \epsilon^{-2} 
    (\chi\runtime(\Aoverlap) + \runtime(\Asampletheta) )
\end{aligned}
\label{eq:runtimeAfastnorm}
\end{equation}
where~$\Asamplestate$ is an algorithm producing a description of a state~$\Theta$ drawn randomly form the appropriate distribution. Importantly, the runtime~\eqref{eq:runtimeAfastnorm} is linear in~$\chi$, resulting in a linear dependence when replacing~$\chi$ by the stabilizer extent~$\xi_\cD(\Psi)$ as discussed in Section~\ref{sec:sparsification}. 

Algorithms~$(\approxAevolve,\approxAmeasprob,\approxApostmeasure)$ can now be obtained by using~$\chi\Afastnorm$ in place of~$\chi\Anorm$. 
The algorithm~$\approxAevolve$ is identical to~$\chi\Aevolve$ since it does not involve norm computations.
In contrast, $\approxAmeasprob$ is a probabilistic algorithm that can fail with probability~$p_f$ and both
$\approxAmeasprob$ and $\approxApostmeasure$ introduce  an error (in the sampled distribution and the post-mea\-surement state, respectively). This is because~$\chi\Afastnorm$ only estimates the norm of a superposition.

Finally, replacing~$\chi$ by the~$\cD$-extent~$\xi_\cD(\Psi)$, see Section~\ref{sec:sparsification} results in a triple of approximate algorithms
$(\approxAevolve,\approxAmeasprob,\approxApostmeasure)$ with 
parameters~$(\epsilon,\delta,p_f)$ describing the quality of approximation and failure probability as discussed in Section~\ref{sec:ourcontribution}.  By construction, the runtimes of these algorithms are
\begin{equation}
\begin{aligned}
\runtime(\approxAevolve) &= O\left( \frac{\xi_\cD(\Psi)}{\delta^2} \runtime(\Aevolve) \right)\\
\runtime(\approxAmeasprob) &= O\left( p_f^{-1} \epsilon^{-2} \left( \frac{\xi_\cD(\Psi)}{\delta^2} \runtime(\Aoverlap) + \runtime(\Asamplestate) \right) \right)\\
 &+ O\left(\frac{\xi_\cD(\Psi)}{\delta^2}\left( \runtime(\Apostmeasure) 
+ \runtime(\Ameasprob) \right) \right) \\
\runtime(\approxApostmeasure) &= O\left( \frac{\xi_\cD(\Psi)}{\delta^2} \left(\runtime(\Apostmeasure) + \runtime(\Asamplestate) \right) \right)
\end{aligned} \ .
\label{eq:runtimeapproximationalgorithms}
\end{equation}

\subsection{Fermionic linear optics with non-Gaussian initial states~\label{sec:classicalsimulationfermionic}}

The algorithms described above can be adapted in a straightforward manner to the problem of classically simulating fermionic linear optics with non-Gaussian initial states: We can simply use the efficient description of Gaussian states introduced in Section~\ref{sec:trackingphases}
and make use of the associated procedures~$\Aoverlap$, $\Aevolve$ as well as~$\Ameasprob$ and~$\Apostmeasure$.
In particular, observe that combining Eq.~\eqref{eq:naivealgorithmruntimes} with the runtimes~$O(n^3)$ for the algorithms~$\Aevolve$, $\Apostmeasure$ and~$\Aoverlap$, and~$O(1)$ for~$\Ameasprob$ (see Section~\ref{sec:trackingphases}) results in the runtimes given in Table~\ref{tab:runtimeextendedalgo} for exact simulation.

To get a linear scaling in the Gaussian extent~$\xi_{\cG_n}(\Psi)$ of the initial state (for approximate simulation), the naive norm estimation needs to be replaced. A fast norm estimation scheme for superpositions of  fermionic Gaussian has been described in Appendix~C of  Ref.~\cite{bravyiComplexityQuantumImpurity2017a}: Consider the following probabilistic process defined for a superposition~$\Psi=\sum_{j=1}^\chi \gamma_j \varphi_j$, $\varphi_j\in\cG_n$, $\gamma_j\in\mathbb{C}$ of~$n$-mode fermionic Gaussian states:
\begin{enumerate}[(i)]
    \item Sample~$K$ random Gaussian states~$\{\Theta_k\}_{k=1}^K$ independently and identically from  the distribution induced by 
    picking a permutation~$\pi\in S_{2n}$ and a string~$y\in \{0,1\}^n$ uniformly at random and outputting
\begin{align}
\ket{\Theta(\pi,y)}=U_{R_\pi}\ket{y} .\label{eq:thetaury} 
\end{align}
Here~$R_\pi=O(2n)$ is a permutation matrix specified by an element~$\pi\in S_{2n}$ and~$y\in \{0,1\}^n$.
    \item 
    Set 
    \begin{align}
\hat{N}=\frac{1}{K}\sum_{k=1}^K 2^n |\langle \Theta_k,\Psi\rangle|^2\ . \label{eq:thetakPsidefinitionaverage}
    \end{align}
    \end{enumerate}
The following  was shown in~\cite{bravyiComplexityQuantumImpurity2017a}.
\begin{lemma}[Lemma~10 in Ref.~\cite{bravyiComplexityQuantumImpurity2017a}]
For any~$p_f\in [0,1]$ and~$\epsilon>0$, consider the probabilistic  process described above with the choice~$K=\lceil 2\sqrt{n}\epsilon^{-2}p_f^{-1} \rceil$. Then the random variable~$\hat{N}$ satisfies 
$$
(1-\epsilon)\|\Psi\|^2 \leq \hat{N} \leq(1+\epsilon)\|\Psi\|^2 \ .
$$
with probability at least~$1-p_f$.
\end{lemma}
A description of a state proportional to~$\Theta_k$ can be computed from the associated pair~$(\pi_k,y_k)\in S_{2n}\times \{0,1\}^n$  using the subroutine~$\Adescribe$, for each~$k\in [K]$, as follows (see Fig.~\ref{fig:Asampletetha} for pseudocode for the algorithm). The covariance matrix~$\cov_k=R_{\pi_k} \cov(\ket{y})R_{\pi_k}^\dagger$ of such a state can be computed in time~$O(n^3)$ from~$(\pi_k,y)$, and applying~$\Adescribe$ to~$\cov_k$ gives the desired description. We note that any such state can be used in place of~$\Theta_k$ since the expression~\eqref{eq:thetakPsidefinitionaverage} (and, in particular,~\eqref{eq:expressionoverlamindep}) does not depend on the global phase of~$\Theta_k$.  With the definition~\eqref{eq:thetakPsidefinitionaverage}, it follows that the probabilistic process described here can be simulated in time given by Eq.~\eqref{eq:runtimeAfastnorm} using~$K$ calls to 
the subroutine~$\Asampletheta$ shown in Fig.~\ref{fig:Asampletetha}, and subsequent use of~$\Aoverlap$
to compute the empirical average~\eqref{eq:thetakPsidefinitionaverage}.
\begin{figure}[H]
\begin{mdframed}[
    linecolor=black,
    linewidth=0.5pt,
    roundcorner=2pt,
    backgroundcolor=white, 
    userdefinedwidth=\textwidth,
]
\begin{algorithmic}[1]
        \Function{$\Asampletheta$}{}
        \State{$y \leftarrow$ uniform random string in~$\{0,1\}^n$ }
        \State{$\pi \leftarrow$ uniform random permutation in~$S_{2n}$ }
        \State{$R_\pi \in O(2n)$} 
        \For{$j \in [2n]$} 
        \Comment{compute the permutation matrix~$R_\pi$}
            \State{$R_\pi[j] \leftarrow e_{\pi(j)}$ }
        \EndFor
        \State{$\cov \leftarrow R_\pi \cov(\ket{y]}) R_\pi^\dagger$}
        \Comment{compute the covariance matrix of~$\ket{\Theta(\pi,y)} = R_{\pi} \ket{y}$}
        \State{\textbf{return} $ \Adescribe(\cov) $ }
        \Comment{output a description of~$\ket{\Theta(\pi,y)}$}
        \EndFunction
\end{algorithmic}
\end{mdframed}
\caption{The algorithm~$\Asampletheta$ outputs a classical description of a state~$\ket{\Theta(\pi,y)} = R_{\pi} \ket{y}$ where~$\pi\in S_{2n}$ and~$y \in \{0,1\}^n$ are taken  uniformly at random.
\label{fig:Asampletetha}}
\end{figure}

Because the runtimes of~$\Adescribe$ and~$\Aoverlap$ are both upper bounded by~$O(n^3)$, this leads to an overall runtime of~$O\left(n^{7 / 2} \epsilon^{-2} p_f^{-1} \chi\right)$ of this algorithm for computing the estimate~$\hat{N}$ of~$\|\Psi\|^2$.
We note this conclusion about the runtime was also reached in  Ref.~\cite{bravyiComplexityQuantumImpurity2017a}, although the issue of a potential lack of a phase reference applicable throughout the computation
was not considered there. This issue is resolved by our description of Gaussian states, see Section~\ref{sec:trackingphases}. 

Combing this algorithm with the runtimes given in Eq.~\eqref{eq:runtimeapproximationalgorithms} and with the runtimes~$O(n^3)$ for the algorithms~$\Aevolve$, $\Apostmeasure$ and~$\Aoverlap$, and~$O(1)$ for~$\Ameasprob$ (see Section~\ref{sec:trackingphases}) 
gives runtimes claimed in Table~\ref{tab:approximateruntimes} for the algorithms~$\approxAevolve$, $\approxAmeasprob$ and $\approxApostmeasure$.

To give an idea of the feasibility of running these algorithms, let us compare to~\cite{bravyiImprovedClassicalSimulation2016}, where Clifford circuits on $n=40$ qubits with up to~$t=48$~$T$-gates were simulated by means of an algorithm relying on the stabilizer formalism augmented with stabilizer-decompositions of magic states. 
These simulations were performed using MATLAB on a laptop with a 2.6GHz Intel i5 Dual Core CPU. Runtimes of individual building blocks are given in~\cite[Table I]{bravyiImprovedClassicalSimulation2016} and are of the order of tens of milliseconds even for 100~qubits. 
The  algorithms of~\cite{bravyiImprovedClassicalSimulation2016} have a scaling which is linear in the approximate (stabilizer) rank. 
Since this quantity scales exponentially in the number of copies of a single magic state, this dominates the runtime, which includes additional factors that are polynomial in the number of qubits and gates. In more detail,  the authors of~\cite{bravyiImprovedClassicalSimulation2016} consider the non-stabilizer state
\begin{align}
    |A\rangle = \frac{1}{\sqrt{2}} (|0\rangle + e^{i\pi/4} |1\rangle  )\ . 
\end{align}
A copy of this (magic) state can  be used together with stabilizer operations to implement the (non-Clifford)~$T$-gate~\cite{PhysRevA.71.022316}. 
The approximate stabilizer rank (which is defined in terms of fidelity rather than distance in~\cite{bravyiImprovedClassicalSimulation2016}) of $t$~copies of this state is shown to be bounded by~
\begin{align}
    \tilde{\chi}^\delta(|A\rangle^{\otimes t})=O(2^{\gamma t} \delta^{-1})\qquad\textrm{ where }\qquad \gamma\approx 0.228\ .\label{eq:chitm}
\end{align}      
Now consider our simulation procedure: Here the dependence of the runtime is linear in the  Gaussian extent~$\xi_{\mathcal{G}_n}$ and thus again essentially linear in the~$\delta$-approximate Gaussian rank~$\chi_{\mathcal{G}_n}^\delta$ (see Theorem~\ref{thm:approximateextentexact}), and this quantity dominates the runtime (which includes polynomial factors in the number of fermionic modes and gates). 
Thus our simulation algorithms have a runtime comparable to that of~\cite{bravyiImprovedClassicalSimulation2016} when substituting the fermionic Gaussian extent of the initial state with the approximate stabilizer rank (respectively the stabilizer extent) of the initial state in the Clifford setting.
    
To make this more concrete, consider the state
\begin{align}
    \label{eq:a8}
    |a_8\rangle = \frac{1}{\sqrt{2}} \left( |0000\rangle + |1111\rangle \right)\ , 
\end{align} which is 
a magic state for Gaussian computations: As shown in \cite[Lemma 1]{PhysRevA.73.042313}, the non-Gaussian unitary (magic gate)~$U_{\textrm{magic}}=\exp(i\pi c_1c_2c_3c_4)$ can be implemented with Gaussian operations and a single copy of this state.     
The state~$|a_8\rangle$ has Gaussian extent
\begin{align}
    \xi_{\mathcal{G}_n}(|a_8\rangle) = 2\ \label{eq:a8extenteq}
\end{align}
and thus we have 
\begin{align}
    \label{eq:upperbounda8} 
    \xi_{\mathcal{G}_n}(|a_8\rangle^{\otimes t'}) \leq  2^{t'} 
\end{align} 
for $t'$ copies of~$|a_8\rangle$. (We note that in fact, the results of~\cite{StrelchukCudby} imply that equality  holds in~\eqref{eq:upperbounda8}.)
Eq.~\eqref{eq:a8extenteq} can be shown as follows: The upper bound~$\xi_{\mathcal{G}_n}(|a_8\rangle) \leq  2$ is a direct consequence of the decomposition given in \eqref{eq:a8} into a sum of the Gaussian states~$|0000\rangle$ and~$|1111\rangle$. 
The dual formulation of the extent implies that  
\begin{align}
    \xi_{\mathcal{G}_n}(|\Psi\rangle) 
    &\geq  \frac{1}{F_{\mathcal{G}_n}(|\Psi\rangle)}\qquad\textrm{ for any state } \Psi\in\mathcal{H}_n\ ,\label{eq:dualityextentineq}
\end{align}
see~\eqref{eq:dualformulationxi} with the choice~$y=|\Psi\rangle / \sqrt{F_{\mathcal{G}_n}(|\Psi\rangle)}$.
Using the fact that~$F_{\mathcal{G}_n}(|a_8\rangle) =1/2$ (see~\cite{BOTERO200439} --  the inequality~$F_{\mathcal{G}_n}(|a_8\rangle) \leq 1/2$ which is needed here was previously shown in \cite{meloPowerNoisyFermionic2013a}) and~\eqref{eq:dualityextentineq} gives the lower bound~$\xi_{\mathcal{G}_n}(|a_8\rangle) \geq  2$, establishing~\eqref{eq:a8extenteq}.

Comparing the exponents in~\eqref{eq:chitm} and~\eqref{eq:upperbounda8}, we can give the following rough estimate: we expect that the cost (runtime) of simulating $t=48$~$T$-gates in the stabilizer framework is comparable with the cost of simulation of roughly~$t'\approx 10$ gates~$U_{\textrm{magic}}$ in the fermionic context.    
Similar estimates apply to other non-Gaussian operations such that the SWAP gate, a magic gate in the context of fermionic Gaussian computation. The latter can be implemented using a certain non-Gaussian state as shown in~\cite{PhysRevLett.123.080503}.

\subsection{Efficient additive-error strong simulation}

In a different direction of generalization, building upon the work~\cite{bravyiImprovedClassicalSimulation2016} and making innovative use of a concentration inequality by Hayes~\cite{Hayes} for vector martingales, Ref.~\cite{PRXQuantum.3.020361} gives a randomized algorithm which, for a state~$\Psi$ obtained by applying~$n$-qubit Clifford gates and~$t$ (non-Clifford)~$T$-gates to~$\ket{0}^{\otimes n}$, provides an additive-error estimate~$\hat{p}(x)$ for the probability~$p(x)=\|(\langle x|\otimes I^{\otimes (n-a)})|\Psi\rangle\|^2$ of observing~$a$~qubits in the state~$\ket{x}$, with~$x\in \{0,1\}^a$. The algorithm is based on a procedure by which the probability~$p(x)$ of interest is expressed in terms of the squared norm~$\left\|\left(\bra{0}^{\otimes t-r}\otimes I^{\otimes r}\right)W\ket{\Psi}\right\|^2$ of a partially projected state, where~$\Psi$ is a tensor product of~$t$ non-stabilizer single-qubit states (arising from gadgetization), $W$ a certain Clifford unitary, and~$r$ a circuit-dependent integer. The failure probability of the constructed algorithm is then upper bounded (see~\cite[Theorem 3]{PRXQuantum.3.020361}) by an expression depending on~$p(x)$, the error~$\epsilon$,  the stabilizer rank~$\xi_{\stab_n}(\Psi)$ (taking the role of~$\chi$) of the product state~$\Psi$, as well as two additional parameters than be chosen freely, but enter into the (polynomial) runtime estimate.

The described method of adapting fast algorithms for simulating Clifford circuits with non-stabilizer initial states can be applied in a similar manner to this algorithm, since this also reduces to computing inner products (including phases) between Gaussian states.

\section{Multiplicativity of the Gaussian fidelity for~$4$ fermions\label{sec:gaussianoverlap}}

The main result of this section is a proof that the fermionic Gaussian fidelity is multiplicative for the tensor product of any two positive-parity 4-fermion states. 

We begin in Section~\ref{sec:4fermionstates} by laying out some specific properties of~$4$-fermion states. We discuss a Gaussianity condition specific to 4-fermion states \cite{oszmaniecClassicalSimulationFermionic2014a} and we write an explicit expression for any 4-fermion state as a superposition of two orthogonal (Gaussian) states. This was first introduced in Refs.~\cite{oszmaniecClassicalSimulationFermionic2014a,kusClassicalQuantumStates2009}. In Section~\ref{sec:gaussianoverlap4fermionstates} we establish properties of the fermionic Gaussian fidelity for 4-fermion states which are subsequently used in Section~\ref{sec:multoverlap4fermionstatesSubsection} to prove that the fermionic Gaussian fidelity is multiplicative for the tensor product of any two 4-fermion states. 

\subsection{Four-fermion Gaussian and non-Gaussian states \label{sec:4fermionstates}}

Key to our considerations is a certain antiunitary map~$\theta$ acting on~$\cHeven^4$, the positive-parity subspace of 4 fermions spanned by~$\{ \ket{x} \}_{x\in\sbin^4_+}$. It is defined
by its action
\begin{align}
\theta \ket{x}&=(-1)^{\vartheta(x)}\ket{\overline{x}} \ ,
\label{eq:mapThetaDef}
\end{align}
for~$x\in \sbineven^4$, on basis states (antilinearly extended to all of~$\cHeven^4$), where~$\vartheta(x) = x_1+x_3 \mod 2= x_2+x_4 \mod 2 = \vartheta(\bar{x})$.  
Here~$\overline{x} = (\overline{x}_1, \ldots, \overline{x}_n)$ is obtained by flipping each bit of~$x$.
The relevant properties of this map are the following. 
 We note that the following statement has been given in~\cite[Eq. (9)]{oszmaniecClassicalSimulationFermionic2014a},
along with a negative-parity version.
\begin{lemma}[\!\!\cite{oszmaniecClassicalSimulationFermionic2014a}]
\label{lem:thetamapGaussianity}
A state~$\Psi\in\cHeven^4$ is Gaussian if and only if
\begin{align}
\langle \Psi,\theta\Psi\rangle &=0\ .
\end{align}
\end{lemma}

\begin{proof}
This follows from the Gaussianity criterion given in Lemma~\ref{lem:gaussianitylambdapurestates}. We give the proof in Appendix~\ref{sec:thetamapGaussianity}.
\end{proof}

\begin{lemma}
\label{lem:thetaCCcommute}
We have~$\theta c_j c_k = c_j c_k \theta~$ for all~$j,k\in[8]$.
\end{lemma}

\begin{proof}
    See Appendix~\ref{sec:thetaCCcommute}.
\end{proof}

The following result was first shown in Ref.~\cite{kusClassicalQuantumStates2009}. 

\begin{lemma}[\hspace{-0.05mm}\cite{kusClassicalQuantumStates2009}]\label{lem:evenstatedecompositionone}
Let~$\Psi\in\cHeven^4$ be a unit vector. Then
there are two orthogonal unit vectors~$\Psi_1,\Psi_2\in \cHeven^4$, $\varphi\in [0,2\pi)$ and~$a\in [0,1/\sqrt{2}]$ such that 
\begin{align}
\theta\Psi_j&=\Psi_j\qquad\textrm{ for }\qquad j\in[2] \label{eq:stabilizationpropertypsi}
\end{align}
and
\begin{align}
\Psi &=e^{i\varphi}\left(\sqrt{1-a^2}\Psi_1+i a\Psi_2\right)\ . \label{eq:psiidentityaia}
\end{align}
\end{lemma}
\begin{proof}
We first argue that it suffices to consider the case where~$\Psi$ satisfies
\begin{align}
\langle \Psi,\theta\Psi\rangle\in\mathbb{R}\ .\label{eq:psithetapsi}
\end{align}
This is because
\begin{align}
\langle (e^{i\varphi}\Psi),\theta (e^{i\varphi}\Psi)\rangle&=e^{-2i\varphi}\langle \Psi,\theta\Psi\rangle\qquad\textrm{ for every }\qquad \varphi\in [0,2\pi)\ ,
\end{align}
which implies that~\eqref{eq:psithetapsi} can be ensured by replacing~$\Psi$ with~$e^{i\varphi}\Psi$ for a suitably chosen~$\varphi\in [0,2\pi)$. 

Let~$\Psi$ be such that~\eqref{eq:psithetapsi} holds. We define ``real'' and ``imaginary'' parts of~$\Psi$ by the expressions
\begin{align}
\Psi_R&=\frac{1}{2}\left(\Psi+\theta\Psi\right)\\
\Psi_I&=\frac{1}{2i}\left(\Psi-\theta\Psi\right)\ .
\end{align}
It follows immediately from this definition that
\begin{align}
\Psi &=\Psi_R+i\Psi_I\ 
\end{align}
and
\begin{align}
\theta \Psi_R&=\Psi_R\\
\theta\Psi_I&=\Psi_I\ .
\end{align}
 because~$\theta$ is antiunitary and an involution.
 Furthermore,  Eq.~\eqref{eq:psithetapsi} implies that  the vectors~$\Psi_R$,$\Psi_I$ are orthogonal: We have
 \begin{align}
 4i \langle \Psi_R,\Psi_I\rangle &=\langle \Psi+\theta\Psi,\Psi-\theta\Psi\rangle\\
 &=\|\Psi\|^2-\|\theta\Psi\|^2+\langle \theta\Psi,\Psi\rangle-\langle\Psi,\theta\Psi\rangle\\
 &=-2i\mathsf{Im} \langle \Psi,\theta\Psi\rangle\\
 &=0
 \end{align}
where we used that~$\theta$ is an antiunitary in the first step, and assumption~\eqref{eq:psithetapsi} in the last step. 
The claim now follows by setting  
\begin{align}
\left(a,\Psi_1,\Psi_2\right)=
\begin{cases}
\left(\|\Psi_I\|,\frac{\Psi_R}{\|\Psi_R\|},\frac{\Psi_I}{\|\Psi_I\|}\right)\qquad\textrm{ if }\qquad \|\Psi_I\|\leq 1/\sqrt{2}\\
\left(\|\Psi_R\|,\frac{\Psi_I}{\|\Psi_I\|},\frac{\Psi_R}{\|\Psi_R\|}\right)\qquad\textrm{ otherwise }\ .
\end{cases}
\end{align}

\end{proof}

\begin{theorem}[\hspace{-0.1mm}\cite{oszmaniecClassicalSimulationFermionic2014a,kusClassicalQuantumStates2009}]\label{thm:generalGaussianstateorbit}
Let~$\Psi\in \cHeven^4$ be an arbitrary unit vector. Then there are a Gaussian pure state~$\Psi_g\in\cGeven_4$, $\varphi\in [0,2\pi)$ and~$f\in [1/2,1]$ such that  the state~$\theta\Psi_g$ is Gaussian and orthogonal to~$\Psi_g$ and
\begin{align}
\Psi&=e^{i\varphi}\left(\sqrt{f}\Psi_g+\sqrt{1-f} \theta \Psi_g\right)\ .\label{eq:psidecompositionpq}
\end{align}
The triple~$\left(\Psi_g,\varphi,f\right)$ is uniquely defined by~$\Psi$, i.e., a function of~$\Psi$. 
Furthermore, the quantity~$f=f(\Psi)$ is invariant under the action of Gaussian unitaries associated with special orthogonal rotations: We have 
\begin{align}
f(U\Psi)&=f(\Psi)\qquad\textrm{ for any Gaussian unitary }U = U_R 
\text{ with } R \in SO(2n) \ .
\label{eq:invariancepropertypUpsi}
\end{align}

\end{theorem}

\begin{proof}
Let~$\Psi\in\cHeven^4$ be an arbitrary unit vector. Let~$\Psi_1,\Psi_2\in\cH_+^4$, $\varphi\in [0,2\pi)$ and~$a\in [0,1]$ be as in Lemma~\ref{lem:evenstatedecompositionone}. 
Define
\begin{align}
\Psi_g^{\pm} &=\frac{1}{\sqrt{2}} \left(\Psi_1\pm i \Psi_2\right)\ .
\end{align}
Then 
\begin{align}
\theta\Psi_g^-=\Psi_g^+\label{eq:psigplusminusx}
\end{align}
because of property~\eqref{eq:stabilizationpropertypsi}. It follows that
\begin{align}
\langle\Psi_g^-,\Psi_g^+\rangle
&=\langle\Psi_g^-,\theta\Psi_g^-\rangle\\
&=\frac{1}{2}\left(\|\Psi_1\|^2+2i\mathsf{Re}\langle\Psi_1,\Psi_2\rangle-\|\Psi_2\|^2\right)\\
&=0
\end{align}
since~$\Psi_1$ and~$\Psi_2$ are orthogonal unit vectors. 

Using~\eqref{eq:psigplusminusx} and the orthogonality of~$\Psi_g^-,\Psi_g^+$ implies that 
\begin{align}
\langle \Psi_g^+,\theta\Psi_g^+\rangle &=\langle \Psi_g^+,\Psi_g^-\rangle=0
\end{align}
and similarly (because~$\theta$ is an involution)
\begin{align}
\langle \Psi_g^-,\theta\Psi_g^-\rangle &=\langle \Psi_g^-,\Psi_g^+\rangle=0\ .
\end{align}
According to the Gaussianity criterion in Lemma~\ref{lem:thetamapGaussianity}, we conclude that both~$\Psi_g^+$ and~$\Psi_g^-$ are Gaussian.

Rewriting Eq.~\eqref{eq:psiidentityaia} by expressing~$(\Psi_1,\Psi_2)$ in terms of~$(\Psi_g^+,\Psi_g^-)$ 
gives
\begin{align}
\Psi &=e^{i\varphi}\left(\sqrt{f}\Psi_g^-+\sqrt{1-f}\Psi_g^+\right)\qquad\textrm{ where }\qquad f=\frac{1}{2}+a\sqrt{1-a^2}\ .
\end{align}
The claim~\eqref{eq:psidecompositionpq} now follows with~\eqref{eq:psigplusminusx}.

It remains to show property~\eqref{eq:invariancepropertypUpsi} of the function~$f$. This follows immediately from the fact that the antiunitary~$\theta$ commutes with all quadratic monomials~$c_jc_k$ of Majorana operators (see Lemma~\ref{lem:thetaCCcommute}), and hence with any Gaussian unitary~$U=U_R$ with~$R\in SO(2n)$, i.e., $U\theta=\theta U$.
Retracing the steps of the proof, it is easy to check that if~$(\Psi_1,\Psi_2)$ are the states of Lemma~\ref{lem:evenstatedecompositionone}, and~$\Psi_g$ the state in expression~\eqref{eq:psidecompositionpq} for~$\Psi$, then the corresponding states~$(\Psi'_1,\Psi'_2)$ and~$\Psi'_g$ for the state~$\Psi'=U\Psi$ are given by~$\Psi'_j=U\Psi_j$ for~$j\in[2]$ and~$\Psi_g'=U\Psi_g$, respectively. This implies the claim.
\end{proof}

\subsection{The Gaussian fidelity for 4-fermion states \label{sec:gaussianoverlap4fermionstates}}

For a subset~$\Gammaset\subset\sbin^4$, we define~$\overline{\Gammaset}=\left\{\overline{x}\ |\ x\in \Gammaset\right\}$. We also write
$\Pi_\Gammaset=\sum_{x\in\Gammaset}\proj{x}$ for the projection  onto the span of~$\{\ket{x}\}_{x\in\Gammaset}$. 
\begin{lemma}\label{lem:generalstateprojectionoverlap}
Let~$\Gammaset\subset\sbineven^4$, $|\Gammaset|=4$ be a subset of even-weight strings such that~$\Gammaset\cup \overline{\Gammaset}=\sbin^4_+$. 
Let~$f(\Psi)\in [1/2,1]$ be defined as in Theorem~\ref{thm:generalGaussianstateorbit}. 
Then 
\begin{align}
\left\|\Pi_\Gammaset\Psi\right\|^2\leq f(\Psi)\qquad\textrm{ for any unit vector }\qquad \Psi\in\cHeven^4\ .\label{eq:upperboundpPsismallf}
\end{align}
\end{lemma}

\begin{proof}
Let~$f=f(\Psi)\in [1/2,1]$, $\varphi\in [0,2\pi)$ and~$\Psi_g\in \cGeven_n$ be  as in Theorem~\ref{thm:generalGaussianstateorbit} such that
\begin{align}
\Psi&=e^{i\varphi}\left(\sqrt{f}\Psi_g+\sqrt{1-f}\theta\Psi_g\right)\  .\label{eq:psidecompositionpqsec}
\end{align}
We  define
\begin{align}
\begin{matrix}
\alpha_x&=&\langle x,\Psi_g\rangle\\
\beta_x&=&\langle x,\theta\Psi_g\rangle
\end{matrix}\qquad\textrm{ for every }\qquad x\in \Gammaset\ .
\end{align}
We claim that we have the identities
\begin{align}
\sum_{x\in \Gammaset} (|\alpha_x|^2+|\beta_x|^2)&=1\label{eq:normconditionalphabeta}\\
\sum_{x\in \Gammaset} \overline{\alpha_x}\beta_x&=0\label{eq:orthogonalityconditionalphabeta}\ .
\end{align}
Observe that these two identities immediately imply~\eqref{eq:upperboundpPsismallf}: Using expression~\eqref{eq:psidecompositionpqsec}, we have
\begin{align}
\|\Pi_\Gammaset\Psi\|^2&=\sum_{x\in \Gammaset} |\langle x,\Psi\rangle|^2\\
&=\sum_{x\in \Gammaset} |\sqrt{f} \alpha_x+\sqrt{1-f} \beta_x|^2\\
&=\left\|\sqrt{f}\vec{\alpha}+\sqrt{1-f}\vec{\beta}\right\|^2\ \label{eq:PigammaPsi}
\end{align}
where we defined the vectors~$\vec{\alpha}=(\alpha_x)_{x\in\Gammaset},\vec{\beta}=(\beta_x)_{x\in\Gammaset}\in \mathbb{C}^{4}$.  
Since~\eqref{eq:normconditionalphabeta} and~\eqref{eq:orthogonalityconditionalphabeta}
are equivalent to the statement that
\begin{align}
\|\vec{\alpha}\|^2+\|\vec{\beta}\|^2&=1\label{eq:normalizationalphabeta}\\
\langle \vec{\alpha},\vec{\beta}\rangle&=0\label{eq:orthogonalityalphabeta} ,
\end{align}
we obtain
\begin{align}
\left\|\sqrt{f}\vec{\alpha}+\sqrt{1-f}\vec{\beta}\right\|^2&= f\|\vec{\alpha}\|^2+(1-f)\|\vec{\beta}\|^2\\
&\leq \max\{f,1-f\}\\
&=f\ \label{eq:upperboundponeminusp}
\end{align}
by the Pythagorean theorem in~$\mathbb{C}^4$ (using~\eqref{eq:orthogonalityalphabeta}) and by maximizing over~$(\alpha,\beta)$ satisfying~\eqref{eq:normalizationalphabeta}. Inserting~\eqref{eq:upperboundponeminusp} into~\eqref{eq:PigammaPsi} results in the upper bound~\eqref{eq:upperboundpPsismallf} on~$\|\Pi_\Gammaset\Psi\|^2$. 

It remains to prove the claimed identities~\eqref{eq:normconditionalphabeta} and~\eqref{eq:orthogonalityconditionalphabeta}. We argue that these are a consequence of the fact that~$\Psi_g$ is normalized and Gaussian, respectively.
\begin{proof}
Observe that by definition of the antiunitary~$\theta$, we have 
\begin{align}
\beta_x &=\langle x,\theta\Psi_g\rangle\\
&=(-1)^{\vartheta(\overline{x})}\overline{\langle\overline{x},\Psi_g\rangle}\qquad\textrm{ for all }\qquad x\in \Gammaset\ .
\end{align}
In particular, this implies that
\begin{align}
|\beta_x|^2=|\langle \overline{x},\Psi_g\rangle|^2\qquad\textrm{ for every }\qquad x\in \Gammaset\ .\label{eq:betaxnormalisation}
\end{align}
Eq.~\eqref{eq:normconditionalphabeta} now follows from the fact that~$\Psi_g$ is normalized and positive-parity: we have 
\begin{align}
\sum_{x\in \Gammaset}\left(|\alpha_x|^2+|\beta_x|^2\right)
&=\sum_{x\in \Gammaset} \left(|\langle x,\Psi_g\rangle|^2+|\langle \overline{x},\Psi_g\rangle|^2\right)\\
&=\|\Psi_g\|^2\\
&=1
\end{align}
where we used the definition of~$\alpha_x$ and~\eqref{eq:betaxnormalisation} in the first step, and the assumption~$\Gammaset\cup\overline{\Gammaset}=\sbin^4_+$ in the second identity.

Similarly, Eq.~\eqref{eq:orthogonalityconditionalphabeta} is a consequence of the fact that~$\Psi_g$ is Gaussian: we have
\begin{align}
\sum_{x\in \sbineven^4}  \overline{\alpha_x}\beta_x &=\sum_{x\in\sbineven^4} \langle \Psi_g,x\rangle\langle x,\theta\Psi_g\rangle\\
&=\langle \Psi_g,\theta\Psi_g\rangle\\
&=0
\end{align}
where we used the definition of~$\alpha_x$ and~$\beta_x$ in the first step, the fact that~$\Psi_g\in\cHeven^4$ in the second step, and the  characterization of Gaussianity from Lemma~\ref{lem:thetamapGaussianity} in the last identity.
\end{proof}
\end{proof}

Lemma~\ref{lem:generalstateprojectionoverlap} immediately implies the following expression for the  fermionic Gaussian fidelity. We note that a more general expression for the ``Gaussian fidelity'' of a mixed state has previously been obtained in~\cite{oszmaniecClassicalSimulationFermionic2014a}. 
The proof for pure states given here is more elementary and illustrates the 
use of Lemma~\ref{lem:generalstateprojectionoverlap}.
\begin{theorem}[Fermionic Gaussian fidelity for~$4$-mode pure states \cite{oszmaniecClassicalSimulationFermionic2014a,kusClassicalQuantumStates2009}]\label{thm:gaussianoverlapexpression}
Let~$\Psi\in\cHeven^4$ be a unit vector. Let~$f(\Psi)\in [1/2,1]$ be defined as in  Theorem~\ref{thm:generalGaussianstateorbit}. 
Then
\begin{align}
F_{\cGeven_4}(\Psi)&=f(\Psi)\ .\label{eq:gaussianoverlapexplicitexpr}
\end{align}
\end{theorem}
\begin{proof}

Let~$f=f(\Psi)$, $\varphi\in [0,2\pi)$ and~$\Psi_g\in\cGeven_4$ be 
as in Theorem~\ref{thm:generalGaussianstateorbit}. Then we have  
\begin{align}
F_{\cGeven_4}(\Psi)&\geq \left|\langle \Psi_g,e^{i\varphi}\left(\sqrt{f} \Psi_g+f \theta \Psi_g\right)\rangle\right|^2\\
&=f\ 
\end{align}
since~$\theta\Psi_g$ is orthogonal to~$\Psi_g$. It thus suffices to show the upper bound
\begin{align}
F_{\cGeven_4}(\Psi)&\leq f\ .\label{eq:fcgevenupperbound}
\end{align}
Let~$\Phi_g\in \cGeven_4$ be an arbitrary positive-parity Gaussian pure state. Then there is a Gaussian unitary~$U=U_R$ with~$R\in SO(2n)$ and a phase~$\mu\in [0,2\pi)$ such that~$\Phi_g=e^{i\mu}U\ket{0_F}$. 
We will use any subset~$\Gammaset\subset\sbineven^4$ of even-weight strings as in Lemma~\ref{lem:generalstateprojectionoverlap} with the additional property that~$0000\in \Gammaset$, e.g., $\Gammaset=\{0000,1100,1010,1001\}$.  Then~$\ket{0_F}=\Pi_\Gammaset\ket{0_F}$ is in the image of~$\Pi_\Gammaset$. It follows that 
\begin{align}
\left|\langle \Phi_g,\Psi\rangle\right|&=\left|\langle 0_F,U^\dagger \Psi\rangle\right|\\
&=\left|\langle\Pi_\Gammaset 0_F,U^\dagger \Psi\rangle\right|\\
&=\left|\langle 0_F,\Pi_\Gammaset  U^\dagger \Psi\rangle\right|\\
&\leq \left\|\Pi_\Gammaset  U^\dagger \Psi\right\|\\
&\leq \sqrt{f(U^\dagger\Psi)}\ ,
\end{align}
where we used the Cauchy-Schwarz inequality in the penultimate step, and Lemma~\ref{lem:generalstateprojectionoverlap} applied to the state~$U^\dagger\Psi$. Since~$\Phi_g\in\cGeven_4$ was arbitrary, the claimed inequality~\eqref{eq:fcgevenupperbound}
 follows by taking the square and using that
$f(U^\dagger \Psi)=f(\Psi)$, see Eq.~\eqref{eq:invariancepropertypUpsi} of Theorem~\ref{thm:generalGaussianstateorbit}.
\end{proof}

Combining Lemma~\ref{lem:generalstateprojectionoverlap} with Theorem~\ref{thm:gaussianoverlapexpression} yields the following statement, which directly relates the weight of a state on certain subspaces to the fermionic Gaussian fidelity. It will be our main technical tool in the following.
\begin{corollary}\label{cor:generalstateprojectionoverlap}
Let~$\Gammaset\subset\sbineven^4$, $|\Gammaset|=4$ be a subset of even-weight strings such that~$\Gammaset\cup \overline{\Gammaset}=\sbin^4_+$. Then 
\begin{align}
\left\|\Pi_\Gammaset\Psi\right\|^2\leq F_{\cGeven_4}(\Psi)\qquad\textrm{ for any unit vector }\qquad \Psi\in\cHeven^4\ .\label{eq:upperboundpPsi}
\end{align}
\end{corollary}

\subsection{Multiplicativity of the  Gaussian fidelity for 4-fermion states
\label{sec:multoverlap4fermionstatesSubsection}
}

Here we prove that the fermionic Gaussian fidelity is multiplicative for~$4$-fermion states (see Theorem~\ref{thm:multiplicativityGaussianoverlap}).
We use two intermediate results stated as Lemmas~\ref{lem:overlapmultaux1} and \ref{lem:overlapmultaux2boundm}. In Lemma~\ref{lem:overlapmultaux1}, we bound the overlap of a tensor product of two (arbitrary) positive-parity pure states~$\Psi_A,\Psi_B$ with a state~$\Phi$ written as a Schmidt decomposition of a bipartite fermionic pure state. In Theorem~\ref{thm:multiplicativityGaussianoverlap}, this result is used to bound the fermionic Gaussian fidelity. More specifically, Lemma~\ref{lem:overlapmultaux2boundm} is used to upper bound the Schmidt coefficients, giving the multiplicativity result for the Gaussian fidelity.

\begin{lemma}
\label{lem:overlapmultaux1}
Let~$\{m_x\}_{x\in \sbin^4}\subset\mathbb{C}$ be arbitrary. Define
\begin{align}
\label{lem:overlapmultaux1eq1}
\ket{\Phi}=\sum_{x\in \sbin^4}m_x \ket{x,x} \in \cH^8_+\  .
\end{align}
Let~$\Gammaset\subset\sbineven^4$, $|\Gammaset|=4$ be a subset of even-weight strings such that~$\Gammaset\cup \overline{\Gammaset}=\sbin^4_+$. Then
\begin{align}
|\langle\Phi, \Psi_A\motimes\Psi_B\rangle|^2 &\leq F_{\cGeven_4}(\Psi_A)F_{\cGeven_4}(\Psi_B) \left(\max_{x\in \Gammaset} |m_x|+\max_{y\in \overline{\Gammaset}} |m_y|\right)^2
\end{align}
for all states~$\Psi_A,\Psi_B\in \cHeven^4$.

\end{lemma}
\begin{proof}
Because~$\Psi_A$ and~$\Psi_B$ are supported on~$\cHeven^4$ by assumption, we have by Eq.~\eqref{eq:phaseintensorproduct} 
\begin{align}
\langle\Phi,\Psi_A\motimes\Psi_B\rangle &=\sum_{x\in \sbineven^4} m_x(-1)^{|x|} \langle x,\Psi_A\rangle\langle x,\Psi_B\rangle\\
&=\sum_{x\in \sbineven^4} m_xe^{i\nu_x} \langle \Psi_A,x\rangle\langle x,\Psi_B\rangle\label{eq:nuxmfdef}
\end{align}
where~$\nu_x$ is defined by the identity
\begin{align}
(-1)^{|x|}\langle x,\Psi_A\rangle&=e^{i\nu_x}\langle\Psi_A,x\rangle\qquad\textrm{ for }\qquad x\in \sbineven^4\ .
\end{align}
Defining the operator
\begin{align}
M_\Omega=\sum_{x\in \Omega} m_xe^{i\nu_x} \proj{x}
\end{align}
for any subset~$\Omega\subset \sbineven^4$, it follows from Eq.~\eqref{eq:nuxmfdef} that 
\begin{align}
\langle\Phi,\Psi_A\motimes\Psi_B\rangle &=\langle \Psi_A,M_\Gammaset\Psi_B\rangle+\langle \Psi_A,M_{\overline{\Gammaset}}\Psi_B\rangle\ .\label{eq:phipsiaboverlap}
\end{align}
Since~$M_\Gammaset$ is supported on~$\mathsf{span}\{\ket{x}\}_{x\in\Gammaset}$, we have
$\langle \Psi_A,M_\Gammaset\Psi_B\rangle=\langle \Pi_\Gammaset \Psi_A,M_\Gammaset\Pi_\Gammaset\Psi_B\rangle$.
With the Cauchy-Schwarz inequality and the definition of the operator norm~$\|M_\Gammaset\|$
 we thus get
\begin{align}
\left|\langle
\Psi_A,M_\Gammaset\Psi_B\rangle\right|&\leq  \|\Pi_\Gammaset\Psi_A\|\cdot \|\Psi_\Gammaset\Psi_B\|\cdot \|M_\Gammaset\|\\
&\leq \sqrt{F_{\cGeven_4}(\Psi_A)}\cdot \sqrt{F_{\cGeven_4}(\Psi_B)}\cdot \|M_\Gammaset\|\ ,\label{eq:inequalityfaone}
\end{align}
where we applied Corollary~\ref{cor:generalstateprojectionoverlap}. Identical reasoning applies to~$\overline{\Gammaset}$ and yields the inequality
\begin{align}
\left|\langle\Psi_A,M_{\overline{\Gammaset}}\Psi_B\rangle\right|&\leq \sqrt{F_{\cGeven_4}(\Psi_A)}\cdot \sqrt{F_{\cGeven_4}(\Psi_B)}\cdot \|M_{\overline{\Gammaset}}\|\ .\label{eq:inequalityfatwo}
\end{align}
Combining Eqs.~\eqref{eq:inequalityfaone}, \eqref{eq:inequalityfatwo} with Eq.~\eqref{eq:phipsiaboverlap}, we conclude that
\begin{align}
\left|\langle\Phi,\Psi_A\motimes\Psi_B\rangle \right|&\leq 
\left|\langle \Psi_A,M_\Gammaset\Psi_B\rangle\right|+\left|\langle \Psi_A,M_{\overline{\Gammaset}}\Psi_B\rangle\right|\\
&\leq \sqrt{F_{\cGeven_4}(\Psi_A) F_{\cGeven_4}(\Psi_B)}\left(\|M_{\Gammaset}\|+\|M_{\overline{\Gammaset}}\|\right)\ .
\end{align}
Taking the square and observing that
\begin{align}
\|M_\Omega\|&=\max_{x\in \Omega} |m_xe^{i\nu_x}|=\max_{x\in \Omega}|m_x|\ 
\end{align}
gives the claim.
\end{proof}

The following lemma will be useful to prove the main theorem.

\begin{lemma}
\label{lem:overlapmultaux2boundm}
The function
\begin{align}
f(\theta,x)&=\prod_{j=1}^4 (\cos\theta_j)^{1-x_j}(\sin\theta_j)^{x_j}\qquad\textrm{ for }\qquad \theta=(\theta_1,\ldots,\theta_4)\in\mathbb{R}^4\textrm{ and }x\in \{0,1\}^4 
\end{align}
satisfies
\begin{align}
|f(\theta,x)|+|f(\theta,y)|&\leq 1
\end{align}
for all~$\theta\in \mathbb{R}^4$ and~$x,y\in \{0,1\}^4$ with~$x,y$ even-weight and~$x\neq y$.
\end{lemma}
\begin{proof}
Because~$x,y$ have even and different weight, it suffices to consider two cases, namely with~$|x-y|\in \{2,4\}$.

Consider first the case where~$|x-y|=2$. Without loss of generality, assume that $(x_1,x_2)=(y_1,y_2)$, $x_3\neq y_3$, $x_4\neq y_4$. 
Since translating~$\theta$ by~$-\pi/2$ interchanges~$\left|\sin\theta\right|$ and~$\left|\cos\theta\right|$, it suffices to show the claim for~$x=(0,0,0,0)$ and~$y=(0,0,1,1)$. In this case we have
\begin{align}
|f(\theta,x)|+|f(\theta,y)|&=
\left|\cos\theta_1\cos\theta_2\cos\theta_3\cos\theta_4\right|+
\left|\cos\theta_1\cos\theta_2\sin\theta_3\sin\theta_4\right|\\
&=\left|\cos\theta_1\cos\theta_2\right|\cdot \left(|\cos\theta_3\cos\theta_4|+|\sin\theta_3\sin\theta_4|\right)\\
&\leq \left|\cos\theta_1\cos\theta_2\right|\\
&\leq 1 \ ,
\end{align}
where the first inequality follows from the Cauchy-Schwarz inequality in~$\mathbb{R}^2$. Since~$\theta\in\mathbb{R}^4$ was arbitrary, this concludes the proof for~$|x-y|=2$.

The proof for~$|x-y|=4$, i.e., $y=\overline{x}$ proceeds similarly. Again it suffices to show the claim for~$x=(0,0,0,0)$. In this case
\begin{align}
	|f(\theta, x)| + |f(\theta,y)|  &= | \cos\theta_1 \cos\theta_2  \cos\theta_3  \cos\theta_4 | +  | \sin\theta_1 \sin\theta_2 \sin\theta_3 \sin\theta_4 |  \\
	&\leq | \cos\theta_1 \cos\theta_2 | +  | \sin\theta_1 \sin\theta_2 | \\
	&\leq  1 \ ,
\end{align}
where the first inequality follows from~$|\cos\theta_3  \cos\theta_4 |\leq 1~$ and~$|\sin\theta_3  \sin\theta_4 |\leq 1~$ and the second one from the Cauchy-Schwarz inequality.
\end{proof}

\begin{theorem}[Multiplicativity of the fermionic Gaussian fidelity for 4-mode pure states.]
\label{thm:multiplicativityGaussianoverlap}
Let~$\cH^n_+$ be the set of pure~$n$-fermion states with positive parity and let~$\cG_n^+$ be the set of pure~$n$-fermion Gaussian states with positive parity.
We have that
\begin{align}
    F_{\cG_8}(\Psi_A \tilde{\otimes} \Psi_B) = 
    F_{\cG_4}(\Psi_A)F_{\cG_4}(\Psi_B)\qquad\textrm{ for all }\qquad \Psi_A,\Psi_B\in\cH^4_+\ .
\end{align}
\end{theorem}

\begin{proof}
We first observe that 
$\cH^8$ is a direct sum of the four spaces~$\cH^4_+\motimes\cH^4_+$, $\cH^4_+\motimes\cH^4_-$, 
$\cH^4_-\motimes\cH^4_+$ and~$\cH^4_-\motimes\cH^4_-$. This is because states in these subspaces have different eigenvalues with respect to the corresponding parity operators on the factors (interpreted as Majorana monomials on~$\cH^8$ these are the  monomials~$c(1^80^8)$ and~$c(0^81^8)$). It follows immediately that the overlap with a state of the form~$\Psi_A\motimes\Psi_B\in\cH^4_+\motimes \cH^4_+$ is maximized for a decomposition into states belonging to~$\cH^4_+\motimes\cH^4_+$ only. In particular, it follows that 
\begin{align}
    F_{\cG_8}(\Psi_A \tilde{\otimes} \Psi_B) &=
    F_{\cG_8^+}(\Psi_A \tilde{\otimes} \Psi_B)\qquad\textrm{ for all }\qquad \Psi_A,\Psi_B\in\cH^4_+\ 
\end{align}
and by the same reasoning, we have 
\begin{align}
    F_{\cG_4}(\Psi)=F_{\cG_4^+}(\Psi)\qquad\textrm{ for all }\qquad \Psi\in\cH^4_+\ .
\end{align}
We conclude that it suffices to show that 
\begin{align}
    F_{\cG_8^+}(\Psi_A \tilde{\otimes} \Psi_B) = 
    F_{\cG_4^+}(\Psi_A)F_{\cG_4^+}(\Psi_B)\qquad\textrm{ for all }\qquad \Psi_A,\Psi_B\in\cH^4_+\ .
\end{align}

Let~$\Psi_A,\Psi_B\in\cH^4_+$ be arbitrary. The inequality~$F_{\cG_8^+}(\Psi_A \tilde{\otimes} \Psi_B) \geq F_{\cG_4^+}(\Psi_A)F_{\cG_4^+}(\Psi_B)$ follows trivially from the definition of fermionic Gaussian fidelity in Eq.~\eqref{eq:defmaxoverlap}, because~$\cG_4^+ \tilde{\otimes} \cG_4^+ \subseteq \cG_8^+$.

The inequality~$F_{\cG_8^+}(\Psi_A \tilde{\otimes} \Psi_B) \leq F_{\cG_4^+}(\Psi_A)F_{\cG_4^+}(\Psi_B)$ is a consequence of the Schmidt decomposition for fermionic states put forward in Ref.~\cite{BOTERO200439} and of Lemmas~\ref{lem:overlapmultaux1} and \ref{lem:overlapmultaux2boundm}.
According to Ref.~\cite{BOTERO200439}, 
an arbitrary pure fermionic state~$\Phi \in \cH^n$ admits a Schmidt decomposition of the form
\begin{align}
    |\Phi\rangle=\sum_{x \in\{0,1\}^n} m_x|x, x\rangle
\end{align}
with 
\begin{align}
    m_x = \prod_{j=1}^n (\cos\theta_j)^{1-x_j} (-\sin\theta_j)^{x_j}  
    \qquad\mathrm{ with }\qquad
    \theta_j \in \bbR
    \text{ for }
    j\in[n] \ .
\end{align}
With this definition of~$m_x$ for~$n=4$, an arbitrary state~$\Phi \in \cG_8^+$ can be written as in Eq.~\eqref{lem:overlapmultaux1eq1} and the conditions for Lemma~\ref{lem:overlapmultaux1} apply. 
We have
\begin{align}
    F_{\cG_8^+}(\Psi_A \tilde{\otimes} \Psi_B) 
    &= \max_{\Phi\in\cG_8^+} |\lrangle{\Phi, \Psi_A \otimes \Psi_B}| \\
    &\leq 
    F_{\cGeven_4}(\Psi_A)F_{\cGeven_4}(\Psi_B) \left(\max_{x\in \Gammaset} |m_x|+\max_{y\in \overline{\Gammaset}} |m_y|\right)^2
\end{align}
where~$\Gammaset \subset \sbin^4_+$ with~$|\Gammaset|=4$ a subset of even weight strings such that~$\Gammaset \cup \overline{\Gammaset} = \sbin^4_+$.
Notice that~$x,y$ have even-weight and that~$x\neq y$ because~$\Gammaset$ and~$\overline{\Gammaset}$ are disjoint sets.
Identifying~$m_x$ (whose dependence on~$\theta\in\bbR^4$ is implicit) with~$f(\theta, x)$ in Lemma~\ref{lem:overlapmultaux2boundm} (apart from a minus sign that is not relevant because we take the absolute value) we have
\begin{align}
    F_{\cG_8^+}(\Psi_A \tilde{\otimes} \Psi_B) 
    &\leq 
    F_{\cGeven_4}(\Psi_A)F_{\cGeven_4}(\Psi_B) \left(\max_{x\in \Gammaset} |m_x|+\max_{y\in \overline{\Gammaset}} |m_y|\right)^2 \\
    &\leq F_{\cGeven_4}(\Psi_A)F_{\cGeven_4}(\Psi_B) \ ,
\end{align}
giving the claim.
\end{proof}

\section{Multiplicativity of~$\cD$-fidelity implies that of~$\cD$-extent \label{sec:multoverlapimpmultextent}}

In this section, we show that multiplicativity of the~$\cD$-fidelity implies multiplicativity of the~$\cD$-extent. In Section~\ref{sec:MultFimpMultXiFinite} we prove this for finite dictionaries: This follows immediately from the fact that~$F_\cD(\Psi)$ and~$\xi_\cD(\Psi)$ are related by (convex programming) duality.  
In Section~\ref{sec:MultFimpMultXiInfinite}, we extend this results for infinite, i.e., continuously parameterized dictionaries. We achieve this extension by using (finite)~$\epsilon$-nets for the set of Gaussian states. Similar approaches have been applied in the signal processing context, see e.g., the work~\cite{tangSparseRecoveryContinuous2013}, which shows how to approximately solve atomic norm minimization problems for sparse recovery when the parameters indexing the dictionary lie in a small-dimensional space.

\subsection{Multiplicativity for finite dictionaries \label{sec:MultFimpMultXiFinite}}

We will restrict our attention to finite dictionaries in this section. 
For~$|\cD| < \infty$, the~$\cD$-fidelity is related to the dual formulation of the~$\cD$-extent as (see~\cite[Eq. (3.2)]{heimendahlStabilizerExtentNot2021} and \cite[Theorem 4]{bravyiSimulationQuantumCircuits2019a})
\begin{align}
    \xi_\cD(\Psi)&=\max_{y\in\cH: F_{\cD}(y)\leq 1} 
        |\langle \Psi,y\rangle|^2\ \label{eq:dualformulationxi} \ .
\end{align}
Let~$\cH_1,\cH_2$ and~$\cH_3$ be a triple of Hilbert spaces. 
Let~$\{\cD_j\}_{j\in[3]}$ be a family of dictionaries, where~$\cD_j\subset\cH^j$ for~$j\in[3]$. We assume that
\begin{align}
\cD_1\otimes\cD_2\subseteq\cD_3\ .\label{eq:productinclusionproperty}
  \end{align}
We are interested in the following two properties:
\begin{align}
 \mult^\xi(\cD_1,\cD_2,\cD_3)&:\quad \xi_{\cD_3}(\Psi_1\otimes \Psi_2)=\xi_{\cD_1}(\Psi_1)\xi_{\cD_2}(\Psi_2)\quad\textrm{ for all }\quad\Psi_j\in \cH_j\textrm{ for }j\in[2]\ \\
 \mult^F(\cD_1,\cD_2,\cD_3)&:\quad F_{\cD_3}(\Psi_1\otimes \Psi_2)=F_{\cD_1}(\Psi_1)F_{\cD_2}(\Psi_2)\quad \textrm{ for all }\quad \Psi_j\in \cH_j\textrm{ for }j\in[2]\ .
\end{align}
As an important example, let~$n_j\in\mathbb{N}$ for~$j\in[2]$, $n_3=n_1+n_2$, $\cH_j=(\mathbb{C}^2)^{\otimes n_j}$ and let~$\mathsf{STAB}_n$ be the set of stabilizer states on~$(\mathbb{C}^{2})^{\otimes n}$. Then~$\mult^\xi(\mathsf{STAB}_{n_1},\mathsf{STAB}_{n_1},\mathsf{STAB}_{n_1+n_2})$ does not hold for certain (large) choices of~$n_1$ and~$n_2$~\cite{heimendahlStabilizerExtentNot2021}. On the other hand, for~$n_1,n_2 \leq 3$ the multiplicativity property~$\mult^\xi(\mathsf{STAB}_{n_1},\mathsf{STAB}_{n_1},\mathsf{STAB}_{n_1+n_2})$ holds (see Ref.~\cite[Proposition 1]{bravyiSimulationQuantumCircuits2019a}).
This was shown using that the stabilizer fidelity is multiplicative, i.e., $\mult^F(\mathsf{STAB}_{n_1},\mathsf{STAB}_{n_1},\mathsf{STAB}_{n_1 +n_2})$.

We claim the property~$\mult^F(\cD_1,\cD_2,\cD_3)$ implies property~$\mult^\xi(\cD_1,\cD_2,\cD_3)$.

\begin{theorem}
\label{thm:MultFimpMultXiFinite}
Property~$\mult^F(\cD_1,\cD_2,\cD_3)$ implies  property~$\mult^\xi(\cD_1,\cD_2,\cD_3)$. 
\end{theorem}
\begin{proof}
We clearly have
\begin{align}
    \label{eq:multProofFiniteDicLeq}
    \xi_{\cD_3}(\Psi_1\otimes\Psi_2)&\leq \xi_{\cD_1}(\Psi_1)\xi_{\cD_2}(\Psi_2)
\end{align}
for all~$\Psi_1\in\cH_1$ and~$\Psi_2\in\cH_2$
  because of property~\eqref{eq:productinclusionproperty} of the dictionaries~$\{\cD_j\}_{j=1}^3$ and the definition~\eqref{eq:Loneminimization} of~$\xi_\cD$. To show the converse inequality, assume that~$y_1\in\cH_1$, $y_2\in\cH_2$ are such that 
  \begin{align}
F_{\cD_j}(y_j)\leq 1\quad \textrm{ and }\quad \xi_{\cD_j}(\Psi_j)=|\langle\Psi_j,y_j\rangle|^2\quad\textrm{ for }\quad j\in[2]\ .
  \end{align}
Then
\begin{align}
    F_{\cD_3}(y_1\otimes y_2)&=F_{\cD_1}(y_1)F_{\cD_2}(y_2) \leq 1
\end{align}
  where we used  the assumption that property~$\mult^F(\cD_1,\cD_2,\cD_3)$ holds to obtain the equality.
This implies that~$y_1\otimes y_2$ is a feasible point of the dual program for the quantity~$\xi_{\cD_3}(\Psi_1\otimes\Psi_2)$, see Eq.~\eqref{eq:dualformulationxi}. Thus
\begin{align}
    \xi_{\cD_3}(\Psi_1\otimes\Psi_2)&\geq |\langle \Psi_1\otimes\Psi_2,y_1\otimes y_2\rangle|^2\\
    &=|\langle\Psi_1,y_1\rangle|^2\cdot |\langle\Psi_2,y_2\rangle|^2\\
    &=\xi_{\cD_1}(\Psi_1)\xi_{\cD_2}(\Psi_2)\ .
    \label{eq:multProofFiniteDicGeq}
\end{align}
Expression \eqref{eq:multProofFiniteDicGeq} together with Eq.~\eqref{eq:multProofFiniteDicLeq} gives the claim.
\end{proof}

\subsection{Multiplicativity for infinite dictionaries \label{sec:MultFimpMultXiInfinite}}

In this section, we extend the results of Section~\ref{sec:MultFimpMultXiInfinite} to dictionaries~$\cD$ that may contain infinitely many elements.
Our strategy is to use an~$\epsilon$-net for~$\cD\in\cH$ with a finite number of elements, we denote by~$\cD^\epsilon$.
We relate the extent and fidelity with respect to the dictionary~$\cD$ to the extent and fidelity with respect to its net~$\cD^\epsilon$ (see Lemmas~\ref{lem:Xiapproximate} and \ref{lem:fapproximate}) to prove that multiplicativity of the~$\cD$-fidelity implies multiplicativity of the~$\cD$-extent in Theorem~\ref{thm:MultFimpMultXiInfinite}. This result is a generalization of Theorem~\ref{thm:MultFimpMultXiFinite} (that considered finite dictionaries) for (possibly) infinite dictionaries.

We will make use of the notion of~$\epsilon$-net to replace our infinite set~$\cD$ by a finite set~$\cD^\epsilon$.
 Let~$\|\Psi\|=\sqrt{\langle \Psi,\Psi\rangle}$ for~$\Psi\in\cH$ denote the norm on~$\cH$.
 Let~$\cD\subset \cH$ and let~$\epsilon>0$. 
Then a set~$\cD^{\epsilon}\subset \cH$ is called an~$\epsilon$-net for~$\cD$ if
for any~$\Psi\in\cD$ there is some~$\Phi\in\cD^\epsilon$ such that~$\|\Phi-\Psi\|\leq \epsilon$.

We are interested in the case where for every~$\epsilon>0$ there is a finite~$\epsilon$-net~$\cD^\epsilon$ for~$\cD$, with the additional property that~$\cD^\epsilon\subset\cD$, i.e., the net consists of elements of~$\cD$. A sufficient condition for this being the case is that the subset~$\cD\subset\cH$ is compact.

\begin{lemma}\label{lem:Xiapproximate}
Let~$\cD\subset\cH$ be a set of states. Assume that there is a finite~$\epsilon$-net~$\cD^\epsilon$ for~$\cD$ such that~$\cD^\epsilon\subset \cD$, for some~$\epsilon>0$. Assume further that~$\cD^\epsilon$ contains an orthonormal basis of~$\cH$. Let~$d$ be the dimension of~$\cH$. Then
\begin{align}
\xi_{\cD}(\Psi)&\leq \xi_{\cD^\epsilon}(\Psi)\leq \xi_{\cD}(\Psi)\left(1+\sqrt{d}\epsilon\right)^2\qquad\textrm{ for all }\qquad \Psi\in\cH\ .
\end{align}
\end{lemma}
\begin{proof}
The first inequality follows immediately from the definition of~$\xi_\cD$ and from the assumption that~$\cD^\epsilon\subset\cD$. 

To prove the second inequality, let~$\Psi\in\cH$ be arbitrary. 
By definition of~$\xi_\cD(\Psi)$ as an infimum, we have the following: For every~$m\in\mathbb{N}$, there exist~$N(m)\in\mathbb{N}$, $\{\varphi_j(m)\}_{j=1}^{N(m)}\subset \cD$ and~$\{c_j(m)\}_{j=1}^{N(m)}\subset\mathbb{C}$ such that \begin{align}
\Psi=\sum_{j=1}^{N(m)}c_j(m) \varphi_j(m)
\end{align}
and
\begin{align}
\|c(m)\|_1&< \sqrt{\xi_{\cD}(\Psi)}+\frac{1}{m}\ .\label{eq:supremumproperty}
\end{align}
Furthermore,
we have
\begin{align}
\|c(m)\|_1^2&\geq \xi_{\cD}(\Psi)\quad\textrm{ for any }\quad {N(m)}\in\mathbb{N}, \{s_j\}_{j=1}^{N(m)}\subset \cD\quad\textrm{  and }\quad\{c_j(m)\}_{j=1}^{N(m)}\subset\mathbb{C}\ .
\end{align}
Fix such an~$m\in\mathbb{N}$. Since~$\cD^\epsilon$ is an~$\epsilon$-net for~$\cD$, there is, for every~$j\in[N(m)]$, an element~$\varphi_j^\epsilon(m)\in\cD^\epsilon$ and~$\delta_j(m)\in\cH$ such that
\begin{align}
\varphi_j(m)&=\varphi_j^\epsilon(m)+\delta_j(m)\qquad\textrm{ and }\qquad \|\delta_j(m)\|\leq \epsilon\  .
\end{align}
It follows that
\begin{align}
\Psi&=\left(\sum_{j=1}^{N(m)} c_j(m)\varphi_j^\epsilon(m)\right)+\delta(m)\qquad\textrm{ where }\qquad \delta(m)=\sum_{j=1}^{N(m)} c_j(m)\delta_j(m)\ .
\end{align}
By the triangle inequality we have 
\begin{align}
    \|\delta(m)\|&\leq \epsilon\sum_{j=1}^{N(m)} |c_j(m)|=\epsilon \|c(m)\|_1 \ .
    \label{eq:deltanormproperty}
\end{align}
Suppose~$\{e_k\}_{k=1}^d$ is an orthonormal basis contained in~$\cD^\epsilon$. Then we can expand
\begin{align}
\delta(m) &=\sum_{k=1}^d \alpha_k(m) e_k\ ,
\end{align}
and it follows from~\eqref{eq:deltanormproperty} and the Cauchy-Schwarz inequality that 
\begin{alignat}{2}
\|\alpha(m)\|_1 &\leq \sqrt{d}\|\alpha(m)\|_2 \\
&=\sqrt{d}\|\delta(m)\|_2
&&\qquad\text{by the Cauchy-Schwarz inequality in~$\bbC^d$} \\
&\leq \sqrt{d}\epsilon \|c(m)\|_1\  &&\qquad\text{because of Eq.~\eqref{eq:deltanormproperty}}\ .
\end{alignat}
where~$\|\cdot\|_2$ is the Euclidean norm in~$\mathbb{C}^d$.
In summary, we have 
\begin{align}
\Psi &=\sum_{j=1}^{N(m)} c_j(m)\varphi_j^\epsilon(m)+\sum_{k=1}^d \alpha_k(m) e_k\ ,
\end{align}
with~$(c,\alpha)\in\mathbb{C}^{N(m)+d}$ satisfying
\begin{align}
\|(c,\alpha)\|_1 &= 
\sum_{j=1}^{N(m)} |c_j(m)|+\sum_{k=1}^d|\alpha_k(m)| \\
&\leq \|c(m)\|_1\cdot \left(1+\sqrt{d}\epsilon\right)\\
&\leq \left(\sqrt{\xi_\cD(\Psi)}+\frac{1}{m}\right)\cdot  \left(1+\sqrt{d}\epsilon\right)\ ,
\end{align}
by~\eqref{eq:supremumproperty}. Because~$\varphi_j^\epsilon(m)\in \cD^\epsilon$ for every~$j \in [N(m)]$ and~$e_k\in \cD^\epsilon$ for every~$k\in[d]$, this implies that
\begin{align}
\sqrt{\xi_{\cD^\epsilon}(\Psi)}\leq \left(\sqrt{\xi_{\cD}(\Psi)}+\frac{1}{m}\right)\cdot  \left(1+\sqrt{d}\epsilon\right)\ .
\end{align}
Since~$m\in\mathbb{N}$ was arbitrary, we can take the limit
$m\rightarrow\infty$ and obtain
\begin{align}
\sqrt{\xi_{\cD^\epsilon}(\Psi)}\leq \sqrt{\xi_{\cD}(\Psi)}\cdot  \left(1+\sqrt{d}\epsilon\right)\ . 
\end{align}
This gives the claim.
\end{proof}

\begin{lemma}\label{lem:fapproximate}
Let~$\cD^\epsilon\subset \cD\subset\cH$  be as in Lemma~\ref{lem:Xiapproximate}. Then we have
\begin{align}
\sqrt{F_{\cD^\epsilon}(\Psi)}&\leq \sqrt{F_{\cD}(\Psi)}\leq \sqrt{F_{\cD^\epsilon}(\Psi)}+\|\Psi\|\cdot\epsilon\qquad\textrm{ for all }\qquad \Psi\in\cH\ .
\end{align}
\end{lemma}
\begin{proof}
The first inequality follows trivially from the definitions using~$\cD^\epsilon\subset\cD$. For the second inequality, let~$\Psi\in\cH$ be arbitrary. Let~$\varphi\in\cD$ be such that 
\begin{align}
F_{\cD}(\Psi)&=|\langle \varphi,\Psi\rangle|^2\ .
\end{align}
Then (by the fact that~$\cD^\epsilon$ is an~$\epsilon$-net for~$\cD$ and~$\varphi\in\cD$) there is an element~$\varphi^\epsilon\in \cD^\epsilon$ and~$\delta\in\cH$ such that
\begin{align}
\varphi&=\varphi^\epsilon+\delta\qquad\textrm{ where }\qquad \|\delta\|\leq \epsilon\ .
\end{align}
It follows that 
\begin{align}
\sqrt{F_{\cD}(\Psi)}&=|\langle \varphi,\Psi\rangle|\\
&\leq |\langle \varphi^\epsilon,\Psi\rangle|+ |\langle \delta,\Psi\rangle|\\
&\leq \sqrt{F_{\cD^\epsilon}(\Psi)}+
|\langle \delta,\Psi\rangle|\\
&\leq \sqrt{F_{\cD^\epsilon}(\Psi)}+\|\Psi\|\cdot \|\delta\|\\
&\leq \sqrt{F_{\cD^\epsilon}(\Psi)}+\|\Psi\|\cdot \epsilon
\end{align}
where we used the definition of~$F_{\cD^\epsilon}(\Psi)$ and the Cauchy-Schwarz inequality (in the penultimate step). The claim follows.
\end{proof}

\begin{lemma}\label{lem:boundednessyops}We have 
\begin{align}
\|y\|^2\leq d^2\cdot F_{\cD^\epsilon}(y)\qquad\textrm{ for every }\qquad y\in\cH\ ,
\end{align}
where~$d=\dim\cH$. 
\end{lemma}
\begin{proof}
Let~$\{e_k\}_{k=1}^d$ be an orthonormal basis contained in~$\cD^\epsilon$.  Then 
\begin{align}
|\langle e_k,y\rangle|^2 & \leq F_{\cD^\epsilon}(y)\qquad\textrm{ for every }\qquad k\in[d]
\end{align}
because~$e_k\in\cD^\epsilon$ for every~$k\in [n]$. 
We have
\begin{align}
\|y\|&=\left(\sum_{k=1}^d
|\langle e_k,y\rangle|^2\right)^{1/2}\\
&\leq \sum_{k=1}^d
|\langle e_k,y\rangle|\\
&\leq d\cdot \sqrt{F_{\cD^\epsilon}(y)}
\end{align}
where we used that~$\|v\|_2 \leq \|v\|_1$ for~$v\in \mathbb{C}^d$. The claim follows.
\end{proof}

\begin{lemma}\label{lem:maincombinationlemma}
Let~$d_j=\dim \cH_j$ for~$j\in[2]$. Assuming that the property~$\mult^F(\cD_1,\cD_2,\cD_3)$ holds, we have 
\begin{align}
\xi_{\cD_3}(\Psi_1\otimes\Psi_2) &\geq g(d_1,d_2,d_3,\epsilon) \xi_{\cD_1}(\Psi_1)\xi_{\cD_2}(\Psi_2)
\end{align}
for a function~$g$ which satisfies
\begin{align}
\lim_{\epsilon\rightarrow 0} g(d_1,d_2,d_3,\epsilon)=1\ .
\end{align}
\end{lemma}
\begin{proof}
Suppose~$y_j\in\cH_j$ for~$j\in[2]$ is such that 
\begin{align}
F_{\cD_j^\epsilon}(y_j) &\leq 1 \label{eq:conditionyjupperbound}
\end{align}
and
\begin{align}
\xi_{\cD_j^\epsilon}(\Psi_j)&=|\langle y_j,\Psi_j\rangle|^2\ \label{eq:Xiychoice}
\end{align}
for~$j\in[2]$. Such a pair~$(y_1,y_2)$ exists since~$\cD^\epsilon_1$ and~$\cD^\epsilon_2$ are finite sets and the dual definition definition of the extent in terms of a maximum applies (see Eq.~\eqref{eq:dualformulationxi}).  Equation~\eqref{eq:conditionyjupperbound} implies that
\begin{align}
\|y_j\| &\leq d_j\qquad\textrm{ for }\qquad j\in[2]\ ,
\end{align}
see Lemma~\ref{lem:boundednessyops}.

We have 
\begin{align}
\sqrt{F_{\cD_3^\epsilon}(y_1\otimes y_2)} &\leq \sqrt{F_{\cD_3}(y_1\otimes y_2)}\qquad\qquad\qquad\qquad\qquad\qquad\,\qquad\textrm{ by Lemma~\ref{lem:fapproximate}}\\
&=\sqrt{F_{\cD_1}(y_1)}\cdot\sqrt{F_{\cD_2}(y_2)} \,\,\textrm{ by the assumption that~$\mult^F(\cD_1,\cD_2,\cD_3)$ holds}\\
&\leq \left(\sqrt{F_{\cD^\epsilon_1}(y_1)}+\|y_1\|\cdot\epsilon\right)\left(\sqrt{F_{\cD^\epsilon_2}(y_2)}+\|y_2\|\cdot\epsilon\right)\quad\textrm{ by Lemma~\ref{lem:fapproximate}}\\
&\leq \sqrt{F_{\cD^\epsilon_1}(y_1)} \cdot\sqrt{F_{\cD^\epsilon_2}(y_2)}
\cdot
\left(1+d_1\cdot\epsilon\right)\left(1+d_2\cdot\epsilon\right)\,\quad\textrm{ by Lemma~\ref{lem:boundednessyops}}\\
&\leq (1+d_1\epsilon)(1+d_2\epsilon)\quad\qquad\qquad\qquad\qquad\qquad\qquad\,\,\textrm{by Eq.~\eqref{eq:conditionyjupperbound}}\ . \label{eq:y1y2upperbound}
\end{align}
Defining
\begin{align}
\tilde{y}_j&=\left(1+d_j\epsilon\right)^{-1}y_j\qquad\textrm{ for }\qquad j\in[2]
\end{align}
it follows from Eq.~\eqref{eq:y1y2upperbound}
that
\begin{align}
F_{\cD_3^\epsilon}(\tilde{y}_1\otimes\tilde{y}_2)&\leq 1\ .
\end{align}
By the dual formulation~\eqref{eq:dualformulationxi} of the quantity~$\xi_{\cD^\epsilon_3}(\Psi_1\otimes\Psi_2)$  this implies that 
\begin{alignat}{2}
\xi_{\cD^\epsilon_3}(\Psi_1\otimes\Psi_2)&\geq 
|\langle \tilde{y}_1\otimes\tilde{y}_2,\Psi_1\otimes\Psi_2\rangle|^2\\
&=\left(\prod_{j=1}^2 \left(1+d_j\epsilon\right)^{-2}|\langle y_j,\Psi_j\rangle|^2\right)\\
&=\left(\prod_{j=1}^2 \left(1+d_j\epsilon\right)^{-2}\right)
\xi_{\cD^\epsilon_1}(\Psi_1)\xi_{\cD^\epsilon_2}(\Psi_2)&&\,\,\textrm{ because of Eq.~\eqref{eq:Xiychoice}}\\
&\geq \left(\prod_{j=1}^2 \left(1+d_j\epsilon\right)^{-2}\right)
\xi_{\cD_1}(\Psi_1)\xi_{\cD_2}(\Psi_2)&&\,\,\textrm{ according to Lemma~\ref{lem:Xiapproximate}.}\label{eq:inequalityam}
\end{alignat}
Lemma~\ref{lem:Xiapproximate} also implies that
\begin{align}
\left(1+\sqrt{d_3}\epsilon\right)^2\xi_{\cD_3}(\Psi_1\otimes\Psi_2)\geq \xi_{\cD^\epsilon_3}(\Psi_1\otimes\Psi_2)\ .\label{eq:inequalitybm}
\end{align}
Combining Eqs.~\eqref{eq:inequalityam} and~\eqref{eq:inequalitybm} gives the claim, with 
\begin{align}
g(d_1,d_2,d_3\epsilon)=\left(1+\sqrt{d_3}\epsilon\right)^{-2}\left(\prod_{j=1}^2 \left(1+d_j\epsilon\right)^{-2}\right)\ .
\end{align}
\end{proof}

\begin{theorem}
\label{thm:MultFimpMultXiInfinite}
Let~$\cH_1,\cH_2$ and~$\cH_3$ be a triple of Hilbert spaces and let~$\{\cD_j\}_{j\in[3]}$ be a family of dictionaries, where~$\cD_j\subset\cH_j$. Assume that~$\cD_j$ contains an orthonormal basis of~$\cH_j$, for~$j\in[3]$, and that~$\cD_1 \otimes \cD_2 \subset \cD_3$. 
Assume further that for any~$\epsilon>0$ there is an~$\epsilon$-net~$\cD_j^\epsilon$ for~$\cD_j$
 such that~$\cD_j^\epsilon\subset\cD_j$, i.e., the net consists of elements of~$\cD_j$. 
 Finally, assume that 
 \begin{align}
 F_{\cD_3}(\Psi_1\otimes\Psi_2)&=F_{\cD_1}(\Psi)F_{\cD_2}(\Psi_2)\qquad\textrm{ for all }\qquad \Psi_1\in\cH_1,\Psi_2\in\cH_2\ .
 \end{align}
Then
\begin{align}
\xi_{\cD_3}(\Psi_1\otimes\Psi_2) &= \xi_{\cD_1}(\Psi_1)\xi_{\cD_2}(\Psi_2)\qquad\textrm{ for all }\qquad\Psi_1\in\cH_1\textrm{ and }\Psi_2\in\cH_2\ .
\end{align}
\end{theorem}
\begin{proof}
By (if necessary) replacing~$\cD^\epsilon_j$ by~$\cD^\epsilon_j\cup \{e^{(j)}_k\}_{k=1}^d$, for~$j\in[3]$, where~$\{e^{(j)}_k\}_{k=1}^d$ is an orthonormal basis of~$\cH_j$ with~$e^{(j)}_k\in \cD_j$, for~$n\in[3]$, we have that each~$\cD^\epsilon_j$ is finite and contains an orthonormal basis of the respective space. The inequality
\begin{align}
\xi_{\cD_3}(\Psi_1\otimes\Psi_2) &\geq \xi_{\cD_1}(\Psi_1)\xi_{\cD_2}(\Psi_2)\qquad\textrm{ for all }\qquad\Psi_1\in\cH_1\textrm{ and }\Psi_2\in\cH_2\ 
\end{align} now
follows immediately from Lemma~\ref{lem:maincombinationlemma} by taking the limit~$\epsilon\rightarrow 0$.  The converse inequality is trivial because~$\cD_1 \otimes \cD_2 \subset \cD_3$.
\end{proof}

\section{Multiplicativity of the Gaussian extent for four fermions\label{sec:multextent}}

In this section we prove that the Gaussian extent is multiplicative for the tensor product of any two 4-fermion pure states with positive parity.

\begin{theorem}[Multiplicativity of the Gaussian extent for 4-fermion pure states.]
\label{thm:multiplicativityGaussianextent}
Let~$\cH^4_+$ be the set of pure~$4$-fermion states with positive parity and let~$\cG_n$ be the set of Gaussian states on~$n$ fermions.
Then
\begin{align}
    \xi_{\cG_8}(\Psi_A \tilde{\otimes} \Psi_B) = 
    \xi_{\cG_4}(\Psi_A)\xi_{\cG_4}(\Psi_B)\qquad\textrm{ for all }\qquad \Psi_A, \Psi_B \in \cH^4_+\ .
\end{align}
\end{theorem}

\begin{proof}
Since the  metaplectic representation defines a surjective, continuous map 
\begin{align}
    \begin{matrix}
    f: & [0,2\pi]\times SO(2n) & \rightarrow & \cG_n\\
     & (\varphi,R) & \mapsto &e^{i\varphi} U_R\ket{0_F}
    \end{matrix}
\end{align}
from the compact set~$[0,2\pi]\times SO(2n)$ to~$\cG_n$, the set~$\cG_n\subset\cH^n$ is compact. We also observe that the occupation number states
\begin{align}
    \left\{\ket{x}\ |\ x\in \{0,1\}^n\right\}
\end{align}
form an orthonormal basis contained in~$\cG_n$. By compactness, we conclude that for any~$\epsilon>0$, there is a 
finite~$\epsilon$-net~$\cG_n^\epsilon\subset \cG_n$ consisting of Gaussian states and containing an orthonormal basis of~$\cH^n$. Finally, we note that we also have the inclusion~$\cG_{n_1}\otimes\cG_{n_2}\subset \cG_{n_1+n_2}$ for~$n_1,n_2\in\mathbb{N}$ arbitrary.

Let us now specialize to~$n_1=n_2=4$. In this case, we have multiplicativity of the fermionic Gaussian fidelity by Theorem~\ref{thm:multiplicativityGaussianoverlap}. In particular, the conditions for Theorem~\ref{thm:MultFimpMultXiInfinite} apply and the claim follows.
\end{proof}

\section*{Acknowledgements}
BD and RK gratefully acknowledge support by the European Research Council under grant agreement no.~101001976 (project EQUIPTNT).

\appendix

\section{Alternative Gaussianity condition for $4$-fermion states\label{sec:thetamapGaussianity}}

In the following, we prove Lemma~\ref{lem:thetamapGaussianity}.

\begin{proof}
Consider~$\Psi\in\cHeven^4$. We will show that~$\langle \Psi,\theta\Psi\rangle =0$ is equivalent to~$\Lambda \left(\ket{\Psi}\otimes \ket{\Psi}\right) = 0$. According to Lemma~\ref{lem:gaussianitylambdapurestates}, this is a sufficient and necessary condition for~$\Psi$ to be Gaussian.

Let~$\Gammaset\subset\sbineven^4$, $|\Gammaset|=4$ be a subset of even-weight strings such that~$\Gammaset\cup \overline{\Gammaset}=\sbin^4_+$. 
We first compute~$\lrangle{\Psi, \theta \Psi}$. We have
\begin{align}
	\lrangle{\Psi, \theta \Psi} &= \sum_{x,y\in\{0,1\}^4_+}  \lrangle{\Psi,x} \bra{x} \theta \left( \ket{y} \lrangle{y,\Psi} \right) \\
	&= \sum_{x,y\in\{0,1\}^4_+} \lrangle{\Psi, x} \lrangle{x, \theta y}  \overline{\lrangle{y, \Psi}} \\
	&= \sum_{x,y\in\{0,1\}^4_+} (-1)^{\vartheta(y)} \lrangle{x, \overline{y}} \lrangle{\Psi, x} \lrangle{ \Psi,y} \\
	&= \sum_{x\in\{0,1\}^4_+} (-1)^{\vartheta(x)} \lrangle{\Psi, x} \lrangle{\Psi,\overline{x}} \\
	&= \sum_{x\in\Gammaset} (-1)^{\vartheta(x)} \lrangle{\Psi, x} \lrangle{\Psi,\overline{x}} + (-1)^{\vartheta(\overline{x})} \lrangle{\Psi, \overline{x}} \lrangle{\Psi,x} \\
	&= 2 \sum_{x\in\Gammaset} (-1)^{\vartheta(x)} \lrangle{\Psi, x} \lrangle{\Psi,\overline{x}} 
	\label{eq:thetaMapEquivCond} \ ,
\end{align}
where the second step follows from~$\theta u = \overline{u} \theta$ for all~$u\in\bbC$ due to the antiunitarity of~$\theta$, the third step from Eq.~\eqref{eq:mapThetaDef} and the fourth and final steps from~$\vartheta(\overline{x}) = \vartheta(x)$ for~$x\in \{0,1\}^4_+$.
Thus, 
\begin{align}
\label{eq:gaussianityAuxEquiv}
\lrangle{\Psi, \theta \Psi} = 0 
\qquad \textrm{ if and only if }\qquad 
\sum_{x\in\{0,1\}^4_+} (-1)^{\vartheta(x)} \lrangle{\Psi, x} \lrangle{\Psi,\overline{x}} = 0 \ . 
\end{align}

We proceed to prove that~$\Lambda (\ket{\Psi} \otimes \ket{\Psi}) = 0$ is equivalent to Eq.~\eqref{eq:gaussianityAuxEquiv}.
We start by using Eq.~\eqref{eq:MajoranaToDiracFermion} to write the operator~$\Lambda$ in terms of creation and annihilation operators:
\begin{align}
	\Lambda &= \sum_{j=1}^{8} c_j \otimes c_j \\
    &= \sum_{j=1}^4 \left( c_{2j-1} \otimes c_{2j-1} + c_{2j} \otimes c_{2j}  \right) \\
	&= \sum_{j=1}^4 \left( (a_j + a_j^\dagger) \otimes (a_j + a_j^\dagger) + i (a_j - a_j^\dagger) \otimes i (a_j - a_j^\dagger) \right) \\
	&= 2 \sum_{j=1}^4 \left( a_j \otimes a_j^\dagger + a_j^\dagger \otimes a_j  \right).
\end{align}
Applying this expression to~$ \ket{\Psi} \otimes \ket{\Psi}$ and using Eq.~\eqref{eq:ajOnState} gives
\begin{equation}
\begin{aligned}
	\Lambda \left(\ket{\Psi} \otimes \ket{\Psi}\right) &= 2 \sum_{x,y \in \{0,1\}^4_+} \sum_{j=1}^4  \left( a_j \otimes a_j^\dagger + a_j^\dagger \otimes a_j  \right) \left(\ket{x} \otimes \ket{y} \right)  \lrangle{x,\Psi} \lrangle{y, \Psi}\\
	&= 2  \sum_{x,y \in \{0,1\}^4_+} \sum_{j=1}^4  (-1)^{\eta_j(x+y)} (x_j \overline{y_j} + \overline{x_j} y_j)  \left( \ket{x \oplus e_j} \otimes \ket{y \oplus e_j}  \right) \lrangle{x,\Psi} \lrangle{y, \Psi} \\
	&= 2  \sum_{x,y \in \{0,1\}^4_+} \sum_{j=1}^4  (-1)^{\eta_j(x+y)} (y_j \oplus \overline{x_j}) \left( \ket{x \oplus e_j} \otimes \ket{y \oplus e_j}  \right) \lrangle{x,\Psi} \lrangle{y, \Psi} \\
	&= 2  \sum_{x,y \in \{0,1\}^4_-} \left( \sum_{j=1}^4  (-1)^{\eta_j(x+ y)} (y_j \oplus \overline{x_j}) \lrangle{x\oplus e_j,\Psi} \lrangle{y\oplus e_j, \Psi} \right)  \left( \ket{x} \otimes \ket{y}  \right) 
\end{aligned}
 \label{eq:LambdaHamming2is0Aux}
\end{equation}
where in the third line we used~$u \overline{v} + \overline{u} v = u \oplus \overline{v}$ for all~$u,v\in\{0,1\}$ and in the last line we used~$\left\{ x \oplus e_j  \,|\, x \in \{0,1\}^4_+ \right\} =  \{0,1\}^4_-$ , $\eta_j(x\oplus e_j) = \eta_j(x)$ and~$u\oplus v \oplus 2e_j = u\oplus v$ valid for~$x\in\{0,1\}_+^4$,~$j\in[4]$ and~$u,v\in\sbin$.
It follows that~$\Lambda \left(\ket{\Psi} \otimes \ket{\Psi}\right) = 0$ if and only if
\begin{align}
\label{eq:LambdaZeroConditionN}
	 \sum_{j=1}^4  (-1)^{\eta_j(x+ y)} (y_j \oplus \overline{x_j}) \lrangle{x\oplus e_j,\Psi} \lrangle{y\oplus e_j, \Psi} = 0  \quad\text{ for all }\quad x,y\in \{0,1\}^4_+\ .
\end{align}

Since~$x$ and~$y$ have the same parity, either~$|x-y| = 4$ (i.e., $y = \overline{x}$), $|x-y| = 2$ or~$|x-y| = 0$ (i.e., $y = x$). The expression~\eqref{eq:LambdaZeroConditionN} is non-zero only if~$|x-y| = 4$ (we argue below why). For~$|x-y| = 4$ Eq.~\eqref{eq:LambdaZeroConditionN} becomes
\begin{align}
	\sum_{j=1}^4  (-1)^{\eta_j(x + \bar{x})} \lrangle{x\oplus e_j,\Psi} \lrangle{\bar{x}\oplus e_j, \Psi} = 0 
	\quad\text{ if and only if }\quad
	\sum_{z \in \Gammaset}^4  (-1)^{\vartheta(z)} \lrangle{z,\Psi} \lrangle{z, \Psi} = 0 \ ,
\end{align}
where we used~$\{x \oplus e_j \,|\, j\in[4] \} = \Gammaset$ with~$\Gammaset\subset\sbineven^4$, $|\Gammaset|=4$ a subset of even-weight strings such that~$\Gammaset\cup \overline{\Gammaset}=\sbin^4_+$. We also used that~$(-1)^{\eta_j(x + \overline{x})} = (-1)^{j-1}$ can be replaced by~$(-1)^{\vartheta(x) + \vartheta(z)}$ upon changing the summation over~$j\in[4]$ to a summation over~$z\in\Gammaset$. We recovered the right hand side of Eq.~\eqref{eq:gaussianityAuxEquiv}, proving the claim.

It remains to argue that Eq.~\eqref{eq:LambdaZeroConditionN} is zero for~$|x-y| \in\{0,2\}$. For~$|x-y| = 0$, i.e., $x=y$, Eq.~\eqref{eq:LambdaZeroConditionN} is zero because~$x_j \otimes \overline{x_j} = 0$ for~$j\in[4]$.
We exemplify that terms~$\propto \ket{x} \otimes \ket{y}$ with~$|x-y| = 2$ are zero by considering~$x=1000$ and~$y=0100$. Starting from Eq.~\eqref{eq:LambdaHamming2is0Aux} we obtain
\begin{align}
	\left(\bra{1000} \otimes\bra{0100} \right)
 \Lambda 
 \left(\ket{\Psi} \otimes \ket{\Psi}\right)
	&=  ((-1)^{\eta_1(1100)} + (-1)^{\eta_2(1100)}) \lrangle{0000,\Psi} \lrangle{1100, \Psi} = 0 \ ,
\end{align}
where~$\eta_1(1100) = 0$ and~$\eta_2(1100) = 1$. The remaining cases with~$|x-y| = 2$ proceed similarly.

\end{proof}

\section{Commutativity of the map~$\theta$ and quadratic Majorana monomials\label{sec:thetaCCcommute}}

In the following, we prove Lemma~\ref{lem:thetaCCcommute}.

\begin{proof}
We start by showing that~$\theta= c_1 c_3 c_5 c_7 K$, where~$K$ denotes the antiunitary given by complex conjugation in the number state basis.
For this, it suffices to show that~$\theta$ is antiunitary, which directly follows from unitarity of~$c_1 c_3 c_5 c_7$, and that it satisfies Eq.~\eqref{eq:mapThetaDef}. We show the later using Eqs.~\eqref{eq:MajoranaToDiracFermion} and \eqref{eq:ajOnState}: We have
\begin{align}
\theta \ket{x} &= c_1 c_3 c_5 c_7 K\ket{x} \\
	&= (a_1 + a_1^\dagger)(a_2 + a_2^\dagger)(a_3 + a_3^\dagger)(a_4 + a_4^\dagger) \ket{x} \\
	&= (-1)^{\eta_4(x) + \eta_3(x)+\eta_2(x)+\eta_1(x)} (x_1 + \overline{x_1})(x_2 + \overline{x_2})(x_3 + \overline{x_3})(x_4 + \overline{x_4}) \ket{\overline{x}}\\
	&= (-1)^{\vartheta(x)} \ket{\overline{x}} \ ,
\end{align}
where we used~$ (-1)^{\eta_4(x) + \eta_3(x)+\eta_2(x)+\eta_1(x)} = (-1)^{3 x_1 + 2x_2 + x_3} = (-1)^{x_1+x_3} = (-1)^{\vartheta(x)}$ and~$x_j + \overline{x_j} = 1$ for~$j\in[4]$. 

The result~$\theta c_j c_k = c_j c_k \theta~$ follows from simple algebra considering~$c_{2j} K = K c_{2j}$ and~$c_{2j-1} K = - K c_{2j-1}$. We show these last two equalities by explicitly computing their action on~$x \in \{0,1\}_+^4$:
\begin{align}
	K c_{2j} \ket{x} &= K (a_j + a_j^\dagger) \ket{x}
	= (x_j + \overline{x_j}) \ket{x \oplus e_j}
	= \ket{x\oplus e_j} \ , \\ 
	c_{2j} K \ket{x} &= (a_j + a_j^\dagger) \ket{x} 
	= (x_j + \overline{x_j}) \ket{x \oplus e_j} 
	= K c_{2j} \ket{x}
\end{align}
and 
\begin{align}
	K c_{2j+1} \ket{x} &= K i (a_j - a_j^\dagger) \ket{x}= -i(x_j - \overline{x_j}) \ket{x \oplus e_j} \ , \\
	c_{2j+1} K \ket{x} &= i(a_j - a_j^\dagger) \ket{x} = i(x_j - \overline{x_j}) \ket{x \oplus e_j} = -K c_{2j+1} \ket{x} \ .
\end{align}
\end{proof}

\begin{proof}
We will prove that~$\theta c_j = - c_j \theta$ for~$j\in[8]$, which implies the result. 

We prove this for~$j$ odd, the proof for~$j$ even proceeds similarly.
We use Eq.~\eqref{eq:MajoranaToDiracFermion} to write the Majorana operators as creation and annihilation operators which act on basis states according to Eq.~\eqref{eq:ajOnState}, and we apply~$\theta$ according to Eq.~\eqref{eq:mapThetaDef}:
\begin{align}
\theta c_{2j-1} \ket{x} &= \theta (a_j + a_j^\dagger) \ket{x} \\
	&=  \theta (-1)^{\eta_j(x)}  (x_j + \overline{x_j}) \ket{x \oplus e_j} \\
	&= (-1)^{\eta_j(x)} (-1)^{\vartheta(x\oplus e_j)}  (x_j + \overline{x_j}) \ket{\overline{x \oplus e_j}} \ , \\ \\
c_{2j-1} \theta \ket{x} &= (-1)^{\vartheta(x)} (a_j + a_j^\dagger) \ket{\overline{x}} \\
	&=  (-1)^{\eta_j(\overline{x})}   (-1)^{\vartheta(x)} (\overline{x_j} + x_j)   \ket{\overline{x} \oplus e_j} \ .
\end{align}
The equality~$\theta c_{2j-1} = - c_{2j-1} \theta$ follows from~$\ket{\overline{x} \oplus e_j} = \ket{\overline{x \oplus e_j}}$ for~$j\in[4]$, from~$(-1)^{\eta_j(\overline{x})} = (-1)^{j+1} (-1)^{\eta_j(x)}$ and from~$(-1)^{\vartheta(x)} = (-1)^{j} (-1)^{\vartheta(x\oplus e_j)}$.
\end{proof}

\bibliographystyle{unsrturl}
\bibliography{references}

\begin{thebibliography}{10}

\bibitem{DennisKitaevPreskill02}
Eric Dennis, Alexei Kitaev, Andrew Landahl, and John Preskill.
\newblock Topological quantum memory.
\newblock {\em Journal of Mathematical Physics}, 43(9):4452--4505, 2002.
\newblock \href {https://doi.org/10.1063/1.1499754}
  {\path{https://doi.org/10.1063/1.1499754}}.

\bibitem{Katzgraberetal12}
Ruben~S. Andrist, H.~Bombin, Helmut~G. Katzgraber, and M.~A. Martin-Delgado.
\newblock Optimal error correction in topological subsystem codes.
\newblock {\em Physical Review A}, 85:050302, 2012.
\newblock \href {https://doi.org/10.1103/PhysRevA.85.050302}
  {\path{https://doi.org/10.1103/PhysRevA.85.050302}}.

\bibitem{bravyiDisorderAssistedErrorCorrection2012}
Sergey Bravyi and Robert K{\"o}nig.
\newblock Disorder-assisted error correction in {Majorana} chains.
\newblock {\em Communications in Mathematical Physics}, 316(3):641--692, 2012.
\newblock \href {https://doi.org/10.1007/s00220-012-1606-9}
  {\path{https://doi.org/10.1007/s00220-012-1606-9}}.

\bibitem{darmawanpoulin17}
Andrew~S. Darmawan and David Poulin.
\newblock Tensor-network simulations of the surface code under realistic noise.
\newblock {\em Physical Review Letters}, 119:040502, 2017.
\newblock \href {https://doi.org/10.1103/PhysRevLett.119.040502}
  {\path{https://doi.org/10.1103/PhysRevLett.119.040502}}.

\bibitem{bravyienglbrechtetal}
Sergey Bravyi, Matthias Englbrecht, Robert König, and Nolan Peard.
\newblock Correcting coherent errors with surface codes.
\newblock {\em npj Quantum Information}, 4(1):1--6, 2018.
\newblock \href {https://doi.org/10.1038/s41534-018-0106-y}
  {\path{https://doi.org/10.1038/s41534-018-0106-y}}.

\bibitem{tuckettdarmawanetal19}
David~K. Tuckett, Andrew~S. Darmawan, Christopher~T. Chubb, Sergey Bravyi,
  Stephen~D. Bartlett, and Steven~T. Flammia.
\newblock Tailoring surface codes for highly biased noise.
\newblock {\em Physical Review X}, 9:041031, 2019.
\newblock \href {https://doi.org/10.1103/PhysRevX.9.041031}
  {\path{https://doi.org/10.1103/PhysRevX.9.041031}}.

\bibitem{aaronsonImprovedSimulationStabilizer2004}
Scott Aaronson and Daniel Gottesman.
\newblock Improved simulation of stabilizer circuits.
\newblock {\em Physical Review A}, 70(5):052328, 2004.
\newblock \href {https://doi.org/10.1103/PhysRevA.70.052328}
  {\path{https://doi.org/10.1103/PhysRevA.70.052328}}.

\bibitem{10.1145/380752.380785}
Leslie~G. Valiant.
\newblock Quantum computers that can be simulated classically in polynomial
  time.
\newblock In {\em Proceedings of the Thirty-Third Annual ACM Symposium on
  Theory of Computing}, STOC '01, page 114–123, New York, USA, 2001.
\newblock \href {https://doi.org/10.1145/380752.380785}
  {\path{https://doi.org/10.1145/380752.380785}}.

\bibitem{PhysRevA.65.032325}
Barbara~M. Terhal and David~P. DiVincenzo.
\newblock Classical simulation of noninteracting-fermion quantum circuits.
\newblock {\em Physical Review A}, 65:032325, 2002.
\newblock \href {https://doi.org/10.1103/PhysRevA.65.032325}
  {\path{https://doi.org/10.1103/PhysRevA.65.032325}}.

\bibitem{knill2001fermionic}
Emanuel Knill.
\newblock Fermionic linear optics and matchgates.
\newblock Technical Report LAUR-01-4472, Los Alamos National Laboratory, 2001.

\bibitem{PhysRevA.71.022316}
Sergey Bravyi and Alexei Kitaev.
\newblock Universal quantum computation with ideal {C}lifford gates and noisy
  ancillas.
\newblock {\em Physical Review A}, 71:022316, 2005.
\newblock \href {https://doi.org/10.1103/PhysRevA.71.022316}
  {\path{https://doi.org/10.1103/PhysRevA.71.022316}}.

\bibitem{PhysRevA.73.042313}
Sergey Bravyi.
\newblock Universal quantum computation with the $\nu=5/2$ fractional quantum
  {H}all state.
\newblock {\em Physical Review A}, 73:042313, 2006.
\newblock \href {https://doi.org/10.1103/PhysRevA.73.042313}
  {\path{https://doi.org/10.1103/PhysRevA.73.042313}}.

\bibitem{PhysRevLett.123.080503}
Martin Hebenstreit, Richard Jozsa, Barbara Kraus, Sergii Strelchuk, and Mithuna
  Yoganathan.
\newblock All pure fermionic non-{G}aussian states are magic states for
  matchgate computations.
\newblock {\em Physical Review Letters}, 123:080503, 2019.
\newblock \href {https://doi.org/10.1103/PhysRevLett.123.080503}
  {\path{https://doi.org/10.1103/PhysRevLett.123.080503}}.

\bibitem{PhysRevLett.118.090501}
Mark Howard and Earl Campbell.
\newblock Application of a resource theory for magic states to fault-tolerant
  quantum computing.
\newblock {\em Physical Review Letters}, 118:090501, 2017.
\newblock \href {https://doi.org/10.1103/PhysRevLett.118.090501}
  {\path{https://doi.org/10.1103/PhysRevLett.118.090501}}.

\bibitem{Heinrich2019robustnessofmagic}
Markus Heinrich and David Gross.
\newblock Robustness of magic and symmetries of the stabiliser polytope.
\newblock {\em {Quantum}}, 3:132, 2019.
\newblock \href {https://doi.org/10.22331/q-2019-04-08-132}
  {\path{https://doi.org/10.22331/q-2019-04-08-132}}.

\bibitem{bravyiTradingClassicalQuantum2016a}
Sergey Bravyi, Graeme Smith, and John~A. Smolin.
\newblock Trading classical and quantum computational resources.
\newblock {\em Physical Review X}, 6(2):021043, 2016.
\newblock \href {https://doi.org/10.1103/PhysRevX.6.021043}
  {\path{https://doi.org/10.1103/PhysRevX.6.021043}}.

\bibitem{bravyiImprovedClassicalSimulation2016}
Sergey Bravyi and David Gosset.
\newblock Improved classical simulation of quantum circuits dominated by
  {Clifford} gates.
\newblock {\em Physical Review Letters}, 116(25), 2016.
\newblock \href {https://doi.org/10.1103/PhysRevLett.116.250501}
  {\path{https://doi.org/10.1103/PhysRevLett.116.250501}}.

\bibitem{bravyiSimulationQuantumCircuits2019a}
Sergey Bravyi, Dan Browne, Padraic Calpin, Earl Campbell, David Gosset, and
  Mark Howard.
\newblock Simulation of quantum circuits by low-rank stabilizer decompositions.
\newblock {\em Quantum}, 3:181, 2019.
\newblock \href {https://doi.org/10.22331/q-2019-09-02-181}
  {\path{https://doi.org/10.22331/q-2019-09-02-181}}.

\bibitem{heimendahlStabilizerExtentNot2021}
Arne Heimendahl, Felipe {Montealegre-Mora}, Frank Vallentin, and David Gross.
\newblock Stabilizer extent is not multiplicative.
\newblock {\em Quantum}, 5:400, 2021.
\newblock \href {https://doi.org/10.22331/q-2021-02-24-400}
  {\path{https://doi.org/10.22331/q-2021-02-24-400}}.

\bibitem{Beverland_2020}
Michael Beverland, Earl Campbell, Mark Howard, and Vadym Kliuchnikov.
\newblock Lower bounds on the non-{C}lifford resources for quantum
  computations.
\newblock {\em Quantum Science and Technology}, 5(3):035009, 2020.
\newblock \href {https://doi.org/10.1088/2058-9565/ab8963}
  {\path{https://doi.org/10.1088/2058-9565/ab8963}}.

\bibitem{PRXQuantum.2.010345}
James~R. Seddon, Bartosz Regula, Hakop Pashayan, Yingkai Ouyang, and Earl~T.
  Campbell.
\newblock Quantifying quantum speedups: Improved classical simulation from
  tighter magic monotones.
\newblock {\em PRX Quantum}, 2:010345, 2021.
\newblock \href {https://doi.org/10.1103/PRXQuantum.2.010345}
  {\path{https://doi.org/10.1103/PRXQuantum.2.010345}}.

\bibitem{bu2023stabilizer}
Kaifeng Bu, Weichen Gu, and Arthur Jaffe.
\newblock Stabilizer testing and magic entropy.
\newblock arXiv:2306.09292, 2023.
\newblock

\bibitem{troppone}
Joel~A. Tropp.
\newblock Greed is good: algorithmic results for sparse approximation.
\newblock {\em IEEE Transactions on Information Theory}, 50(10):2231--2242,
  2004.
\newblock \href {https://doi.org/10.1109/TIT.2004.834793}
  {\path{https://doi.org/10.1109/TIT.2004.834793}}.

\bibitem{chenetal}
Scott~Shaobing Chen, David~L. Donoho, and Michael~A. Saunders.
\newblock Atomic decomposition by basis pursuit.
\newblock {\em SIAM Review}, 43(1):129--159, 2001.
\newblock \href {https://doi.org/10.1137/S003614450037906X}
  {\path{https://doi.org/10.1137/S003614450037906X}}.

\bibitem{fuchs}
Jean-Jacques Fuchs.
\newblock On sparse representations in arbitrary redundant bases.
\newblock {\em IEEE Transactions on Information Theory}, 50(6):1341--1344,
  2004.
\newblock \href {https://doi.org/10.1109/TIT.2004.828141}
  {\path{https://doi.org/10.1109/TIT.2004.828141}}.

\bibitem{troppRecoveryShortComplex2005}
Joel~A. Tropp.
\newblock Recovery of short, complex linear combinations via /spl lscr//sub 1/
  minimization.
\newblock {\em IEEE Transactions on Information Theory}, 51(4):1568--1570,
  2005.
\newblock \href {https://doi.org/10.1109/TIT.2005.844057}
  {\path{https://doi.org/10.1109/TIT.2005.844057}}.

\bibitem{PRXQuantum.3.020361}
Hakop Pashayan, Oliver Reardon-Smith, Kamil Korzekwa, and Stephen~D. Bartlett.
\newblock Fast estimation of outcome probabilities for quantum circuits.
\newblock {\em PRX Quantum}, 3:020361, 2022.
\newblock \href {https://doi.org/10.1103/PRXQuantum.3.020361}
  {\path{https://doi.org/10.1103/PRXQuantum.3.020361}}.

\bibitem{alizadehSecondorderConeProgramming2003}
Farid Alizadeh and Donald Goldfarb.
\newblock Second-order cone programming.
\newblock {\em Mathematical Programming}, 95(1):3--51, 2003.
\newblock \href {https://doi.org/10.1007/s10107-002-0339-5}
  {\path{https://doi.org/10.1007/s10107-002-0339-5}}.

\bibitem{boyd2004convex}
Stephen~P. Boyd and Lieven Vandenberghe.
\newblock {\em Convex optimization}.
\newblock Cambridge university press, 2004.
\newblock \href {https://doi.org/10.1017/CBO9780511804441}
  {\path{https://doi.org/10.1017/CBO9780511804441}}.

\bibitem{Chandrasekaranetal}
Venkat Chandrasekaran, Benjamin Recht, Pablo~A. Parrilo, and Alan~S. Willsky.
\newblock The {Convex} {Geometry} of {Linear} {Inverse} {Problems}.
\newblock {\em Foundations of Computational Mathematics}, 12(6):805--849, 2012.
\newblock \href {https://doi.org/10.1007/s10208-012-9135-7}
  {\path{https://doi.org/10.1007/s10208-012-9135-7}}.

\bibitem{piveteausutter}
Christophe Piveteau and David Sutter.
\newblock Circuit knitting with classical communication.
\newblock {\em IEEE Transactions on Information Theory}, 70(4):2734--2745,
  2024.
\newblock \href {https://doi.org/10.1109/TIT.2023.3310797}
  {\path{https://doi.org/10.1109/TIT.2023.3310797}}.

\bibitem{meloPowerNoisyFermionic2013a}
Fernando de~Melo, Piotr {\'C}wikli{\'n}ski, and Barbara~M. Terhal.
\newblock The power of noisy fermionic quantum computation.
\newblock {\em New Journal of Physics}, 15(1):013015, 2013.
\newblock \href {https://doi.org/10.1088/1367-2630/15/1/013015}
  {\path{https://doi.org/10.1088/1367-2630/15/1/013015}}.

\bibitem{jozsamiyake08}
Richard Jozsa and Akimasa Miyake.
\newblock Matchgates and classical simulation of quantum circuits.
\newblock {\em Proceedings of the Royal Society A: Mathematical, Physical and
  Engineering Sciences}, 464(2100):3089--3106, 2008.
\newblock \href {https://doi.org/10.1098/rspa.2008.0189}
  {\path{https://doi.org/10.1098/rspa.2008.0189}}.

\bibitem{bartlettsanders02b}
Stephen~D. Bartlett and Barry~C. Sanders.
\newblock Efficient classical simulation of optical quantum information
  circuits.
\newblock {\em Physical Review Letters}, 89:207903, 2002.
\newblock \href {https://doi.org/10.1103/PhysRevLett.89.207903}
  {\path{https://doi.org/10.1103/PhysRevLett.89.207903}}.

\bibitem{Bartlettetal02}
Stephen~D. Bartlett, Barry~C. Sanders, Samuel~L. Braunstein, and Kae Nemoto.
\newblock Efficient classical simulation of continuous variable quantum
  information processes.
\newblock {\em Physical Review Letters}, 88:097904, 2002.
\newblock \href {https://doi.org/10.1103/PhysRevLett.88.097904}
  {\path{https://doi.org/10.1103/PhysRevLett.88.097904}}.

\bibitem{sommaetalbarnumknill}
Rolando Somma, Howard Barnum, Gerardo Ortiz, and Emanuel Knill.
\newblock Efficient solvability of {H}amiltonians and limits on the power of
  some quantum computational models.
\newblock {\em Physical Review Letters}, 97:190501, 2006.
\newblock \href {https://doi.org/10.1103/PhysRevLett.97.190501}
  {\path{https://doi.org/10.1103/PhysRevLett.97.190501}}.

\bibitem{gottesman1997stabilizer}
Daniel Gottesman.
\newblock {\em Stabilizer Codes and Quantum Error Correction}.
\newblock {PhD} thesis, California Institute of Technology, 1997.
\newblock \href {https://doi.org/10.7907/rzr7-dt72}
  {\path{https://doi.org/10.7907/rzr7-dt72}}.

\bibitem{VeitchFerriGrossEmerson}
Victor Veitch, Christopher Ferrie, David Gross, and Joseph Emerson.
\newblock Negative quasi-probability as a resource for quantum computation.
\newblock {\em New Journal of Physics}, 14(11):113011, 2012.
\newblock \href {https://doi.org/10.1088/1367-2630/14/11/113011}
  {\path{https://doi.org/10.1088/1367-2630/14/11/113011}}.

\bibitem{MariEisert12}
Andrea Mari and Jens Eisert.
\newblock Positive wigner functions render classical simulation of quantum
  computation efficient.
\newblock {\em Physical Review Letters}, 109:230503, 2012.
\newblock \href {https://doi.org/10.1103/PhysRevLett.109.230503}
  {\path{https://doi.org/10.1103/PhysRevLett.109.230503}}.

\bibitem{PhysRevLett.115.070501}
Hakop Pashayan, Joel~J. Wallman, and Stephen~D. Bartlett.
\newblock Estimating outcome probabilities of quantum circuits using
  quasiprobabilities.
\newblock {\em Physical Review Letters}, 115:070501, 2015.
\newblock \href {https://doi.org/10.1103/PhysRevLett.115.070501}
  {\path{https://doi.org/10.1103/PhysRevLett.115.070501}}.

\bibitem{PhysRevA.101.012350}
Robert Raussendorf, Juani Bermejo-Vega, Emily Tyhurst, Cihan Okay, and Michael
  Zurel.
\newblock Phase-space-simulation method for quantum computation with magic
  states on qubits.
\newblock {\em Physical Review A}, 101:012350, 2020.
\newblock \href {https://doi.org/10.1103/PhysRevA.101.012350}
  {\path{https://doi.org/10.1103/PhysRevA.101.012350}}.

\bibitem{RaussendorfBrownOkayBermejo17}
Robert Raussendorf, Dan~E. Browne, Nicolas Delfosse, Cihan Okay, and Juan
  Bermejo-Vega.
\newblock Contextuality and wigner-function negativity in qubit quantum
  computation.
\newblock {\em Physical Review A}, 95:052334, 2017.
\newblock \href {https://doi.org/10.1103/PhysRevA.95.052334}
  {\path{https://doi.org/10.1103/PhysRevA.95.052334}}.

\bibitem{PhysRevLett.119.120505}
Juan Bermejo-Vega, Nicolas Delfosse, Dan~E. Browne, Cihan Okay, and Robert
  Raussendorf.
\newblock Contextuality as a resource for models of quantum computation with
  qubits.
\newblock {\em Physical Review Letters}, 119:120505, 2017.
\newblock \href {https://doi.org/10.1103/PhysRevLett.119.120505}
  {\path{https://doi.org/10.1103/PhysRevLett.119.120505}}.

\bibitem{Frembs_2018}
Markus Frembs, Sam Roberts, and Stephen~D. Bartlett.
\newblock Contextuality as a resource for measurement-based quantum computation
  beyond qubits.
\newblock {\em New Journal of Physics}, 20(10):103011, 2018.
\newblock \href {https://doi.org/10.1088/1367-2630/aae3ad}
  {\path{https://doi.org/10.1088/1367-2630/aae3ad}}.

\bibitem{VershyninaPhysRevA.90.062329}
Anna Vershynina.
\newblock Complete criterion for convex-gaussian-state detection.
\newblock {\em Physical Review A}, 90:062329, 2014.
\newblock \href {https://doi.org/10.1103/PhysRevA.90.062329}
  {\path{https://doi.org/10.1103/PhysRevA.90.062329}}.

\bibitem{PhysRevResearch.4.043100}
Shigeo Hakkaku, Yuichiro Tashima, Kosuke Mitarai, Wataru Mizukami, and Keisuke
  Fujii.
\newblock Quantifying fermionic nonlinearity of quantum circuits.
\newblock {\em Physical Review Research}, 4:043100, 2022.
\newblock \href {https://doi.org/10.1103/PhysRevResearch.4.043100}
  {\path{https://doi.org/10.1103/PhysRevResearch.4.043100}}.

\bibitem{mocherla2023extending}
Avinash Mocherla, Lingling Lao, and Dan~E. Browne.
\newblock Extending matchgate simulation methods to universal quantum circuits.
\newblock arxiv:2302.02654, 2023.
\newblock

\bibitem{reardon-smithImprovedClassicalSimulation2022}
Oliver {Reardon-Smith}, Micha{\l} Oszmaniec, and Kamil Korzekwa.
\newblock Improved classical simulation of quantum circuits dominated by
  fermionic linear optical gates.
\newblock arXiv:2307.12702, 2023.

\bibitem{StrelchukCudby}
Joshua Cudby and Sergii Strelchuk.
\newblock {G}aussian decomposition of magic states for matchgate computations.
\newblock arXiv:2307.12654, 2023.

\bibitem{BOTERO200439}
Alonso Botero and Benni Reznik.
\newblock {BCS}-like modewise entanglement of fermion {G}aussian states.
\newblock {\em Physics Letters A}, 331(1):39--44, 2004.
\newblock \href {https://doi.org/10.1016/j.physleta.2004.08.037}
  {\path{https://doi.org/10.1016/j.physleta.2004.08.037}}.

\bibitem{amosovfillippov}
Grigori~G. Amosov and Sergey~N. Filippov.
\newblock Spectral properties of reduced fermionic density operators and parity
  superselection rule.
\newblock {\em Quantum Information Processing}, 16(1):2, 2016.
\newblock \href {https://doi.org/10.1007/s11128-016-1467-9}
  {\path{https://doi.org/10.1007/s11128-016-1467-9}}.

\bibitem{doi:10.1137/0106004}
Wallace Givens.
\newblock Computation of plain unitary rotations transforming a general matrix
  to triangular form.
\newblock {\em Journal of the Society for Industrial and Applied Mathematics},
  6(1):26--50, 1958.
\newblock \href {https://doi.org/10.1137/0106004}
  {\path{https://doi.org/10.1137/0106004}}.

\bibitem{10.5555/248979}
Gene~H. Golub and Charles~F. Van~Loan.
\newblock {\em Matrix Computations (3rd Ed.)}.
\newblock Johns Hopkins University Press, USA, 1996.
\newblock \href {https://doi.org/10.56021/9781421407944}
  {\path{https://doi.org/10.56021/9781421407944}}.

\bibitem{bravyiLagrangianRepresentationFermionic2004}
Sergey Bravyi.
\newblock Lagrangian representation for fermionic linear optics.
\newblock {\em Quantum Information \& Computation}, 5(3):216–238, 2005.
\newblock \href {https://doi.org/10.26421/qic5.3-3}
  {\path{https://doi.org/10.26421/qic5.3-3}}.

\bibitem{lichtenstein}
Woody Lichtenstein.
\newblock A system of quadrics describing the orbit of the highest weight
  vector.
\newblock {\em Proceedings of the American Mathematical Society},
  84(4):605--608, 1982.
\newblock \href {https://doi.org/10.1090/s0002-9939-1982-0643758-8}
  {\path{https://doi.org/10.1090/s0002-9939-1982-0643758-8}}.

\bibitem{kusClassicalQuantumStates2009}
Marek Ku{\'s} and Ingemar Bengtsson.
\newblock ``{C}lassical'' quantum states.
\newblock {\em Physical Review A}, 80(2):022319, 2009.
\newblock \href {https://doi.org/10.1103/PhysRevA.80.022319}
  {\path{https://doi.org/10.1103/PhysRevA.80.022319}}.

\bibitem{oszmanieckus}
Michał Oszmaniec and Marek Kuś.
\newblock On detection of quasiclassical states.
\newblock {\em Journal of Physics A: Mathematical and Theoretical},
  45(24):244034, 2012.
\newblock \href {https://doi.org/10.1088/1751-8113/45/24/244034}
  {\path{https://doi.org/10.1088/1751-8113/45/24/244034}}.

\bibitem{oszmaniecClassicalSimulationFermionic2014a}
Micha{\l} Oszmaniec, Jan Gutt, and Marek Ku{\'s}.
\newblock Classical simulation of fermionic linear optics augmented with noisy
  ancillas.
\newblock {\em Physical Review A}, 90(2):020302, 2014.
\newblock \href {https://doi.org/10.1103/PhysRevA.90.020302}
  {\path{https://doi.org/10.1103/PhysRevA.90.020302}}.

\bibitem{Loewdin55}
Per-Olov L\"owdin.
\newblock Quantum theory of many-particle systems. i. physical interpretations
  by means of density matrices, natural spin-orbitals, and convergence problems
  in the method of configurational interaction.
\newblock {\em Physical Review}, 97:1474--1489, 1955.
\newblock \href {https://doi.org/10.1103/PhysRev.97.1474}
  {\path{https://doi.org/10.1103/PhysRev.97.1474}}.

\bibitem{bravyiComplexityQuantumImpurity2017a}
Sergey Bravyi and David Gosset.
\newblock Complexity of quantum impurity problems.
\newblock {\em Communications in Mathematical Physics}, 356(2), 2017.
\newblock \href {https://doi.org/10.1007/s00220-017-2976-9}
  {\path{https://doi.org/10.1007/s00220-017-2976-9}}.

\bibitem{arakimoriya}
Huzihiro Araki and Hajime Moriya.
\newblock Joint extension of states of subsystems for a {{CAR}} system.
\newblock {\em Communications in Mathematical Physics}, 237(1):105--122, 2003.
\newblock \href {https://doi.org/10.1007/s00220-003-0832-6}
  {\path{https://doi.org/10.1007/s00220-003-0832-6}}.

\bibitem{Hayes}
Thomas~P Hayes.
\newblock A large-deviation inequality for vector-valued martingales, 2005.
\newblock {\em Unpublished manuscript}.
\newblock https://www.cs.unm.edu/$\sim$hayes/papers/VectorAzuma/.

\bibitem{tangSparseRecoveryContinuous2013}
Gongguo Tang, Badri~Narayan Bhaskar, and Benjamin Recht.
\newblock Sparse recovery over continuous dictionaries-just discretize.
\newblock In {\em 2013 {{Asilomar Conference}} on {{Signals}}, {{Systems}} and
  {{Computers}}}, pages 1043--1047, 2013.
\newblock \href {https://doi.org/10.1109/ACSSC.2013.6810450}
  {\path{https://doi.org/10.1109/ACSSC.2013.6810450}}.

\end{thebibliography}

\end{document}